\documentclass[11pt]{article}

\usepackage{fullpage}
\usepackage{amsmath,amsfonts,amsthm,mathrsfs,mathpazo,xspace,hyperref,graphicx}
\usepackage{endnotes}
\usepackage{color}
\usepackage{bbm}
\usepackage{times}
\usepackage{amssymb,latexsym}
\usepackage{enumitem}

\newtheorem{theorem}{Theorem}[section]
\newtheorem{proposition}[theorem]{Proposition}

\newtheorem{lemma}[theorem]{Lemma}
\newtheorem{claim}[theorem]{Claim}

\newtheorem{corollary}[theorem]{Corollary}

\theoremstyle{remark}
\newtheorem{remark}[theorem]{Remark}

\theoremstyle{definition}
\newtheorem{definition}[theorem]{Definition}

\newcommand{\beq}{\begin{eqnarray}}
\newcommand{\eeq}{\end{eqnarray}}

\newcommand{\ket}[1]{|#1\rangle}
\newcommand{\bra}[1]{\langle#1|}
\newcommand{\proj}[1]{\ket{#1}\!\bra{#1}}
\newcommand{\Tr}{\mbox{\rm Tr}}
\newcommand{\Id}{\ensuremath{\mathop{\rm Id}\nolimits}}
\newcommand{\Es}[1]{\textsc{E}_{#1}}

\newcommand{\reg}[1]{{\textsf{#1}}}
\newcommand{\ol}[1]{\overline{#1}}

\newcommand{\C}{\ensuremath{\mathbb{C}}}
\newcommand{\N}{\ensuremath{\mathbb{N}}}
\newcommand{\bbN}{\ensuremath{\mathbb{N}}}

\newcommand{\R}{\ensuremath{\mathbb{R}}}
\newcommand{\Z}{\ensuremath{\mathbb{Z}}}

\newcommand{\mB}{\ensuremath{\mathcal{B}}}

\newcommand{\mD}{\ensuremath{\mathcal{D}}}
\newcommand{\mF}{\ensuremath{\mathcal{F}}}
\newcommand{\mK}{\ensuremath{\mathcal{K}}}

\newcommand{\mX}{\ensuremath{\mathcal{X}}}
\newcommand{\mY}{\ensuremath{\mathcal{Y}}}

\newcommand{\Inv}{\ensuremath{\textsc{Inv}}}
\newcommand{\GEN}{\ensuremath{\textsc{GEN}}}

\newcommand{\mH}{\mathcal{H}}

\newcommand{\setft}[1]{\mathrm{#1}}
\newcommand{\Density}{\setft{D}}
\newcommand{\Pos}{\setft{Pos}}

\DeclareMathOperator{\poly}{poly}
\DeclareMathOperator{\negl}{negl}
\newcommand{\dset}{G}

\newcommand{\supp}{\textsc{Supp}}
\newcommand{\Gen}{\textsc{Gen}}
\newcommand{\GenTrap}{\textsc{GenTrap}}
\newcommand{\Invert}{\textsc{Invert}}
\newcommand{\lossy}{\textsc{lossy}}

\newcommand{\eps}{\varepsilon}

\newcommand{\Acc}{\textsc{Acc}}

\newcommand{\inj}{J}
\newcommand{\mZ}{\mathbbm{Z}}
\newcommand{\mN}{\mathbbm{N}}

\newcommand{\sX}{\mathcal{X}}
\newcommand{\sY}{\mathcal{Y}}
\newcommand{\sR}{\mathcal{R}}

\newcommand{\trnq}[1]{\left[ {#1} \right]_q}

\newcommand{\lwe}{\mathrm{LWE}}

\newcommand{\bbZ}{\mathbb{Z}}

\newcommand{\vc}[1]{\mathbf{{#1}}}
\newcommand{\abs}[1]{\left\vert {#1} \right\vert}

\newcommand{\sivp}{\mathrm{SIVP}}
\newcommand{\otild}{{\widetilde{O}}}

\def\*#1{\mathbf{#1}}

\newcommand{\Hmin}{H_\infty}

\DeclareMathOperator{\arcsinh}{arcsinh}

\newif\ifnotes\notesfalse

\ifnotes
\usepackage{color}
\definecolor{mygrey}{gray}{0.50}
\newcommand{\notename}[2]{{\textcolor{mygrey}{\footnotesize{\bf (#1:} {#2}{\bf ) }}}}
\newcommand{\noteswarning}{{\begin{center} {\Large WARNING: NOTES ON}\endnote{Warning: notes on}\end{center}}}
\newcommand{\notesendofpaper}{{\theendnotes}}
\newcommand{\pnote}[1]{{\endnote{#1}}}

\newcommand{\authnote}[3]{\textcolor{#3}{\small {\textbf{[ {#1}:} #2 \textbf{]
}}}}

\else

\newcommand{\notename}[2]{{}}
\newcommand{\noteswarning}{{}}
\newcommand{\notesendofpaper}{}
\newcommand{\pnote}[1]{}

\newcommand{\authnote}[3]{}

\fi

\newcommand{\unote}[1]{\authnote{Urmila}{#1}{blue}}

\begin{document}

\title{A Cryptographic Test of Quantumness and Certifiable Randomness from a Single Quantum Device}
\author{Zvika Brakerski\thanks{Weizmann Institute of Science, Israel. Email: \texttt{zvika.brakerski@weizmann.ac.il}.} \and Paul Christiano\thanks{OpenAI, USA. Work performed while at UC Berkeley} \and Urmila Mahadev\thanks{UC Berkeley, USA. Email: \texttt{mahadev@berkeley.edu}} \and Umesh Vazirani\thanks{UC Berkeley, USA. Email: \texttt{vazirani@cs.berkeley.edu}} \and Thomas Vidick\thanks{California Institute of Technology, USA. Email: \texttt{vidick@cms.caltech.edu}}}
\date{}
\maketitle

\noteswarning

\begin{abstract}

We consider a new model for the testing of untrusted quantum devices, consisting of a single polynomial time bounded quantum device interacting with a classical polynomial time verifier. In this model we propose solutions to two tasks --- a protocol for efficient classical verification that the untrusted device is ``truly quantum," and a protocol for producing certifiable randomness from a single untrusted quantum device. Our solution relies on the existence of a new cryptographic primitive for constraining the power of an untrusted quantum device : post-quantum secure trapdoor claw-free functions which must satisfy an adaptive hardcore bit property. We show how to construct this primitive based on the hardness of the learning with errors (LWE) problem.


\end{abstract}

\newpage

\tableofcontents

\newpage
\section{Introduction}

The testing of quantum devices, besides being a pressing practical challenge, touches on foundational questions in quantum computational complexity. The classical verifier of such a device is necessarily at a disadvantage due to the exponential power of quantum systems, and the laws of quantum mechanics severely limit the amount of information that can be accessed in principle. Nevertheless, a sequence of results have shown that it is possible to verify the correctness of untrusted quantum devices (also referred to as provers) in a variety of settings, including certifiable random number generation, quantum key distribution and quantum computation. These results have been established in two models: in the first, the classical verifier is augmented with the ability to prepare a sequence of quantum states on small numbers of qubits and transmit them to the quantum device~\cite{abe2008,broadbent2008ubq,fk2012,abem}, and in the second, the classical verifier interacts with two non-communicating quantum devices that share entanglement~\cite{Colbeck09,ruv2012,deviceindependentqkd}.  

In this paper we consider a new model, in which a purely classical verifier interacts with a single, polynomial time bounded quantum machine. The restriction to an efficient quantum device allows the verifier to leverage post-quantum cryptography, i.e.\ cryptographic primitives that can be implemented efficiently on a classical computer but that cannot be broken by any efficient quantum computer. 

In this model we propose solutions to two basic tasks: how to efficiently verify that an untrusted device is ``truly" quantum, and how to generate certifiably random strings from a single untrusted quantum device. The first task is also referred to as "quantum supremacy," and existing protocols for this~\cite{aaronsonboson, boixo2016characterizing, aaronson2016complexity, bfnv19, supremacy19} rely on exponential time classical verification using a classical supercomputer. 
By contrast, our qubit certification test below provides a proof of quantumness that can be verified by a classical verifier in polynomial time. 
There has also been considerable research into certifiable random number expansion from quantum devices~\cite{Colbeck09,Pironio,VV12,miller2016robust,arnon2018practical}, including experimental demonstrations~\cite{Pironio,bierhorst2018experimentally}. However, all prior works have focused on the setting where there are multiple quantum devices that share entanglement, and where the randomness certification relies on the violation of a Bell inequality. 


The core of the difficulty in interacting with untrusted quantum devices lies in enforcing a qubit structure in the device's operations, 
i.e.\ that the quantum device actually holds qubits, and is performing measurements on them to respond to the verifier's queries. In the two previous models of testing quantum devices, this issue was handled in two different ways. For slightly quantum verifiers, the verifier could simply send qubits to the prover, encoded in such a way that the prover was forced to work with only those qubits. In the model of two entangled quantum devices, Bell inequality violations were used to prove that the two devices must share Bell states and measure them in the $X$ and $Z$ basis as requested.  
Our work relies on post-quantum cryptography to enforce qubits. Roughly, our qubit certification protocol enables the quantum device (the prover) to create a qubit in the state $\frac{1}{\sqrt{2}} \ket{0} + (-1)^b\frac{1}{\sqrt{2}} \ket{1}$, where the bit $b$ is computationally hidden from the prover. With knowledge of the trapdoor for the post-quantum cryptosystem, the classical verifier can compute $b$ and use it to verify that the prover actually holds the above state, thereby gaining leverage over the quantum prover. 


Our certifiable randomness protocol uses the qubit certificaton protocol as a subroutine, and provides an information-theoretic guarantee about the random string output by the untrusted quantum device. The guarantee is stronger than computational pseudorandomness, which is easily achievable under standard cryptographic assumptions, since the verifier starts with a short uniformly random seed. It is illuminating to understand how an information-theoretic guarantee could even be connected to the the computational assumptions about the device. 
We imagine that there is an adversary with unbounded computing power and an unboundedly large quantum register $\reg{E}$, which may be entangled with the quantum device register $\reg{D}$. The guarantee can now be expressed as saying that the unbounded adversary, who is allowed to design the quantum device and to perform an arbitrary measurement on the register $\reg{E}$, cannot distinguish the output of the protocol from a uniform sequence of bits, provided the device is unable to break the post-quantum cryptography during the execution of the protocol.

\paragraph{A qubit certification test}

The specific cryptographic primitive we rely on is a post-quantum secure trapdoor claw-free  (in short, TCF) family of function pairs $f_0,f_1:\{0,1\}^n \rightarrow \{0,1\}^m$, the post-quantum analogue of a notion introduced by Goldwasser, Micali and Rivest in the context of digital signatures \cite{GoldwasserMR84}. A TCF pair is a pair of functions which are injective, with the same image, and satisfy the following property. With knowledge of a secret trapdoor it is possible to efficiently (classically) compute the two preimages $x_0$ and $x_1$ of a given $y$ ($f_0(x_0) = f_1(x_1) = y$), but without the trapdoor, there is no efficient quantum algorithm that can compute such a triple $(x_0,x_1,y)$, referred to as a \textit{claw}, for any $y$. 

While the quantum device cannot compute a claw, nevertheless
it can simultaneously hold an image $y$ as well as a superposition \begin{equation}\frac{1}{\sqrt{2}}(\ket{0}\ket{x_0} +\ket{1}\ket{x_1})\label{eq:introclawsuperposition}\end{equation} over the two preimages of $y$, simply by evaluating $f$ on a uniform superposition over all inputs and measuring the image $y$. If the quantum device were to measure the above state in the standard basis, it would obtain a random preimage, $x_0$ or $x_1$. This is not particularly interesting since a classical machine could sample from the same distribution by first sampling a random bit $b$ and string $x$ and then computing $y=f_b(x)$. To take advantage of the fact that the preimages $x_0,x_1$ are stored in superposition the quantum device can instead perform a Fourier (Hadamard) basis measurement on all but the first qubit of the state, yielding a string $d\in\{0,1\}^n$. At this point we are back to the state mentioned earlier; the quantum device currently holds, for $c = d\cdot (x_0\oplus x_1)$,
\begin{equation}
    \frac{1}{\sqrt{2}}(\ket{0} + (-1)^c\ket{1})\;.
\end{equation}
A Fourier measurement of the single qubit state above will yield the bit $c$. 

This Fourier measurement is the aspect that separates quantum and classical devices; intuitively, the output $(d,c)$ should be hard to reproduce in the classical setting, as it is dependent on both elements $x_0$ and $x_1$ in the superposition and the claw-free property implies that it is computationally intractable to hold both $x_0$ and $x_1$ simultaneously. This suggests the qubit certification test, between a classical verifier and a quantum prover, written in Figure~\ref{fig:protocolintro}.

\begin{figure}[htbp]
\rule[1ex]{16.5cm}{0.5pt}\\
\begin{enumerate}
    \item The verifier generates a TCF pair, along with a trapdoor, and sends just the function pair to the prover.
    \item The prover returns an image $y$ of the TCF pair.
    \item The verifier challenges the prover by randomly asking for either a preimage of $y$, or a bit $c$ and and an $n$-bit string $d$ such that $d\cdot(x_0 \oplus x_1) = c$.
    \item The prover measures in standard or Hadamard basis to return the requested output and the verifier checks the validity by using the trapdoor to compute the two preimages $x_0,x_1$ of $y$.
\end{enumerate}
\rule[1ex]{16.5cm}{0.5pt}
\caption{The quantum certification protocol.}
\label{fig:protocolintro}
\end{figure}

The quantum prover can successfully answer either challenge in the qubit certification protocol by measuring the state in \eqref{eq:introclawsuperposition} in the standard or Hadamard basis.
By contrast, we would like to argue that no classical algorithm can succeed at this task. 
This is counter-intuitive, as ultimately our proof must rely on the security of the TCF, which applies equally to classical and quantum attacks. The crux of the proof is that classical computations can be rewound, while quantum measurements cannot be: if a classical device can pass either challenge, then the device can be rewound to hold \textit{both} a valid equation and a preimage, and we will show that knowledge of both is sufficient to break the TCF. Since quantum measurements cannot be rewound this argument does not apply to quantum machines; if a quantum machine passes the preimage test, it cannot then be used to pass an equation test (and vice versa), since the measurement would cause its state to collapse. 

Showing that knowledge of both a preimage and an equation is sufficient to break the TCF presents a new challenge. Specifically we wish to claim that no efficient (classical or quantum) algorithm can produce both a preimage $x_b$, as well as an $n$-string $d$ and a bit $c$ such that $c = d\cdot (x_0\oplus x_1)$ (even with probability $1/2 + \epsilon$). This may be thought of as a hardcore bit property for the TCF, for the bit of  the $n$-bit string $(x_0\oplus x_1)$ specified by $d$. 
The difficulty is that the specification of the hardcore bit $d$ can be chosen by the quantum device after it gets to see the particular TCF chosen by the verifier, as well as the image $y$. In this sense, what is required is establishing that the TCF has a kind of ``adaptive hardcore bit property.'' We describe this property next.

\paragraph{The adaptive hardcore bit property}

The adaptive hardcore bit described above is a crucial ingredient in classically testing quantum computers, yet it has not been studied in classical cryptography. Luckily, it turns out that it can be built by relying on structural properties of the well-studied learning with errors (LWE) assumption; more specifically, it relies on a property called \textit{leakage resilience}. In this section we give an overview of the ideas required to prove the adaptive hardcore bit. We begin by describing how the learning with errors assumption can be used to construct a trapdoor claw-free function pair, and then describe how the leakage resilience properties of LWE imply the validity of the adaptive hardcore bit for this construction.

Recall that the learning with errors problem starts with a system of $m'$ linear equations modulo $q$ on a set of $n'$ variables, with $m' > n'$. Starting with a uniformly random matrix $A\in \mZ_q^{m'\times n'}$, and a vector $s\in\mZ_q^{n'}$ and letting $t = As$ results in an easily solvable linear system of equations $(A,t)$.
To make the inversion problem challenging, a noise vector $e\in \mZ_q^{m'}$ is added, so instead $t = As + e$. The distribution over the noise vector $e$ is judiciously chosen (from a suitable Gaussian distribution) so that while $s$ is uniquely determined by $(A, t)$, it is computationally difficult to recover it. 
The learning with errors assumption states that the distribution over $(A, t)$ is computationally indistinguishable from the distribution over $(A, u)$, for a uniformly random string $u\in \mZ_q^{m'}$; in other words, 
the addition of the noise $e$ computationally hides $s$.

Given an LWE sample $(A, t = As + e)$, it is natural to try to define a TCF family by letting $f_0(x) = Ax + e_0$ and $f_1(x) = Ax + e_0 + t$. Note that the output of each function is now a random sample from a distribution, since $e_0$ is randomly chosen. Substituting $t = As + e$, we see that $f_1(x) = A(x+s) + e_0 + e$. If $e$ were $0$, this would mean that $f_1(x) = f_0(x+s)$ (i.e. the two distributions are the same). By sampling $e_0$ from a Gaussian much wider than $e$, we can ensure that the distributions $f_1(x)$ and $f_0(x+s)$ are statistically close, thus effectively ensuring that $f_1(x) = f_0(x+s)$. We refer to such a function pair as a noisy trapdoor claw-free function pair (NTCF). Each claw of such a function pair will now have the following property: for all claws $(x_0,x_1,y)$ of the function pair, $x_1 = x_0 - s$. Note that the claw-free property of this pair of functions follows immediately from the LWE assumption since knowledge of both $x_0$ and $x_1$ reveals the secret $s$. 

A quantum device can use an NTCF to set up a superposition over a claw: $\frac{1}{\sqrt{2}}(\ket{0,x_0} + \ket{1,x_0 -s})$. This follows easily by observing that it can create the superposition $\sum_b \sum_x \sum_{e_0} \ket{b}\ket{x}\ket{Ax + e_0 + bt}$ (omitting normalization factors), and measure the last register to obtain $y$, creating the desired superposition in the first two registers. Recall that in our earlier description the quantum device worked over qubits, whereas we worked modulo $q$ while defining the NTCF. This is easily remedied by converting all mod $q$ entries to binary strings -- letting $n = n' {\lceil \log q \rceil}$ and $m = m' {\lceil \log q \rceil}$, we may think of $f_0,f_1:\{0,1\}^n \rightarrow \{0,1\}^m$.
It might be tempting, given the form of the superposition over the claw $\frac{1}{\sqrt{2}}(\ket{0,x_0} + \ket{1,x_0 -s})$, to try to apply standard period-finding quantum algorithms to compute $s$. Of course this does not work, since even though $x_0$ and $x_0 - s$ are stored in binary, $x_0 - s$ is computed modulo $q$ and is incompatible with Fourier sampling performed modulo $2$.
As we will see shortly, this mixing of $\mZ_q$ vectors with the Fourier transform mod 2 is what makes the proof of the adaptive hardcore bit possible. 

Given this additional structure, we can now state the adaptive hardcore bit property a bit more precisely. The adaptive hardcore bit property states that it is difficult to hold both a single preimage $x_b$, as well as a string $d \in \{0,1\}^n \setminus 0^n$ and a bit $c$ such that $c = d\cdot (x_0\oplus x_1)$. 
Since $x_1 = x_0 - s$, one might hope to express the last condition directly in terms of $s$. Note that 
$d\cdot (x_0\oplus x_1)$ is \textit{not} equal to $d\cdot s$, due to the fact that a binary XOR and a difference modulo $q$ do not cancel. 
Instead, we restrict to the case when $s$ is binary, and write it as a string in $\{0,1\}^{n'}$. It turns out that it is possible to use $x_b$ to efficiently compute a string $d'\in\{0,1\}^{n'}$ such that $d\cdot (x_0\oplus (x_1)) = d'\cdot s$, via a  
linear map which relies only on the fact that $s$ is binary.
To see how this map works, consider the special case in which $n' = 1$. Now, $s$ is a single bit, and $x_0, x_1=x_0-s \in  \mZ_q$. Let $x'_0 = x_0 -1$. As before we think of $x_0, x_1, x'_0$ as $\lceil \log q \rceil$ bit numbers. Then if we let $d' = d\cdot (x_0\oplus x'_0)$ it follows that $d\cdot (x_0\oplus (x_0 - s)) = d'\cdot s$. This reasoning immediately extends to the general case. It follows that the adaptive hardcore bit can be reformulated as stating
that it is difficult to produce a string $d'$ and a bit $c$ such that $d'\cdot s = c$; in other words, not only is it hard to find the secret $s$, it is even hard to find any bit of the secret $s$. 

If we could assume that $d'$ was chosen independently of the LWE sample, the desired hardcore bit property would follow from current results in classical cryptography which strengthen the security of LWE to prove leakage resilience: given an LWE sample $As + e$, any given bit of $s$ is computationally indistinguishable from a uniformly random bit. Unfortunately, there is an added difficulty in our setting: the quantum device can choose the string $d'$ after seeing the LWE sample (after all, the device requires the LWE sample $t = As + e$ to evaluate the function). It is in this sense that the hardcore bit property is adaptive. We now outline the leakage resilience argument in order to describe how it can be adapted to our setting. 

In proving leakage resilience~\cite{GKPV10}, the matrix $A$ is replaced with a computationally indistinguishable matrix $BC + E$, where $C\in\mZ_q^{\ell'\times n'}$, for $\ell' \ll n'$. The computational indistinguishability is immediately implied by treating $BC + E$ as a smaller LWE sample, in which $C$ is the secret. Moreover, $E$ is chosen from a Gaussian with width sufficiently smaller than the Gaussian noise $e$, implying that $(BC + E)s + e$ is statistically close to $BCs + e$. The point of this substitution is that the matrix $C$ compresses $s$, and the leftover hash lemma can be invoked to argue that even given $Cs$ (which is at least as much information as $(BC + E)s$), any bit of $s$ is statistically close to uniform, thus showing that it is a hardcore bit given $A$. 

In the situation we are interested in, the choice of $d'$ may depend upon the LWE sample $As + e$, which corresponds in the leakage resilience argument to $d'$ depending on $Cs$. We wish to argue that $d' \cdot s$ is still statistically close to uniform. This is where the mod $q$ versus mod 2 difference comes into play: in our setting, the string $d'$ is binary (as is the inner product of $d'$ and $s$), whereas the entries of $C$ are uniformly random entries in $\mZ_q$. It can be shown via a Fourier analytic argument that even if $Cs$ is fixed, 
there is enough entropy left in $s$ that $d'\cdot s$ is statistically close to uniform.

As you might expect, the argument outlined above requires that $d'$ is non-zero. Therefore, the verifier must check that the string $d$ returned by the prover yields a non-zero $d'$. Note that although $d'$ is an easily computable function of $x_b$ and $d$, the verifier has no way of knowing which preimage ($x_0$ or $x_1$) the adversary may have in mind. It follows that there are two different values $d'_0, d'_1$ (corresponding to $x_0$ and $x_1$) which \textit{both} must  be checked to be non-zero. Of course, the verifier knows the trapdoor and can perform this check efficiently. 

The problem is a little more serious in the adaptive hardcore bit proof, which requires that the validity of $d'$ be testable efficiently without the trapdoor, in order to maintain the entropy of $s$. Clearly, knowledge of both preimages $x_0$ and $x_1$ is not the answer, since this uniquely determines $s$. Instead, we modify the protocol so that the verifier imposes a more restrictive constraint on $d$, by only accepting $d$ such that the first half of $d'_0$ is non-zero, and the second half of $d'_1$ is non-zero. Observe that checking whether $d'$ satisfies these constraints can either be done with $x_0$ and the second half of $s$ (combined with $x_0$, this can be used to compute the second half of $x_1$, and therefore the second half of $d'_1$), or $x_1$ and the first half of $s$. Moreover, observe that with this limited knowledge, the adaptive hardcore bit still holds: for example, knowledge of only $x_0$ and the second half of $s$ preserves the entropy of the first half of $s$, thereby allowing us to apply the hardcore bit argument to the first half of $d'_0$ (which we know to be non-zero).

\paragraph{Quantum Supremacy.}
The qubit certification protocol described above has implications for an important milestone in the experimental realization of quantum computers, namely ``quantum supremacy'': a proof that an (untrusted) quantum computing device performs some computational task that cannot be solved classically without impractical resources. While this could in principle be achieved by demonstrating quantum factoring, the latter requires quantum resources well beyond the capability of near term experiments. Instead current proposals are based on sampling problems (see e.g.~\cite{harrow2017quantum} for a recent survey). The major challenge for these proposals is verifying that the quantum computer did indeed sample from the desired probability distribution, and all existing proposals rely on exponential time classical algorithms for verification. By contrast, our supremacy provides a proof of quantumness that can be verified by a classical verifier in polynomial time. This proposal seems promising from a practical viewpoint --- indeed, even using off-the-shelf bounds for LWE-based cryptography suggests that a protocol providing $50$ bits of security could be implemented with a quantum device of around $2000$ qubits (see e.g.~\cite{lindner2011better}). It would be worth exploring whether there are clever implementations of this scheme that can lead to a protocol in the $200-500$ qubit range. 

Another challenge in making our proposal suitable for near term devices is fault tolerance. While our protocol will require some level of fault-tolerance, the hope is that it might not require general fault-tolerance techniques, due to its robustness: our protocol is robust to a device that only successfully answers the verifier's challenges with a sufficiently large, but constant, success probability. 

\paragraph{Certifiable randomness.}

The challenge in achieving certifiable randomness lies in using computational assumptions to establish not pseudorandomness, but rather that the output of the protocol must be (close to) statistically random. In our analysis we leverage the properties of the NTCF to characterize the quantum state and measurements of the untrusted quantum device --- essentially showing that it must have a qubit initialized in state $\ket{+}$, which it measured in the standard basis, thus generating one bit of statistical randomness. This is the analogue of the use of the violation of Bell inequalities to characterize the state of the device in entanglement-based testing. \unote{}

We first explain how to show that a device that succeeds in the 
qubit certification test (which we will often refer to as a single round test or single round protocol) must generate randomness. In the test the device must make one of two measurements: either a ``preimage'' measurement, or an ``equation'' measurement. We focus on a single bit of information provided by each measurement. The ``preimage'' measurement can be treated as a projection into one of two orthogonal subspaces corresponding to the two preimages $x_0,x_1$ for the element $y$ that the device has returned to the verifier. The ``equation'' measurement can similarly be coarse-grained into a projection on one of two orthogonal subspaces, ``valid'' or ``invalid'', i.e.\ the subspace that corresponds to all measurement outcomes $d,c$ such that $d\cdot(x_0 + x_1) = c$, or the subspace associated with outcomes $d,c$ such that $d\cdot(x_0 + x_1) = c\oplus 1$. 

Applying Jordan's lemma, it is possible to decompose the device's Hilbert space into a direct sum of one- and two-dimensional subspaces, such that within each two-dimensional subspace the ``preimage'' and ``equation'' measurements each correspond to an orthonormal basis, such that the two bases make a certain angle with each other. We argue that almost all angles must be very close to $\pi/4$. Indeed, whenever the angles are \emph{not} near-maximally unbiased, it is possible to show that by considering the effect of performing the measurements in sequence, one can devise an ``attack'' on the NTCF of a kind that contradicts the adaptive hardcore bit property of the NTCF --- informally, the attack can simultaneously produce a valid preimage and a valid equation, with non-negligible advantage. 

As a result it is possible to show that the state and (coarse-grained)
measurements of the device are, up to a global change of basis, close to the following: the device starts with a qubit initialized to $\ket{+}$, which it measures in the standard basis for the case of a preimage test and in the Hadamard basis for the case of an equation test. The fact that an efficient quantum device cannot break the cryptographic assumption has thus been translated into a characterization of the state and actions of the quantum device, which 
further implies that the output of the device in the single round test must contain close to a bit of true (information theoretic) randomness. 

One might further conjecture that for a generic TCF  (e.g.\ modeled as a random oracle), if the output of any efficient quantum device passes the single round test with non-negligible advantage over $\frac{1}{2}$, then the triple $y,d,c$ returned in the equation test must have high min-entropy. Such a strong statement would immediately yield a randomness certification protocol. Among the many difficulties in showing such a statement is that both $y$ and $d$ may be adaptively and adversarially chosen --- in the single round protocol above this issue is addressed by the adaptive hardcore bit property of the NTCF. 

\paragraph{Outline of randomness generation protocol.}
Going beyond the analysis of the single round test requires significantly more work. So far we have argued that if an efficient quantum algorithm has the ability to generate a valid equation with probability sufficiently close to $1$, then, if instead it is asked for a preimage, this preimage must be close to uniformly distributed over the two possibilities. To leverage this our randomness expansion protocol proceeds in multiple rounds, repeatedly asking for new images $y$ and a preimage of $y$ (to generate randomness) while inserting a few randomly located equation tests to test the device. Each time an ``equation'' challenge has been answered, we refresh the pseudorandom keys used for the NTCF. This is required to avoid a simple ``attack'' by the device, which would repeatedly use the same $y$, preimage $x$, and guessed equation $d$ --- succeeding in the protocol with probability $\frac{1}{2}$ without generating any randomness. 

Let's call the sequence of rounds with a particular set of pseudorandom keys an epoch. Intuitively, we would like to claim that if the device passes all the equation tests, then for most epochs and for most rounds within that epoch, the state of the device and its measurements must be (close to) as characterized above: it starts with a qubit initialized to $\ket{+}$, which it measures in the standard basis for the case of a preimage test, and in the Hadamard basis for the case of an equation test. To show this we would like to claim that if the device passes all the equation tests, for most such tests it must produce a valid equation with probability close to $1$. Since each equation test occurs at a random round in the epoch, it should follow from the adaptive hardcore bit property that the sequence of bits that the verifier extracts from the device's answers to preimage tests during that epoch must look statistically random. We give a martingale-based argument to formalize this intuition. 

There is however a bigger challenge to analyzing the protocol --- we must show that the sequence that the verifier extracts from the device's answers to preimage tests must look statistically random even to an infinitely powerful quantum adversary, who may share an arbitrary entangled state with the quantum device. If we could assert that each round of the protocol is played with a qubit exactly in state $\ket{+}$, and measured in the standard basis basis for the case of a preimage test, then this would lead to an easy proof that the extracted sequence looks random to the adversary. Unfortunately the characterization of the device's qubits leaves plenty of room for entanglement with the adversary. Showing that such entanglement cannot leak too much information about the device's measurements was the major challenge in previous work on certified randomness through Bell inequality violations~\cite{VV12,miller2014universal,arnon2018practical}. 
Our cryptographic setting presents a new difficulty, which is that in contrast to the  Bell inequality violation scenarios, in our setting it is not \emph{impossible} for a deterministic device to succeed in the test: it is merely \emph{computationally hard} to do so. This prevents us from directly applying the results in~\cite{miller2014universal,arnon2018practical} and requires us to suitably modify their framework. We describe this part of the argument in more detail below.

In terms of efficiency, 
for the specific LWE-based NTCF  that we construct, our protocol can use as few as $\poly\log(N)$ bits of randomness to 
generate $O(N)$ bits that are statistically within negligible distance from uniform. However, this  requires assuming that the underlying LWE assumption is hard even for sub-exponential size quantum circuits with polynomial-size quantum advice (which is consistent with current knowledge). The more conservative assumption that our variant of LWE is only hard for polynomial size quantum circuits requires $O(N^{\epsilon})$ bits of randomness for generating the NTCF, for any constant $\epsilon >0$. The following is an informal description; see Theorem~\ref{thm:expansion} for a more formal statement.

\begin{theorem}[Informal]
Let $\mathcal{F}$ be an NTCF family and $\lambda$ a security parameter. Let $N = \Omega(\lambda^2)$ and assume the quantum hardness of solving lattice problems of dimension $\lambda$ in time $\poly(N)$. There is an $N$-round protocol for the interaction between a classical polynomial-time verifier and a quantum polynomial-time device such that the protocol can be executed using $\poly(\log(N),\lambda)$ bits of randomness, and for any efficient device and side information $\reg{E}$ correlated with the device's initial state,
$$\Hmin^{\delta}(O|CE)_{\ol{\rho}} \geq  (\xi-o(1)) N\;.$$
Here $\xi$ is a positive constant, $\delta$ is a negligible function of $\lambda$, and $\ol{\rho}$ is the final state of the classical output register $\reg{O}$, the classical register $\reg{C}$ containing the verifier's messages to the device, and the side information $\reg{E}$, restricted to transcripts that are accepted by the verifier in the protocol.
\end{theorem}

\paragraph{Sketch of the security analysis.}
We describe the protocol in slightly more detail (see Section~\ref{sec:protocol} for a formal description). The verifier first uses $\poly(\log(N),\lambda)$ bits of randomness to 
select a pair of functions $\{f_{k,b}\}_{b\in\{0,1\}}$ from an NTCF family, and sends the public function key $k$ to the quantum device. This pair of functions can be interpreted as a single $2$-to-$1$ function $f_k:(b,x)\mapsto f_{k,b}(x)$. The verifier keeps private the trapdoor information that allows to invert $f_k$.
The protocol then proceeds for $N$ rounds. In each round the device first outputs a value $y$ in the common range of $f_{k,0}$ and $f_{k,1}$. After having received $y$, the verifier issues one of two 
challenges: $0$ or $1$, preimage or equation. If the challenge is ``preimage'', then the device must output an $x$ such that $f(x) = y$. If the challenge is 
``equation'' then the device must output a nontrivial binary vector $d$ and a bit $c$ such that $d\cdot(x_0 + x_1) = c$, where $x_0$ and $x_1$ are the unique preimages of $y$ under $f_{k,0}$ and $f_{k,1}$ respectively. Since the verifier has the secret key, she can efficiently compute $x_0$ and $x_1$ from $y$, and therefore check the correctness of the
device's response to each challenge. The verifier chooses 
$\poly\log(N)$ rounds in which to issue the challenge $1$, or ``equation'', at random. Selecting these rounds requires only
$\poly\log(N)$ random bits. At the end of each such round, the verifier samples a new pair of functions from the NTCF family, and communicates the new public key to the device. On each of the remaining $N - \poly\log(N)$ rounds the verifier records a bit according to 
whether the device returns the preimage $x_0$, or $x_1$ (e.g.\ recording $0$ for the lexicographically smaller preimage). At the end of the protocol the verifier
uses a strong quantum-proof randomness extractor to extract $\Omega(N)$ bits of randomness from the recorded string (this requires at most an additional $\poly\log(N)$ bits of uniformly random seed).

To guarantee that the extractor produces bits that are statistically close to uniform, we would like to prove that the $N - \poly\log(N)$ random bits recorded by the verifier must have  $\Omega(N)$ bits of (smoothed) min-entropy,\footnote{We refer to Section~\ref{sec:prelim} for definitions of entropic quantities.} even conditioned on the side information available to an infinitely powerful quantum adversary, who may share an arbitrary entangled state with
the quantum device. 

The analysis proceeds as follows. First we assume without loss of generality that the entire protocol is run coherently, i.e.\ we may
assume that the initial state of the quantum device (holding quantum register $\reg{D}$) and the adversary (holding quantum register $\reg{E}$)
is a pure state $\ket{\phi}_\reg{DE}$, since the adversary may as well start with a purification of their joint state. We may also assume that the verifier 
starts with a cat state on $\poly\log(N)$ qubits, and uses one of the registers of the state, $\reg{C}$, to provide the random bits used to select the type of test being performed in each round. (This is for the sake of analysis only, the actual verifier is of course completely classical.) We can similarly  arrange that 
the state remains pure throughout the protocol by using the principle of deferred measurement. Our goal is to show a lower bound on the smooth
min-entropy of the output register $\reg{O}$ in which the verifier has recorded the device's outputs, conditioned on the state $\reg{E}$ of the adversary, and on the register $\reg{C}$ of the cat state  (conditioning on the latter represents the fact that the verifier's choice of challenges may be leaked to the adversary, and we would like security even in this scenario). 
Intuitively, this amounts to bounding the information accessible to the most powerful adversary quantum mechanics allows, conditioned on the joint state of the verifier and device.

In order to bound the entropy of the final state we need to show that the entropy ``accumulates'' at each round of the protocol. A general framework to establish entropy accumulation in quantum protocols such as the one considered here was introduced in~\cite{arnon2018practical}. At a high level, the approach consists in reducing the goal of a min-entropy bound to a bound on the appropriate notion of $(1+\eps)$ quantum conditional R\'enyi entropy, and then arguing that, under suitable conditions on the process that generates the outcomes recorded in the protocol, entropy accumulates sequentially throughout the protocol.

In a little more detail, the first step on getting a handle on the smooth min-entropy is to use the quantum asymptotic equipartion property (QAEP)~\cite{tomamichel2009fully} to relate it to the $(1+\eps)$ conditional R\'enyi entropy, for suitably small $\eps$. The second step uses a duality relation for the conditional R\'enyi entropy to relate the $(1+\eps)$ conditional R\'enyi entropy of the output register $\reg{O}$, conditioned on the adversary side information in $\reg{R}$ and the register $\reg{C}$ of the cat state,  to a quantity analogous to the $(1-\eps')$ conditional R\'enyi entropy of the output register, conditioned on the register $\reg{E}$ for the device, and a purifying copy of the register $\reg{C}$ of the cat state. The latter quantity, a suitable conditional entropy of the output register conditioned on the challenge register and the state of the device, is the quantity that we ultimately aim to bound. Note what these transformations have achieved for us: it is now sufficient to consider as side information only ``known'' quantities in the protocol, the verifier's choice of challenges and the device's state; the information held by the adversary plays no other role than that of a purifying register. 

As mentioned earlier, our cryptographic setting presents the additional difficulty that our guarantee is only that it is {computationally hard} for a deterministic device to succeed in the protocol. The results in~\cite{arnon2018practical,miller2014universal} crucially rely on the fact that the process that generates the randomness does so irrespective of the quantum state in which it is initialized (as long as the output of the process satisfies the test's success criterion). This requirement comes  from the conditioning that is performed in order to show that entropy accumulates; in our setting, conditioning is more delicate as it can in principle induce non-computationally efficient states for the device.

Recall that we argued that for a single round of the protocol, we can decompose the device's Hilbert space into a direct sum of one- or two-dimensional subspaces, such that within most two-dimensional subspace the ``preimage'' and ``equation'' measurements correspond to orthonormal bases that make an angle close to $\pi/4$ with each other.
Showing that the R\'enyi entropy accumulates in each round requires a device in which \emph{all} angles are close to $\pi/4$, not ``almost all''. To accommodate for this we ``split'' the state of the device into its component on the good subspace, where the angles are unbiased, and the bad subspace, where the measurements may be aligned. The fact that the distinction between good and bad subspace is not measured in the protocol, but is only a distinction made for the analysis, requires us to apply a fairly delicate martingale based argument that takes into account possible interference effects and bounds those ``branches'' where the state has gone through the bad subspace an improbably large number of times. Whenever the state lies in the good susbpace, we can appeal to
an uncertainty principle from~\cite{miller2014universal} to show that the device's measurement increases the conditional R\'enyi entropy of the output register by a small additive constant. 
Pursuing this approach across all $N$ rounds, we obtain a linear lower bound on the conditional R\'enyi entropy of the output register, conditioned on the  state of the device. As argued above this in turn translates into a linear lower bound on the smooth conditional min-entropy of the output, conditioned on the state of the adversary and the verifier's choice of challenges. It only remains to apply a quantum-proof randomness extractor to the output, using a poly-logarithmic number of additional bits of randomness, to obtain the final result.

\paragraph{Concurrent and related work.}
The idea of using a TCF as a basic primitive in interactions between an efficient quantum prover and a classical verifier has been 
further developed in recent work by Mahadev \cite{mahadev2017classical}, giving the first construction of a quantum fully homomorphic encryption scheme with classical keys. 
In further follow-up work, Mahadev~\cite{UrmilaVerifiability} shows a remarkable use of a NTCF family with adaptive hardcore bit. Namely, that the NTCF can be used to certify that a prover measures a qubit in a prescribed basis (standard or Hadamard). This allows to achieve single prover \emph{verifiability} for quantum computations using a purely classical verifier (but relying on computational assumptions).

Independently of this work, a construction of trapdoor one-way functions with second preimage resistance based on LWE was recently introduced in~\cite{cojocaru2018delegated}, where it is used to achieve delegated computation in the weaker honest-but-curious model for the adversary (i.e.\ without soundness against provers not following the protocol). The family of functions considered in~\cite{cojocaru2018delegated} is not sufficient for our purposes, as it lacks the  adaptive hardcore bit property.

After the completion of our work, in~\cite{gheorghiu2019computationally} the construction of NTCF family introduced here was extended to a more general hardcore bit property (informally, over $\Z_8$ instead of $\Z_2$ here) and used to implement a two-party functionality called ``remote state preparation'' by which a classical client can ``force'' the preparation of one out of eight possible single-qubit quantum states by the prover. The authors of~\cite{cojocaru2019qfactory} also generalize~\cite{cojocaru2018delegated} to obtain a similar functionality; however, their construction does not offer the property of being verifiable (informally, it is possible for the server to prepare a state that is not the expected one). 

We believe that the technique of constraining the power of a quantum device using NTCFs promises to be a powerful tool for the field of untrusted quantum devices. 

\paragraph{Organization.} We start with some notation and preliminaries in Section~\ref{sec:prelim}. Section~\ref{sec:tcf} contains the definition of a noisy trapdoor claw-free family (NTCF). Our construction for such a family is given in Section~\ref{sec:lwetcf} (with Appendix~\ref{sec:lweprelim} containing relevant preliminaries on the learning with errors problem). The randomness generation protocol is described in Section~\ref{sec:protocol}. In Section~\ref{sec:device} we introduce our formalism for modeling the actions of an arbitrary prover, or device, in the protocol. In Section~\ref{sec:soundness} we analyze a single round of the protocol, and in Section~\ref{sec:multi-round} we show that randomness accumulates across multiple rounds.

\paragraph{Acknowledgments.}
We thank Tony Metger and the anonymous JACM referees for corrections and suggestions that improved the presentation of the paper. Zvika Brakerski is supported by the Israel Science Foundation (Grant No.\ 468/14), Binational Science Foundation (Grants No.\ 2016726, 2014276), and by the European Union Horizon 2020 Research and Innovation Program via ERC Project REACT (Grant 756482) and via Project PROMETHEUS (Grant 780701). Paul Christiano and Urmila Mahadev are supported by a Templeton Foundation Grant 52536, ARO Grant W911NF-12-1-0541, and NSF Grant CCF-1410022.
Umesh Vazirani is supported by MURI Grant FA9550-18-1-0161, ARO Grant W911NF-12-1-0541, NSF Grant CCF-1410022, NSF QLCI Grant OMA-2016245, and a Vannevar Bush Faculty Fellowship.
Thomas Vidick is supported by NSF CAREER Grant CCF-1553477, AFOSR YIP award number FA9550-16-1-0495, MURI Grant FA9550-18-1-0161, a CIFAR Azrieli Global Scholar award, and the IQIM, an NSF Physics Frontiers Center (NSF Grant PHY-1125565) with support of the Gordon and Betty Moore Foundation (GBMF-12500028).

\section{Preliminaries}
\label{sec:prelim}

\subsection{Notation}

$\bbZ$ is the set of integers, and $\N$ the set of natural numbers. 
For any $q \in \bbN$ such that $q\geq 2$ we let $\bbZ_q$ denote the ring of integers modulo $q$. We generally identify an element $x\in\bbZ_q$ with its unique representative $\trnq{x}\in (-\tfrac{q}{2}, \tfrac{q}{2}] \cap \bbZ$. For $x\in\bbZ_q$ we define $\abs{x}=|{\trnq{x}}|$.
When considering an $s\in \{0,1\}^n$ we sometimes also think of $s$ as an element of $\mZ_q^n$, in which case we write it as $\*s$.

We use the terminology of polynomially bounded and negligible functions. A function $n: \N \to \R_+$ is \emph{polynomially bounded} if there exists a polynomial $p$ such that $n(\lambda)\leq p(\lambda)$ for all $\lambda \in \N$. A function $n: \N \to \R_+$ is \emph{negligible} if for every polynomial $p$, $p(\lambda) n(\lambda)\to_{\lambda\to\infty} 0$. We write $\negl(\lambda)$ to denote an arbitrary negligible function of $\lambda$. For two parameters $\kappa,\lambda$ we write $\kappa \ll \lambda$ to express the constraint that $\kappa$ should be ``sufficiently smaller than'' $\lambda$, meaning that there exists a small universal constant $c>0$ such that $\kappa \leq c \lambda$, where $c$ is usually implicit for context. 

 $\mH$ always denotes a finite-dimensional Hilbert space. We use indices $\mH_\reg{A}$, $\mH_\reg{B}$, etc., to refer to distinct spaces. $\Pos(\mH)$ is the set of positive semidefinite operators on $\mH$, and $\Density(\mH)$ the set of density matrices, i.e. the positive semidefinite operators with trace $1$. For an operator $X$ on $\mH$, we use $\|X\|$ to denote the operator norm (largest singular value) of $X$, and $\|X\|_{tr} = \frac{1}{2}\|X\|_1 = \frac{1}{2}\Tr\sqrt{XX^\dagger}$ for the trace norm. 

\subsection{Distributions}

We generally use the letter $D$ to denote a distribution over a finite domain $X$, and $f$ for a density on $X$, i.e. a function $f:X\to[0,1]$ such that $\sum_{x\in X} f(x)=1$. We often use the distribution and its density interchangeably. We write $U$ for the uniform distribution. We write $x\leftarrow D$ to indicate that $x$ is sampled from distribution $D$, and $x\leftarrow_U X$ to indicate that $x$ is sampled uniformly from the set $X$. 
We write $\mathcal{D}_X$ for the set of all densities on $X$.
For any $f\in\mathcal{D}_X$, $\supp(f)$ denotes the support of $f$,
\begin{equation*}
    \supp(f) \,=\, \big\{x\in X \,|\; f(x)> 0\big\}\;.
\end{equation*}
For two densities $f_1$ and $f_2$ over the same finite domain $X$, the Hellinger distance  between $f_1$ and $f_2$ is
\begin{equation}\label{eq:bhatt}
H^2(f_1,f_2) \,=\, 1- \sum_{x\in X}\sqrt{f_1(x)f_2(x)}\;.
\end{equation}
The total variation distance between $f_1$ and $f_2$ is
\begin{equation}\label{eq:stattobhatt}
\|f_1-f_2\|_{TV} \,=\, \frac{1}{2} \sum_{x\in X}|f_1(x) - f_2(x)| \,\leq\, \sqrt{2H^2(f_1,f_2)}\;.
\end{equation}
The following immediate lemma relates the Hellinger distance and the trace distance of superpositions. 
\begin{lemma}
Let $X$ be a finite set and $f_1,f_2\in\mathcal{D}_X$. Let 
$$ \ket{\psi_1}=\sum_{x\in X}\sqrt{f_1(x)}\ket{x}\qquad\text{and}\qquad  \ket{\psi_2}=\sum_{x\in X}\sqrt{f_2(x)}\ket{x}\;.$$
 Then 
 $$\|\ket{\psi_1}-\ket{\psi_2}\|_{tr}\,=\, \sqrt{ 1 - (1-H^2(f_1,f_2))^2}\;.$$
\end{lemma}

We say that a family of quantum circuits $\{C_\lambda\}_{\lambda\in\N}$ (resp. observables $\{O_\lambda\}_{\lambda\in\N}$) is \emph{polynomial-time generated} if there exists a polynomial-time deterministic Turing
machine that, on every input $\lambda\in \N$, returns a gate-by-gate encoding of the circuit $C_\lambda$ (resp. of a circuit that implements $O_\lambda$). 
We introduce a notion of efficient distinguishability between distributions. 

\begin{definition}\label{def:compinddist}
We say that two families of distributions $D_0=\{D_{0,\lambda}\}_{\lambda\in\mN}$ and $D_1=\{D_{1,\lambda}\}_{\lambda\in\mN}$ on the same finite set $\{X_\lambda\}$ are \emph{computationally indistinguishable} if for every polynomial-time generated family of quantum circuits $\mathcal{A}=\{A_\lambda:\,X_\lambda\to\{0,1\}\}$ it holds that 
\begin{equation}
\Big|\Pr_{x\leftarrow D_{0,\lambda}}[A_\lambda(x) = 0] - \Pr_{x\leftarrow D_{1,\lambda}}[A_\lambda(x) = 0]\Big| \,=\, \negl(\lambda)\;,
\end{equation}
where the probability is taken over the choice of $x$ from either distribution as well as randomness inherent in any measurement performed by the circuit $A_\lambda$. 
\end{definition}

The next definition generalizes the previous one to the case of quantum states. 

\begin{definition}
We say that two families of sub-normalized density matrices $\sigma_0=\{\sigma_{0,\lambda}\}_{\lambda\in \N}$ and $\sigma_1 = \{\sigma_{1,\lambda}\}_{\lambda\in \N}$ on the same Hilbert space $\{\mH_\lambda\}$ are \emph{computationally indistinguishable} if for every polynomial-time generated family of observables $O=\{O_\lambda\}_{\lambda\in\N}$ it holds that 
$$ \big|\Tr\big(O_\lambda (\sigma_{0,\lambda} - \sigma_{1,\lambda})\big)\big| \,=\,\negl(\lambda)\;.$$
\end{definition}

\subsection{The Learning with Errors problem}
\label{sec:lweprelim}

We give some background on the Learning with Errors problem (LWE). 
For a positive real $B$ and a positive integer $q$, the 
truncated discrete Gaussian distribution over $\mZ_q$ with parameter $B$ is the distribution supported on $\{x\in\mZ_q:\,\|x\|\leq B\}$ with density
\begin{equation}\label{eq:d-bounded-def}
 D_{\mZ_q,B}(x) \,=\, \frac{e^{\frac{-\pi\lVert x\rVert^2}{B^2}}}{\sum\limits_{x\in\mZ_q,\, \|x\|\leq B}e^{\frac{-\pi\lVert x\rVert^2}{B^2}}} \;.
\end{equation}
More generally, for a positive integer $m$ the truncated discrete Gaussian distribution over $\mZ_q^m$ with parameter $B$ is the distribution supported on $\{x\in\mZ_q^m:\,\|x\|\leq B\sqrt{m}\}$ with density
\begin{equation}\label{eq:d-bounded-def-m}
\forall x = (x_1,\ldots,x_m) \in \mZ_q^m\;,\qquad D_{\mZ_q^m,B}(x) \,=\, D_{\mZ_q,B}(x_1)\cdots D_{\mZ_q,B}(x_m)\;.
\end{equation}

\begin{lemma}\label{lem:distributiondistance}
Let $B$ be a positive real and $q,m$ positive integers. Consider $\*e \in \mZ_q^m$ such that $\|\*e\|\leq B\sqrt{m}$. The Hellinger distance between the distribution $D = D_{\mZ_q^{m},B}$ and the shifted distribution $D+\*e$, with density $(D+\*e)(x)=D(x-\*e)$, satisfies
\begin{equation}
H^2(D,D+\*e) \,\leq\, 1- e^{\frac{-2\pi \sqrt{m}\|\*e\|}{B}}\;,
\end{equation}
and the statistical distance between the two distributions satisfies
\begin{equation}
\big\| D - (D+\*e) \big\|_{TV}^2 \,\leq\, 2\Big(1 - e^{\frac{-2\pi \sqrt{m}\|\*e\|}{B}}\Big)\;.
\end{equation}
\end{lemma}

\begin{proof}
Let $\tau = \sum\limits_{x\in\mZ_q,\, \|x\|\leq B}e^{\frac{-\pi\lVert x\rVert^2}{B^2}}$. We will rely on the fact that for any $\*e_0$ in the support of $D_{\mZ_q^m,B}$, $\|\*e_0\|\leq B\sqrt{m}$. We can compute the bound as follows:
\begin{eqnarray*}
\sum_{\*e_0\in \mZ_q^m} \sqrt{D_{\mZ_q^{m},B}(\*e_0)D_{\mZ_q^{m},B}(\*e_0-\*e)} &=&  \frac{1}{\tau^m}\sum_{\*e_0\in \mZ_q^m}e^{\frac{-\pi(\|\*e_0\|^2 + \|\*e_0 - \*e\|^2)}{2B^2}}\\
&\geq& \frac{1}{\tau^m}\sum_{\*e_0\in \mZ_q^m}e^{\frac{-\pi(\|\*e_0\|^2 + (\|\*e_0\| + \|\*e\|)^2)}{2B^2}}\\
&=& \frac{1}{\tau^m}\sum_{\*e_0\in \mZ_q^m}e^{\frac{-\pi(\|\*e_0\|^2)}{B^2}}e^{\frac{-\pi(2\|\*e_0\|\|\*e\|)}{2B^2}}e^{\frac{-\pi(\|\*e\|^2)}{2B^2}}\\
&\geq& e^{\frac{-\pi(\|\*e\|^2 + 2B\sqrt{m}\|\*e\|)}{2B^2}}\frac{1}{\tau^m}\sum_{\*e_0\in \mZ_q^m}e^{\frac{-\pi(\|\*e_0\|)^2}{B^2}}\\
&=& e^{\frac{-\pi(\|\*e\|^2 + 2B\sqrt{m}\|\*e\|)}{2B^2}}\\
&\geq& e^{\frac{-\pi(4B\sqrt{m}\|\*e\|)}{2B^2}}\\
&=& e^{\frac{-2\pi\sqrt{m}\|\*e\|}{B}}\;.
\end{eqnarray*}
The bound on the statistical distance follows from the bound on the Hellinger distance using the inequality in~\eqref{eq:stattobhatt}.
\end{proof}

We define the main assumption that underlies all computational hardness claims made in the paper. 

\begin{definition}\label{def:lwe-ass}
For a security parameter $\lambda$, let $n,m,q\in \bbN$ be integer functions of $\lambda$. Let $\chi = \chi(\lambda)$ be a distribution over $\mZ$. The $\lwe_{n,m,q,\chi}$ problem is to distinguish between the distributions $(\*A, \*A\*s + \*e \pmod{q})$ and $(\*A, \*u)$, where $\*A\leftarrow_U \bbZ_q^{n \times m}$, $\*s\leftarrow_U \mZ_q^n$, $\vc{e}\leftarrow\chi^m$, and $\*u \leftarrow_U \mZ_q^m$. Often we consider the hardness of solving $\lwe$ for {any} function $m$ such that $m$ is at most a polynomial in $n \log q$. This problem is denoted $\lwe_{n,q,\chi}$. 

In this paper we make the assumption that no quantum polynomial-time procedure can solve the $\lwe_{n,q,\chi}$ problem with more than a negligible advantage in $\lambda$, even when given access to a quantum polynomial-size advice state depending on the parameters $n,m,q$ and $\chi$ of the problem. We refer to this assumption as ``the $\lwe_{n,q,\chi}$ assumption''.
\end{definition}

As shown in \cite{regev2005,PRS17}, for any $\alpha>0$ such that  $\sigma = \alpha q \ge 2 \sqrt{n}$ the $\lwe_{n,q,D_{\mZ_q,\sigma}}$ problem,  where $D_{\mZ_q,\sigma}$ is the discrete Gaussian distribution, is at least as hard as approximating the shortest independent vector problem ($\sivp$) to within a factor of $\gamma = \otild({n}/\alpha)$, where $\tilde{O}$ hides factors logarithmic in the argument, in \emph{worst case} dimension $n$ lattices. This is proven using a quantum reduction. Classical reductions (to a slightly different problem) exist as well \cite{Peikert09,BLPRS13} but with somewhat worse parameters. The best known (classical or quantum) algorithm for these problems run in time $2^{\otild(n/\log \gamma)}$. For our construction, given in Section~\ref{sec:lwetcf}, we assume hardness of the problem against a quantum polynomial-time adversary in the case that $\gamma$ is a super polynomial function in $n$. This is a commonly used assumption in cryptography (for e.g. homomorphic encryption schemes such as \cite{fhelwe}).

We use two additional properties of the LWE problem. The first is that it is possible to generate LWE samples $(\*A,\*A\*s+\*e)$ such that there is a trapdoor allowing recovery of $\*s$ from the samples. 

\begin{theorem}[Theorem 5.1 in~\cite{miccancio2012}]\label{thm:trapdoor}
Let $n,m\geq 1$ and $q\geq 2$ be such that $m = \Omega(n\log q)$. There is an efficient randomized algorithm $\GenTrap(1^n,1^m,q)$ that returns a matrix $\*A \in \mZ_q^{m\times n}$ and a trapdoor $t_{\*A}$ such that the distribution of $\*A$ is negligibly (in $n$) close to the uniform distribution. Moreover, there is an efficient algorithm $\Invert$ that, on input $\*A, t_{\*A}$ and $\*A\*s+\*e$ where $\|\*e\| \leq q/(C_T\sqrt{n\log q})$ and $C_T$ is a universal constant, returns $\*s$ and $\*e$ with overwhelming probability over $(\*A,t_{\*A})\leftarrow \GenTrap(1^n,1^m,q)$. \end{theorem}

The second property is the existence of a ``lossy mode'' for LWE. The following definition is Definition~3.1 in~\cite{lwr}. 

\begin{definition}\label{def:lossy}
Let $\chi = \chi(\lambda)$ be an efficiently sampleable distribution over $\mZ_q$. Define a lossy sampler $\tilde{\*A} \leftarrow \lossy(1^n,1^m,1^\ell,q,\chi)$ by  $\tilde{\*A} = \*B\*C +\*F$, where $\*B\leftarrow_U \mZ_q^{m\times \ell}$, $\*C\leftarrow_U \mZ_q^{\ell \times n}$, $\*F\leftarrow \chi^{m\times n}$. 
\end{definition}

\begin{theorem}[Lemma 3.2 in~\cite{lwr}]\label{thm:lossy}
Under the $\lwe_{\ell,q,\chi}$ assumption, the distribution of a random $\tilde{\*A} \leftarrow \lossy(1^n,1^m,1^\ell,q,\chi)$ is computationally indistinguishable from $\*A\leftarrow_U \mZ_q^{m\times n}$. 
\end{theorem}

\subsection{Entropies}

For $p\in[0,1]$ we write $H(p) = -p\log p -(1-p)\log(1-p)$ for the binary Shannon entropy.
We measure randomness using R\'enyi conditional entropies. For a positive semidefinite matrix $\sigma\in\Pos(\mH)$ and $\eps\geq 0$, let 
$$\big\langle \sigma \big\rangle_{1+\eps} \,=\, \Tr \big(\sigma^{1+\eps}\big)\;.$$
This quantity satisfies the following approximate linearity relations:
\begin{equation}\label{eq:approx-lin}
 \forall\eps\in[0,1]\;,\qquad\langle \sigma \rangle_{1+\eps} + \langle \tau \rangle_{1+\eps} \,\leq\, \langle \sigma + \tau \rangle_{1+\eps} \,\leq\, \big(1+O(\eps)\big) \big( \langle \sigma \rangle_{1+\eps}+\langle \tau \rangle_{1+\eps}\big)\;.
\end{equation}
In addition, for positive semidefinite $\sigma,\rho\in\Pos(\mH)$ such that the support of $\rho$ is included in the support of $\sigma$, and $\eps\geq 0$, let
\begin{equation}\label{eq:def-q}
\tilde{Q}_{1+\eps}(\rho\|\sigma) \,=\, \langle \sigma^{-\frac{\eps}{2(1+\eps)}} \rho \sigma^{-\frac{\eps}{2(1+\eps)}} \rangle_{1+\eps}\;.
\end{equation}

%
%

Quantum analogues of the conditional R\'enyi entropies can be defined as follows. 

\begin{definition}\label{def:renyi}
Let $\rho_\reg{AB} \in \Pos(\mH_\reg{A}\otimes \mH_\reg{B})$ be positive semidefinite.  Given $\eps >0$, the $(1+\eps)$ \emph{R\'enyi entropy} of $A$ conditioned on $B$ is defined as 
$$H_{1+\eps}(A|B)_{\rho} \,=\, \sup_{\sigma\in\Density(\mH_\reg{B})} H_{1+\eps}(A|B)_{\rho|\sigma}\;,$$
where for any  $\sigma_\reg{B}\in\Density(\mH_\reg{B})$,
$$H_{1+\eps}(A|B)_{\rho|\sigma} \,=\, -\frac{1}{\eps} \log \tilde{Q}_{1+\eps}(\rho\|\sigma)\;.$$.
\end{definition}

R\'enyi entropies are used in the proofs because they have better ``chain-rule-like'' properties than the min-entropy, which is the most appropriate measure for randomness quantification. 

\begin{definition}\label{def:min-entropy}
Let $\rho_\reg{AB} \in \Pos(\mH_\reg{A}\otimes \mH_\reg{B})$ be positive semidefinite.  Given a density matrix  the \emph{min-entropy} of $A$ conditioned on $B$ is defined as
$$\Hmin(A|B)_\rho \,=\, \sup_{\sigma\in\Density(\mH_\reg{B})} \Hmin(A|B)_{\rho|\sigma}\;,$$
where for any $\sigma_\reg{B}\in \Density(\mH_\reg{B})$,
  \begin{equation*}
    \Hmin({A|B})_{\rho|\sigma} \,=\, \max \big\{\lambda \geq 0 \,|\; 2^{-\lambda} \Id_A \otimes \sigma_B \geq \rho_{AB}\big\}\;.
  \end{equation*}
\end{definition}

It is often convenient to consider the \emph{smooth} min-entropy, which is obtained by maximizing the min-entropy over all positive semidefinite operators matrices in an $\eps$-neighborhood of $\rho_\reg{AB}$. The definition of neighborhood depends on a choice of metric; the canonical choice is the ``purified distance''. Since this choice will not matter for us we defer to~\cite{tomamichel2015quantum} for a precise definition.

\begin{definition}\label{prelim:def:smooth-min-entropy}
  Let $\eps \geq 0$ and $\rho_\reg{AB}\in\Pos(\mH_\reg{A}\otimes\mH_\reg{B})$ positive semidefinite. The
  \emph{$\eps$-smooth min-entropy} of $A$ conditioned on $B$ is defined as
  \begin{equation*}
    \Hmin^\eps(A|B)_\rho \,=\, \sup_{\sigma_\reg{AB} \in \mathcal{B}(
      \rho_\reg{AB},\eps) } \Hmin(A|B)_\sigma\;,
  \end{equation*}
	where $\mathcal{B}(
      \rho_\reg{AB},\eps) $ is the ball of radius $\eps$ around $\rho_\reg{AB}$, taken with respect to the purified distance.
\end{definition}

The following theorem relates the min-entropy to the the R\'enyi entropies introduced earlier. The theorem expresses the fact that, up to a small amount of ``smoothing'' (the parameter $\delta$ in the theorem), all these entropies are of similar order. 

\begin{theorem}[Theorem 3.2~\cite{miller2017universal}]\label{thm:ms}
Let $ \rho_{\reg{XE}}\in\Pos(\mH_\reg{X}\otimes \mH_\reg{E}) $ be positive semidefinite of the form $\rho_{\reg{XE}} = \sum_{x\in\mX} \proj{x} \otimes \rho^x_{\reg{E}}$, where $\mX$ is a finite alphabet. Let $\sigma_{\reg{E}}\in\Density(\mH_\reg{E})$ be an arbitrary density matrix. Then for any $\delta >0$ and $0<\eps\leq 1$,
$$ \Hmin^\delta(X|E)_\rho \,\geq\, -\frac{1}{\eps} \log \Big( \sum_x \tilde{Q}_{1+\eps}\big(\rho_{\reg{E}}^x \|\sigma_{\reg{E}} \big)\Big) - \frac{1+2\log(1/\delta)}{\eps}\;.$$
\end{theorem}

\section{Trapdoor claw-free hash functions}
\label{sec:tcf}

Let $\lambda$ be a security parameter, and $\sX$ and $\sY$ finite sets (depending on $\lambda$). For our purposes an ideal family of functions $\mathcal{F}$ would have the following properties. For each public key $k$, there are two functions $ \{f_{k,b}:\sX\rightarrow \sY\}_{b\in\{0,1\}}$ that are both injective and have the same range, and are invertible given a suitable trapdoor $t_k$ (i.e. $t_k$ can be used to compute $x$ given $b$ and $y=f_{k,b}(x)$). Furthermore, the pair of functions should be claw-free: it must be hard for an attacker to find two pre-images $x_0,x_1\in\sX$ such that $f_{k,0}(x_0) = f_{k,1}(x_1)$. Finally, the functions should satisfy an adaptive hardcore bit property, which is a stronger form of the claw-free property: assuming for convenience that $\sX= \{0,1\}^w$, we would like that it is computationally infeasible to simultaneously generate a pair $(b,x_b)\in\{0,1\}\times \sX$ and a $d\in \{0,1\}^w\setminus \{0^w\}$ such that with non-negligible advantage over $\frac{1}{2}$ the equation $d\cdot (x_0\oplus x_1)=0$, where $x_{1-b}$ is defined as the unique element such that $f_{k,1-b}(x_{1-b})=f_{k,b}(x_b)$, holds.

Unfortunately, we do not know how to construct a function family that exactly satisfies all these requirements under standard cryptographic assumptions. Instead, we  construct a family that satisfies slightly relaxed requirements, that we will show still suffice for our purposes, based on the hardness of the learning with errors problem introduced in Section~\ref{sec:lweprelim}. The requirements are relaxed as follows. First, the range of the functions is no longer a set $\sY$; instead, it is  $\mathcal{D}_{\sY}$, the set of probability densities over $\sY$. That is, each function returns a density, rather than a point. The trapdoor injective pair property is then described in terms of the support of the output densities: these supports should either be identical, for a colliding pair, or be disjoint, in all other cases. 

The consideration of functions that return densities gives rise to an additional requirement of efficiency: there should exist a quantum polynomial-time procedure that efficiently prepares a superposition over the range of the function, i.e. for any key $k$ and $b\in\{0,1\}$, the procedure can prepare the state
\begin{equation}\label{eq:perfectsuperposition}
\frac{1}{\sqrt{\sX}}\sum_{x\in \sX, y\in \sY}\sqrt{\big(f_{k,b}(x)\big)(y)}\ket{x}\ket{y}\;.
\end{equation}
In our instantiation based on LWE, it is not possible to prepare~\eqref{eq:perfectsuperposition} perfectly, but it is possible to create a superposition with coefficients $\sqrt{(f'_{k,b}(x))(y)}$, such that the resulting state is within negligible trace distance of~\eqref{eq:perfectsuperposition}. The density $f'_{k,b}(x)$ is required to satisfy two properties used in our protocol. First, it must be easy to check, without the trapdoor, if an $y\in \sY$ lies in the support of $f'_{k,b}(x)$. Second, the inversion algorithm should operate correctly on all $y$ in the support of $f'_{k,b}(x)$.

We slightly modify the adaptive hardcore bit requirement as well. Since the set $\sX$ may not be a subset of binary strings, we first assume the existence of an injective, efficiently invertible map $\inj:\sX\to\{0,1\}^w$. Next, we only require the adaptive hardcore bit property to hold for a subset of all nonzero strings, instead of the  set $\{0,1\}^w\setminus \{0^w\}$. Finally, membership in the appropriate set should be efficiently checkable, given access to the trapdoor. 

A formal definition follows.

\begin{definition}[NTCF family]\label{def:trapdoorclawfree}
Let $\lambda$ be a security parameter. Let $\sX$ and $\sY$ be finite sets.
 Let $\mathcal{K}_{\mathcal{F}}$ be a finite set of keys. A family of functions 
$$\mathcal{F} \,=\, \big\{f_{k,b} : \sX\rightarrow \mathcal{D}_{\sY} \big\}_{k\in \mathcal{K}_{\mathcal{F}},b\in\{0,1\}}$$
is called a \emph{noisy trapdoor claw free (NTCF) family} if the following conditions hold:

\begin{enumerate}
\item{\textbf{Efficient Function Generation.}} There exists an efficient probabilistic algorithm $\textrm{GEN}_{\mathcal{F}}$ which generates a key $k\in \mathcal{K}_{\mathcal{F}}$ together with a trapdoor $t_k$: 
$$(k,t_k) \leftarrow \textrm{GEN}_{\mathcal{F}}(1^\lambda)\;.$$
\item{\textbf{Trapdoor Injective Pair.}} For all keys $k\in \mathcal{K}_{\mathcal{F}}$ the following conditions hold. 
\begin{enumerate}
\item \textit{Trapdoor}: There exists an efficient deterministic algorithm $\textrm{INV}_{\mathcal{F}}$ such that for all $b\in \{0,1\}$,  $x\in \sX$ and $y\in \supp(f_{k,b}(x))$, $\textrm{INV}_{\mathcal{F}}(t_k,b,y) = x$. Note that this implies that for all $b\in\{0,1\}$ and $x\neq x' \in \sX$, $\supp(f_{k,b}(x))\cap \supp(f_{k,b}(x')) = \emptyset$. 
\item \textit{Injective pair}: There exists a perfect matching $\sR_k \subseteq \sX \times \sX$ such that $f_{k,0}(x_0) = f_{k,1}(x_1)$ if and only if $(x_0,x_1)\in \sR_k$. \end{enumerate}

\item{\textbf{Efficient Range Superposition.}}
For all keys $k\in \mathcal{K}_{\mathcal{F}}$ and $b\in \{0,1\}$ there exists a function $f'_{k,b}:\sX\to \mathcal{D}_{\sY}$ such that the following hold.
\begin{enumerate} 
\item For all $(x_0,x_1)\in \mathcal{R}_k$ and $y\in \supp(f'_{k,b}(x_b))$, INV$_{\mathcal{F}}(t_k,b,y) = x_b$ and INV$_{\mathcal{F}}(t_k,b\oplus 1,y) = x_{b\oplus 1}$. 
\item There exists an efficient deterministic procedure CHK$_{\mathcal{F}}$ that, on input $k$, $b\in \{0,1\}$, $x\in \sX$ and $y\in \sY$, returns $1$ if  $y\in \supp(f'_{k,b}(x))$ and $0$ otherwise. Note that CHK$_{\mathcal{F}}$ is not provided the trapdoor $t_k$. 
\item For every $k$ and $b\in\{0,1\}$, 
$$ \Es{x\leftarrow_U \sX} \big[\,H^2(f_{k,b}(x),\,f'_{k,b}(x))\,\big] \,\leq\, \mu(\lambda)\;,$$
 for some negligible function $\mu(\cdot)$. Here $H^2$ is the Hellinger distance; see~\eqref{eq:bhatt}. Moreover, there exists an efficient procedure  SAMP$_{\mathcal{F}}$ that on input $k$ and $b\in\{0,1\}$ prepares the state
\begin{equation}
    \frac{1}{\sqrt{|\sX|}}\sum_{x\in \sX,y\in \sY}\sqrt{(f'_{k,b}(x))(y)}\ket{x}\ket{y}\;.
\end{equation}

\end{enumerate}

\item{\textbf{Adaptive Hardcore Bit.}}
For all keys $k\in \mathcal{K}_{\mathcal{F}}$ the following conditions hold, for some integer $w$ that is a polynomially bounded function of $\lambda$. 
\begin{enumerate}
\item For all $b\in \{0,1\}$ and $x\in \sX$, there exists a set $\dset_{k,b,x}\subseteq \{0,1\}^{w}$ such that $\Pr_{d\leftarrow_U \{0,1\}^w}[d\notin \dset_{k,b,x}]$ is negligible, and moreover there exists an efficient algorithm that checks for membership in $\dset_{k,b,x}$ given $k,b,x$ and the trapdoor $t_k$. 
\item There is an efficiently computable injection $\inj:\sX\to \{0,1\}^w$, such that $\inj$ can be inverted efficiently on its range, and such that the following holds. If
\begin{eqnarray*}\label{eq:defsetsH}
H_k &=& \big\{(b,x_b,d,d\cdot(\inj(x_0)\oplus \inj(x_1)))\,|\; b\in \{0,1\},\; (x_0,x_1)\in \mathcal{R}_k,\; d\in \dset_{k,0,x_0}\cap \dset_{k,1,x_1}\big\}\;,\text{\footnotemark}\\
\overline{H}_k &=& \{(b,x_b,d,c)\,|\; (b,x,d,c\oplus 1) \in H_k\big\}\;,
\end{eqnarray*}
\footnotetext{Note that although both $x_0$ and $x_1$ are referred to define the set $H_k$, only one of them, $x_b$, is explicitly specified in any $4$-tuple that lies in $H_k$.}
then for any quantum polynomial-time procedure $\mathcal{A}$ there exists a negligible function $\mu(\cdot)$ such that 
\begin{equation}\label{eq:adaptive-hardcore}
\Big|\Pr_{(k,t_k)\leftarrow \textrm{GEN}_{\mathcal{F}}(1^{\lambda})}[\mathcal{A}(k) \in H_k] - \Pr_{(k,t_k)\leftarrow \textrm{GEN}_{\mathcal{F}}(1^{\lambda})}[\mathcal{A}(k) \in\overline{H}_k]\Big| \,\leq\, \mu(\lambda)\;.
\end{equation}
\end{enumerate}

\end{enumerate}
\end{definition}

\section{A Trapdoor Claw-Free family based on LWE}
\label{sec:lwetcf}

In this section we present our LWE-based construction of an NTCF. For LWE-related preliminaries and definitions see Section~\ref{sec:lweprelim}.
Let $\lambda$ be a security parameter. All other parameters are functions of $\lambda$. Let $q\geq 2$ be a prime. 
Let $\ell,n,m\geq 1$ be polynomially bounded functions of $\lambda$ and $B_L, B_V, B_P$ be positive integers such that the following conditions hold:
\begin{enumerate}[label=(\textbf{A.\arabic*})]
\item\label{a1} $n = \Omega(\ell \log q + \lambda)$
\item\label{a2} $m = \Omega(n\log q)$,
\item\label{a3} $B_P = \frac{q}{2C_T\sqrt{mn\log q}}$, for $C_T$ the universal constant in Theorem~\ref{thm:trapdoor},
\item \label{a4} We have $B_L < B_V < B_P$ so that the ratios $\frac{B_P}{B_V}$ and $\frac{B_V}{B_L}$ are both super-polynomial  in $\lambda$.
\end{enumerate}
Given a choice of parameters satisfying all conditions~\ref{a1} to~\ref{a4},
we describe the function family $\mathcal{F}_{\lwe}$. Let $\sX = \mZ_q^n$ and $\sY = \mZ_q^m$. 
The key space $\mathcal{K}_{\mathcal{F}_{\lwe}}$ is a subset  of  $\mZ_q^{m\times n} \times \mZ_q^m$ defined in Section~\ref{sec:tcfgen}. For $b\in \{0,1\}$, $x\in \sX$ and key $k = (\*A,\*A\*s + \*e)$,  the density $f_{k,b}(x) $ is defined as
\begin{equation}\label{eq:defprobdensity}
  \forall y \in \sY,\quad   (f_{k,b}(x))(y) = D_{\mZ_q^m,B_P}(y - \*Ax - b\cdot \*A\*s)\;,
\end{equation}
where the density $D_{\mZ_q^m,B_P}$ is defined in~\eqref{eq:d-bounded-def}. It follows from the definition of the key generation procedure GEN$_{\mathcal{F}_{\lwe}}$ given in Section~\ref{sec:tcfgen} that $f_{k,b}$ is well-defined given $k=(\*A,\*A\*s + \*e)$, as for our choice of parameters $k$ uniquely identifies $s$. 

 The four properties required for a noisy trapdoor claw-free family, as specified in Definition~\ref{def:trapdoorclawfree}, are verified in the following subsections, providing a proof of the following theorem. Recall the definition of the hardness assumption $\lwe_{n,q,\chi}$ given in Definition~\ref{def:lwe-ass}.

\begin{theorem}\label{thm:lwetcf}
For any choice of parameters satisfying the conditions~\ref{a1} to~\ref{a4}, the function family $\mathcal{F}_{\lwe}$ is a noisy trapdoor claw free family under the hardness assumption $\lwe_{\ell,q,D_{\mZ_q,B_L}}$. 
\end{theorem}

\begin{remark}\label{rmk:parameters}
We briefly discuss possible parameter settings for a correct and secure realization of the construction. 

In order for known worst-case to average-case reductions to apply~\cite{regev2005} we should  set $B_L \ge 2 \sqrt{\ell}$. For the sake of efficiency we can choose $B_L$ so that equality holds. Since the evaluation algorithms run in $\poly(\ell)$ time we should take $\ell = \poly(\lambda)$.
The ratios $\frac{B_P}{B_V} = \frac{B_V}{B_L}$ affect the so-called ``statistical security parameter'' of the construction. Aiming for $2^{-\lambda}$ statistical security, we may set $\frac{B_P}{B_V} = \frac{B_V}{B_L} = 2^{\lambda}$. 

Once $\ell$ has been chosen, the parameters  $n,m$ are determined by conditions~\ref{a1}, \ref{a2} and $q$ is determined by condition~\ref{a3}. These conditions already imply that $q = 2^{2 \lambda} / \poly(\lambda)$. We need to set $\ell$ so that the LWE problem with the resulting $q$ is computationally hard. The hardness of the LWE problem scales very roughly as $2^{\tilde{\Omega}(\ell/\log(q/B_L))}$ (see e.g.\ \cite{Schnorr87,BKZ,BKZ20}). In our case $\log (q/B_L) = O(\lambda)$ and therefore we can choose $\ell \approx \lambda^2$, which would imply exponential hardness (in $\lambda$). 

We note that other choices of parameters are possible. For example, one could be satisfied with a statistical security parameter that is smaller than the computational security guarantee, thus choosing $\frac{B_P}{B_V}$, $\frac{B_V}{B_L}$ as more moderate functions of $\lambda$ and improving efficiency. Another possible consideration is that in our suggested setting the ratio $q/B_L$ scales sub-exponentially with $\ell$, which corresponds to the hardness of sub-exponential approximation for lattice problems. One might not want to assume that sub-exponential approximation is hard and instead choose the parameters so that $q/B_L$ scales more moderately as a function of $\ell$.
%
%
\end{remark}

\subsection{Efficient Function Generation}
\label{sec:tcfgen}

GEN$_{\mathcal{F}_{\lwe}}$ is defined as follows. First, the procedure samples a random $\*A\in \mZ_q^{m\times n}$, together with trapdoor information $t_{\*A}$. This is done using the procedure $\GenTrap(1^n,1^m,q)$ from Theorem~\ref{thm:trapdoor}. Recall that Assumption~\ref{a2} requires that $m = \Omega(n \log q)$ as needed for the theorem to hold. The trapdoor allows the evaluation of an inversion algorithm $\Invert$  that, on input $\*A$, $t_{\*A}$ and $b=\*A\*s + \*e$ returns $\*s$ and $\*e$ as long as $\|\*e\|\leq \frac{q}{C_T\sqrt{n\log q}}$. Moreover, the distribution on matrices $\*A$ returned by $\GenTrap$ is negligibly close to the uniform distribution on $\mZ_q^{m\times n}$.

Next, the sampling procedure selects $s\in \{0,1\}^n$ uniformly at random, and a vector $\*e\in \mZ_q^m$ by sampling each coordinate independently according to the distribution $D_{\mZ_q,B_V}$ defined in~\eqref{eq:d-bounded-def}. GEN$_{\mathcal{F}_{\lwe}}$ returns $k = (\*A,\*A\*s + \*e)$ and $t_k = t_{\*A}$.

\subsection{Trapdoor Injective Pair}\label{sec:trapdoortwotoonereq}

\begin{enumerate}
\item[(a)] \textit{Trapdoor.} It follows from~\eqref{eq:defprobdensity} and the definition of the distribution $D_{\mZ_q^m,B_P}$ in~\eqref{eq:d-bounded-def} that for any key $k=(\*A,\*A\*s+\*e)\in \mathcal{K}_{\mathcal{F}_{\lwe}}$ and for all $x\in \sX$,
\begin{eqnarray}
\supp(f_{k,0}(x)) &=& \big\{ \*Ax + \*e_0\,| \; \|\*e_0\|\leq B_P\sqrt{m}\big\}\;,\label{eq:supportofpdf0}\\
\supp(f_{k,1}(x)) &=& \big\{ \*Ax + \*A\*s + \*e_0\,| \;\|\*e_0\|\leq B_P\sqrt{m}\big\}\;.\label{eq:supportofpdf1}
\end{eqnarray}
The procedure $\textrm{INV}_{\mathcal{F}_{\lwe}}$ takes as input the trapdoor $t_{\*A}$, $b\in\{0,1\}$, and $y\in \sY$. It uses the algorithm $\Invert$ to determine $\*s_0,\*e_0$ such that $y = \*A\*s_0+\*e_0$, and returns the element $\*s_0 - b \cdot\*s\in\sX$. Using Theorem~\ref{thm:trapdoor}, this procedure returns the unique correct outcome provided $y = \*A\*s_0+\*e_0$ for some $\*e_0$ such that $  \|\*e_0\|  \,\leq\, \frac{q}{C_T\sqrt{n\log q}}$. This condition is satisfied for all $y\in \supp(f_{k,b}(x))$ provided $B_P$ is chosen so that
\begin{equation}\label{eq:trapdoortwotoonerequirement}
		 B_P \leq \frac{q}{C_T\sqrt{mn\log q}}\;,
\end{equation}
which is satisfied by the choice in~\ref{a3}.
\item[(b)] \textit{Injective Pair.} We let $\mathcal{R}_k$ be the set of all pairs $(x_0,x_1)$ such that $f_{k,0}(x_0) = f_{k,1}(x_1)$. By definition this occurs if and only if $x_1 = x_0 - \*s$, and so $\mathcal{R}_k$ is a perfect matching. 
\end{enumerate}

\subsection{Efficient Range Superposition} 
For $k=(\*A,\*A\*s+\*e)\in \mathcal{K}_{\mathcal{F}_{\lwe}}, b\in\{0,1\}$ and $x\in \sX$,
let
\begin{equation}\label{eq:defprobdensitymodified}
    (f'_{k,b}(x))(y) \,=\, D_{\mZ_q^m,B_P}(y - \*Ax - b\cdot (\*A\*s + \*e))\;.
\end{equation}
Note that $f'_{k,0}(x) = f_{k,0}(x)$ for all $x\in \sX$. The distributions $f'_{k,1}(x)$ and $f_{k,1}(x)$ are shifted by $\*e$. Given the key $k$ and $x\in\sX$, the densities $f'_{k,0}(x)$ and $f'_{k,1}(x)$ are efficiently computable. For all $x\in \sX$,
\begin{eqnarray}
\supp(f'_{k,0}(x)) &=& \supp(f_{k,0}(x))\;,\\
\supp(f'_{k,1}(x)) &=& \big\{ \*Ax + \*e_0 + \*A\*s + \*e\,| \; \|\*e_0\|\leq B_P\sqrt{m}\big\}\;.\label{eq:fprime1}
\end{eqnarray}
\begin{enumerate}
\item[(a)] Using that $B_V < B_P$, it follows that the norm of the term $\*e_0 + \*e$ in~\eqref{eq:fprime1} is always at most $2B_P\sqrt{m}$. Therefore, the inversion procedure $\textrm{INV}_{\mathcal{F}_{\lwe}}$ can be guaranteed to return $x$ on input $t_{\*A}$, $b\in \{0,1\}$, $y\in \supp(f'_{k,b}(x))$ if we strengthen the requirement on $B_P$ given in~\eqref{eq:trapdoortwotoonerequirement} to
\begin{equation}\label{eq:superpositiontrapdoorrequirement}
    B_P \,\leq\, \frac{q}{2C_T\sqrt{mn\log q}}\;,
\end{equation}
which is still satisfied by~\ref{a3}. 
This strengthened trapdoor requirement also implies that for all $b\in \{0,1\}$, $(x_0,x_1)\in\mathcal{R}_k$, and $y\in \supp(f'_{k,b}(x_b))$, INV$_{\mathcal{F}_{\lwe}}(t_{\*A},b\oplus 1,y) = x_{b\oplus 1}$. 
\item[(b)] On input $k = (\*A,\*A\*s + \*e)$, $b\in\{0,1\}$, $x \in \sX$, and $y\in\sY$, the procedure CHK$_{\mathcal{F}_{\lwe}}$ operates as follows. If $b=0$, it computes $\*e' = y - \*Ax$. If $\|\*e'\| \leq B_P\sqrt{m}$, the procedure returns $1$, and $0$ otherwise. If $b = 1$, it computes $\*e' = y - \*Ax - (\*A\*s + \*e)$. If $\|\*e'\| \leq B_P\sqrt{m}$, it returns $1$, and $0$ otherwise. 

\item[(c)] 
We bound the Hellinger distance between the densities $f_{k,b}(x)$ and $f'_{k,b}(x)$. If $b=0$ they are identical. If $b=1$, both densities are shifts of $D_{\mZ_q^m,B_P}$, where the shifts differ by $\*e$. Each coordinate of $\*e$ is drawn independently from $D_{\mZ_q,B_V}$, so $\|\*e\|\leq \sqrt{m}B_V$. Applying Lemma~\ref{lem:distributiondistance}, we get that 
\begin{eqnarray*}
H^2(f_{k,1}(x),f'_{k,1}(x))\,\leq \, 1 - e^{\frac{-2\pi mB_V}{B_P}}\;.
\end{eqnarray*}
Using the assumption that $B_P/B_V$ is super-polynomial as required in Assumption~\ref{a4}, this is negligible, as desired. 
It remains to describe the  procedure SAMP$_{\mathcal{F}_{\lwe}}$. At the first step, the procedure creates the following superposition
\begin{equation}\label{eq:discretesup}
\sum_{\*e_0\in \mZ_q^m} \sqrt{D_{\mZ_q^m,B_P}(\*e_0)}\ket{\*e_0}\;.
\end{equation}
This state can be prepared efficiently as described in~\cite[Lemma 3.12]{regev2005}.\footnote{Specifically, the state can be created using a technique by Grover and Rudolph (\cite{distributionsuperpositions}), who show that in order to create such a state, it suffices to have the ability to efficiently compute the sum $\sum\limits_{x=c}^d D_{\mZ_q,B_P}(x)$  for any $c,d\in\{-\lfloor\sqrt{B_P}\rfloor,\ldots,\lceil\sqrt{B_P}\rceil\}\subseteq \mZ_q$  and to within good precision. This can be done using standard techniques used to sample from the normal distribution.}

At the second step, the procedure creates a uniform superposition over $x\in \sX$, yielding the state
\begin{equation}
q^{-\frac{n}{2}}\sum_{\substack{x\in \sX \\ \*e_0\in \mZ_q^{m}}} \sqrt{D_{\mZ_q^m,B_P}(\*e_0)}\ket{x}\ket{\*e_0}\;.
\end{equation}
At the third step, using the key $k = (\*A,\*A\*s + \*e)$ and the input bit $b$ the procedure computes  
\begin{equation}\label{eq:priortoerasingerror}
q^{-\frac{n}{2}}\sum_{\substack{x\in \sX \\ \*e_0\in \mZ_q^{m}}} \sqrt{D_{\mZ_q^m,B_P}(\*e_0)}\ket{x}\ket{\*e_0}\ket{\*Ax + \*e_0 + b\cdot (\*A\*s + \*e)}\;.
\end{equation}
At this point, observe that $\*e_0$ can be computed from $x$, the last register, $b$ and the key $k$. The procedure can then uncompute the register containing $\*e_0$, yielding 
\begin{align}
q^{-\frac{n}{2}}&\sum_{\substack{x\in \sX \\ \*e_0\in \mZ_q^{m}}} \sqrt{D_{\mZ_q^m,B_P}}(\*e_0)\ket{x}\ket{\*Ax + \*e_0 + b\cdot (\*A\*s + \*e)}\nonumber\\
&=q^{-\frac{n}{2}}\sum_{x\in \sX , y\in \sY}\sqrt{D_{\mZ_q^m,B_P}(y - \*Ax - b\cdot(\*A\*s + \*e))}\ket{x}\ket{y}\nonumber\\
&=  q^{-\frac{n}{2}}\sum_{x\in \sX ,y\in\sY}\sqrt{(f'_{k,b}(x))(y)}\ket{x}\ket{y}\;.\label{eq:fprimestate}
\end{align}

\end{enumerate}

\subsection{Adaptive Hardcore Bit} 
\label{sec:adaptive-bit}

This section is devoted to the proof that condition 4 of Definition \ref{def:trapdoorclawfree} holds. We start by providing a formal statement. Recall that $\sX=\mZ_q^n$ and let $w=n\lceil \log q \rceil$. Let $\inj:\sX\to\{0,1\}^w$  be such that $\inj(x)$ returns the binary representation of $x\in\sX$. For $b\in\{0,1\}$, $x\in \sX$, and $d\in\{0,1\}^w$, let
$I_{b,x}(d) \in \{0,1\}^n$ be the vector whose each coordinate is obtained by taking the inner product mod $2$ of the corresponding block of $\lceil\log q\rceil$ coordinates of $d$ and of $\inj(x)\oplus \inj(x-(-1)^b\*1 )$, where $\*1 \in \mZ_q^n$ is the vector with all its coordinates equal to $1\in\mZ_q$. For $k = (\*A,\*A\*s + \*e), b\in\{0,1\}$ and $x\in\sX$, we define the set $\dset_{k,b,x}$ as 
$$ \dset_{k,b,x}\,=\,\Big\{d\in \{0,1\}^w\,\Big|\; \exists i\in\Big\{b\,\frac{n}{2},\ldots,b\,\frac{n}{2}+\frac{n}{2}\Big\}:\,(I_{b,x}(d))_i \neq 0 \Big\}\;.$$
Observe that for all $b\in \{0,1\}$ and $x\in \sX$, if $d$ is sampled uniformly at random, $d\notin \dset_{k,b,x}$ with negligible probability. This follows simply because for any $b\in\{0,1\}$, $\inj(x)\oplus \inj(x-(-1)^b\*1)$ is non-zero, since $\inj$ is injective. Observe also that checking membership in $\dset_{k,b,x}$ is possible given only $b,x$. This shows condition 4.(a) in the adaptive hardcore bit condition in Definition~\ref{def:trapdoorclawfree}.

Given $(x_0,x_1)\in\mathcal{R}_k$ (where $k = (\*A,\*A\*s + \*e)$), recall from Section \ref{sec:trapdoortwotoonereq} that $x_1 = x_0 - \*s$. We use the following notation: we write $s\in\{0,1\}^n$ as $s = (s_0,s_1)$, where $s_0,s_1\in\{0,1\}^{\frac{n}{2}}$ are the $\frac{n}{2}$-bit prefix and suffix of $s$ respectively (for simplicity, assume $n$ is even; if not, ties can be broken arbitrarily). For convenience we also introduce the following set, where $y=f_{k,0}(x_0)=f_{k,1}(x_1)$:
\begin{equation}\label{eq:def-hatc}
	\hat{\dset}_{s_{1},0,x_0}\,=\,\hat{\dset}_{s_{0},1,x_1}\,=\,\dset_{k,0,x_0}\cap \dset_{k,1,x_1}\;.
\end{equation}
The motivation for using two different notation for the same set is to clarify that membership in the set can be decided given $(s_{b\oplus 1},b,x_b)$, for either $b\in\{0,1\}$. This point will be important in the proof of Lemma~\ref{lem:lweadaptiveleakage}.


The following lemma establishes item 4.(b) in  Definition~\ref{def:trapdoorclawfree}. 

\begin{lemma}\label{lem:lweadaptivehardcore} 
	Assume a choice of parameters satisfying the conditions~\ref{a1} to~\ref{a4}. Assume the hardness assumption $\lwe_{\ell,q,D_{\mZ_q,B_L}}$ holds. Let $s\in\{0,1\}^n$.  Let \footnote{We write the sets as $H_s$ instead of $H_k$ to emphasize the dependence on $s$.}
	\begin{eqnarray}
		H_s &=& \big\{(b,x,d,d\cdot(\inj(x)\oplus \inj( x - (-1)^b \*s )))\,|\; b\in \{0,1\},\, x\in \sX,\, d\in \hat{\dset}_{s_{b\oplus 1},b,x}  \big\}\;,\label{eq:def-h}\\
		\overline{H}_s &=& \big\{(b,x,d,c) \,|\; (b,x,d,c\oplus 1) \in H_s\big\}\;.
	\end{eqnarray}
	Then for any quantum polynomial-time procedure 
	$$\mathcal{A}:\, \mZ_q^{m\times n} \times \mZ_q^m \,\to\,\{0,1\}\times \sX \times \{0,1\}^w \times\{0,1\}$$
	there exists a negligible function $\mu(\lambda)$ such that 
	\begin{equation}\label{eq:lwebitattackeradvantage}
		\Big|\Pr_{(\*A,\*A\*s+\*e)\leftarrow \textrm{GEN}_{\mathcal{F}_{\lwe}}(1^{\lambda})}\big[\mathcal{A}(\*A,\*A\*s+\*e) \in H_s\big] - \Pr_{(\*A,\*A\*s+\*e)\leftarrow \textrm{GEN}_{\mathcal{F}_{\lwe}}(1^{\lambda})}\big[\mathcal{A}(\*A,\*A\*s+\*e) \in \overline{H}_s\big]\Big| \,\leq\, \mu(\lambda)\;.
	\end{equation} 
\end{lemma}


The proof of the lemma proceeds in three steps. First, in Section~\ref{sec:moderate} we establish some preliminary results on the distribution of the inner product $(\hat{d}\cdot s \bmod 2)$, where $\hat{d}\in\{0,1\}^n$ is a fixed nonzero binary vector and $s\leftarrow_U \{0,1\}^n$ a uniformly random binary vector, conditioned on $\*C\*s=\*v$ for some randomly chosen matrix $\*C\in \mZ_q^{\ell\times n}$ and arbitrary $\*v\in\mZ_q^{\ell}$. This condition is combined with the LWE assumption in Section~\ref{sec:lwehc} to argue that  $(\hat{d}\cdot s \bmod 2)$ remains computationally indistinguishable from uniform even when the matrix $\*C$ is an LWE matrix $\*A$, and the adversary is able to choose $\hat{d}$ after being given access to $\*A\*s+\*e$ for some error vector $\*e\in\mZ_q^m$. This will allow us to derive the following lemma, whose proof is provided in Section~\ref{sec:lwehc}.

We will show computational indistinguishability based on the hardness assumption $\lwe_{\ell,q,D_{\mZ_q,B_L}}$ specified in Definition~\ref{def:lwe-ass}. Since our goal is to prove Lemma~\ref{lem:lweadaptivehardcore}, we consider procedures that output a tuple $(b,x,d,c)\in \{0,1\}\times \sX\times \{0,1\}^w \times \{0,1\}$.

\begin{lemma}\label{lem:lweadaptiveleakage}
	Assume a choice of parameters satisfying the conditions~\ref{a1} to~\ref{a4}. Assume the hardness assumption $\lwe_{\ell,q,D_{\mZ_q,B_L}}$ holds. Let
	$$\mathcal{A}:\, \mZ_q^{m\times n} \times \mZ_q^m \,\to\,\{0,1\}\times \sX \times \{0,1\}^w \times\{0,1\}$$
	be a quantum polynomial-time procedure. For $b\in\{0,1\}$ and $x\in\sX$ let $I_{b,x}:\{0,1\}^w \to \{0,1\}^n$ be an efficiently computable map. For every $s = (s_0,s_1)\in\{0,1\}^n$ and $(b,x)\in\{0,1\}\times\sX$, let $\hat{\dset}_{s_{b\oplus 1},b,x}\subseteq \{0,1\}^w$ be a set depending only on $b,x$ and $s_{b\oplus 1}$ and such that for all $d\in \hat{\dset}_{s_{b\oplus 1},b,x}$ the first (if $b=0$) or last (if $b=1$) $\frac{n}{2}$ bits of $I_{b,x}(d)$ are not all $0$.
	Then the distributions
	\begin{eqnarray}\label{eq:l14-1}
		{D}_0 \,=\, \big( (\*A,\*A\*s+\*e) \leftarrow \textrm{GEN}_{\mathcal{F}_{\lwe}}(1^{\lambda}) ,\; (b,x,d,c) \leftarrow \mathcal{A}(\*A,\*A\*s+\*e),\; I_{b,x}(d) \cdot s \mod 2\big) 
	\end{eqnarray}
	and
	\begin{eqnarray}\label{eq:l14-2}
		{D}_1 \,=\, \big( (\*A,\*A\*s+\*e) \leftarrow \textrm{GEN}_{\mathcal{F}_{\lwe}}(1^{\lambda}) ,\;(b,x,d,c) \leftarrow \mathcal{A}(\*A,\*A\*s+\*e), \; (\delta_{d\in \hat{\dset}_{s_{b\oplus 1},b,x}} r) \oplus (I_{b,x}(d) \cdot s\mod 2) \big)\;, 
	\end{eqnarray}
	where $r\leftarrow_U \{0,1\}$ and $\delta_{d\in \hat{\dset}_{s_{b\oplus 1},b,x}}$ is $1$ if $d\in \hat{\dset}_{s_{b\oplus 1},b,x}$ and $0$ otherwise, are computationally indistinguishable. 
\end{lemma}

%

We prove Lemma~\ref{lem:lweadaptivehardcore} from Lemma~\ref{lem:lweadaptiveleakage} by relating the inner product appearing in the definition of $H_k$ (in condition 4.(b) of Definition \ref{def:trapdoorclawfree}) to an inner product of the form $\hat{d}\cdot s$, where $\hat{d}$ can be efficiently computed from $d$. This proof appears in Section~\ref{sec:hc-lemma} below.

\subsubsection{Proof of Lemma~\ref{lem:lweadaptivehardcore} from Lemma~\ref{lem:lweadaptiveleakage}}
\label{sec:hc-lemma}

The proof is by contradiction. Assume that there exists a quantum polynomial-time procedure $\mathcal{A}$ such that the left-hand side of~\eqref{eq:lwebitattackeradvantage} is at least some non-negligible function $\eta(\lambda)$. We derive a contradiction by showing that the two distributions ${D}_0$ and ${D}_1$ in Lemma~\ref{lem:lweadaptiveleakage}, for $I_{b,x}$ defined at the start of this section and $\hat{\dset}_{s_{b\oplus 1},b,x}$ defined in \eqref{eq:def-hatc}, are computationally distinguishable, giving a contradiction. 

Let $(\*A,\*A\*s+\*e) \leftarrow \textrm{GEN}_{\mathcal{F}_{\lwe}}(1^{\lambda})$ and $(b,x,d,c) \leftarrow \mathcal{A}(\*A,\*A\*s+\*e)$. To link $\mathcal{A}$ to the distributions in Lemma~\ref{lem:lweadaptiveleakage} we relate the inner product condition in~\eqref{eq:def-h} to an inner product $\hat{d} \cdot s$ of the form appearing in~\eqref{eq:l14-1}, for $\hat{d} = I_{b,x}(d)$ that can be computed efficiently from $b,x$ and $d$. This is based on the following claim.

\begin{claim}\label{cl:dependenceonsecret}
	For all $b\in\{0,1\},x\in \sX, d\in\{0,1\}^w$ and $s\in\{0,1\}^n$ the following equality holds:
	\begin{equation}\label{eq:def-hatd}
		d\cdot(\inj(x)\oplus \inj( x - (-1)^b \*s )) \,=\,  I_{b,x}(d)\cdot s \;.
	\end{equation}
	Moreover, the function $I_{b,x}$ is efficiently computable given $b,x$. 
\end{claim}

\begin{proof}
	We do the proof in case $n=1$ and $w=\lceil\log q\rceil$, as the case of general $n$ follows by linearity. In this case $s$ is a single bit. If $s=0$ then both sides of~\eqref{eq:def-hatd} evaluate to zero, so the equality holds trivially. It then suffices to define $ I_{b,x_b}(d)$ so that the equation holds when $s=1$. A choice of either of
	\begin{equation*}
		I_{0,x_0}(d)\,=\, d \cdot (\inj(x_0) \oplus \inj(x_0 -\*1))\;,\quad  I_{1,x_1}(d)\,=\,d \cdot (\inj(x_1) \oplus \inj(x_1 +\*1 ))\;
	\end{equation*}
	satisfies all requirements. It is clear from the definition of $I_{b,x}$ that it can be computed efficiently given $b,x$.  
\end{proof}

The procedure $\mathcal{A}$, the function $I_{b,x}$ defined at the start of this section and the sets $\hat{\dset}_{s_{b\oplus 1},b,x}$ in~\eqref{eq:def-hatc} fully specify ${D}_0$ and ${D}_1$. To conclude we construct a distinguisher $\mathcal{A}'$ between ${D}_0$ and ${D}_1$. Consider two possible distinguishers, $\mathcal{A}'_u$ for $u\in \{0,1\}$. Given a sample $w=((\*A,\*A\*s+\*e),(b,x,d,c),t)$, $\mathcal{A}'_u$ returns $0$ if $c=t\oplus u$, and $1$ otherwise. First note that
\begin{align}
	\sum_{u\in \{0,1\}} \Big|&\Pr_{w\leftarrow {D}_0}\big[\mathcal{A}_u'(w)=0\big]-\Pr_{w\leftarrow {D}_1}\big[\mathcal{A}_u'(w)=0\big]\Big|\notag\\
	&= \sum_{u\in\{0,1\}}\Big|\Pr_{w\leftarrow {D}_0}\big[\mathcal{A}_u'(w)=0 \wedge d\in \hat{\dset}_{s_{b\oplus 1},b,x} \big]-\Pr_{w\leftarrow {D}_1}\big[\mathcal{A}_u'(w)=0 \wedge d\in \hat{\dset}_{s_{b\oplus 1},b,x}\big]\Big|\label{eq:beforeusingHs}
\end{align}
since if $d\notin \hat{\dset}_{s_{b\oplus 1},b,x}$ the distributions ${D}_0$ and ${D}_1$ are identical by definition. Next, if the sample held by $\mathcal{A}'_u$ is from the distribution ${D}_0$ and if $(b,x,d,c) \in H_s$, then by the definition of $H_s$ and~\eqref{eq:def-hatd} it follows that $c = d\cdot(\inj(x)\oplus \inj( x - (-1)^b \*s )  = I_{b,x}(d)\cdot s = t$. If instead $(b,x,d,c)\in \overline{H}_s$ then $c\oplus 1 = d\cdot(\inj(x)\oplus \inj( x - (-1)^b \*s ) = I_{b,x}(d)\cdot s = t$.
The expression in \eqref{eq:beforeusingHs} is thus equal to: 
\begin{align*}
	\eqref{eq:beforeusingHs}&= \Big|\Pr_{(\*A,\*A\*s+\*e)\leftarrow \textrm{GEN}_{\mathcal{F}_{\lwe}}(1^{\lambda})}\big[\mathcal{A}(\*A,\*A\*s+\*e) \in H_s\big] -\frac{1}{2}\Pr_{w\leftarrow {D}_1}\big[ d\in \hat{\dset}_{s_{b\oplus 1},b,x}\big]\Big| \\
	&\hskip2cm+ \Big|\Pr_{(\*A,\*A\*s+\*e)\leftarrow \textrm{GEN}_{\mathcal{F}_{\lwe}}(1^{\lambda})}\big[\mathcal{A}(\*A,\*A\*s+\*e) \in \overline{H}_s\big]-\frac{1}{2}\Pr_{w\leftarrow {D}_1}\big[ d\in\hat{\dset}_{s_{b\oplus 1},b,x}\big]\Big|\\
	&\geq \Big|\Pr_{(\*A,\*A\*s+\*e)\leftarrow \textrm{GEN}_{\mathcal{F}_{\lwe}}(1^{\lambda})}\big[\mathcal{A}(\*A,\*A\*s+\*e) \in H_s\big] - \Pr_{(\*A,\*A\*s+\*e)\leftarrow \textrm{GEN}_{\mathcal{F}_{\lwe}}(1^{\lambda})}\big[\mathcal{A}(\*A,\*A\*s+\*e) \in \overline{H}_s\big]\Big| \\
	&\geq \eta\;.
\end{align*} 
Therefore, at least one of $\mathcal{A}'_0$ or $\mathcal{A}'_1$ must successfully distinguish between ${D}_0$ and ${D}_1$ with advantage at least $\frac{\eta}{2}$, a contradiction with the statement of Lemma~\ref{lem:lweadaptiveleakage}. 
\qed

\subsubsection{A Building block: Moderate matrices}
\label{sec:moderate}

The following lemma argues that, provided the matrix $\*C \in\mZ_q^{\ell\times n}$ is a uniformly random matrix with sufficiently few rows, the distribution $(\*C,\*C\*s)$ for arbitrary $s \in \{0,1\}^n$ does not reveal any parity of $s$.  

\begin{lemma}\label{lem:hardcore-1}
Let $q$ be a prime, $\ell,n\geq 1$ integers, and $\*C\in \mZ_q^{\ell\times n}$ a uniformly random matrix. With probability at least $1-q^\ell\cdot 2^{-\frac{n}{8}}$ over the choice of $\*C$ the following holds. For a fixed $\*C$, all $\*v\in\mZ_q^\ell$ and $\hat{d}\in \{0,1\}^n\setminus\{0^n\}$, the distribution of $(\hat{d}\cdot s \bmod 2)$, where $s$ is uniform in $\{0,1\}^n$ conditioned on $\*C\*s = \*v$, is within  statistical distance $O(q^{\frac{3\ell}{2}} \cdot 2^{-\frac{n}{40}})$ of the uniform distribution over $\{0,1\}$. 
\end{lemma}

To prove the lemma we introduce the notion of a \emph{moderate} matrix.  

\begin{definition}\label{def:moderate}
Let $\*b\in \mZ_q^n$. We say that $\*b$ is \textnormal{moderate} if it contains at least $\frac{n}{4}$ entries whose unique representative in $(-q/2,q/2]$ has its absolute value in the range $(\frac{q}{8},\frac{3q}{8}]$. A matrix $\*C\in \mZ_q^{\ell\times n}$ is moderate if its entire row span (except $0^n$) is moderate.
\end{definition}

\begin{lemma}\label{lem:moderate}
Let $q$ be prime and $\ell,n$ be integers. Then 
$$\Pr_{\*C\leftarrow_U \mZ_q^{\ell\times n}}\big(\text{$\*C$ is moderate}\big)\,\geq \, 1 - q^\ell \cdot 2^{-\frac{n}{8}}\;.$$ 
\end{lemma}

\begin{proof}
Consider an arbitrary non zero vector $\*b$ in the row-span of a uniform $\*C$. Then the marginal distribution of $\*b$ is uniform. By Chernoff, $\*b$ is moderate with probability at least $1 - e^{-\frac{2n}{16}}\geq 1 - 2^{-\frac{n}{8}}$. Applying the union bound over all at most $q^{\ell} - 1$ non zero vectors in the row span, the result follows. 
\end{proof}

\begin{lemma}\label{lem:singled}
Let $\*C\in \mZ_q^{\ell\times n}$ be an arbitrary moderate matrix and let $\hat{d}\in\{0,1\}^n\setminus\{0^n\}$ be an arbitrary non zero binary vector. Let $s$ be uniform over $\{0,1\}^n$ and consider the random variables $\*v = \*C\*s \bmod q$ and $z = \hat{d}\cdot s \bmod 2$. Then $(\*v,z)$ is within total variation distance at most $q^{\frac{\ell}{2}}\cdot 2^{-\frac{n}{40}}$ of the uniform distribution over $\mZ_q^{\ell}\times\{0,1\}$. 
\end{lemma}

\begin{proof}
Let $f$ be the probability density function of $(\*v,z)$. Interpreting $z$ as an element of $\mZ_2$, let $\hat{f}$ be the Fourier transform over $\mZ_q^{\ell}\times \mZ_2$. Let $U$ denote the density of the uniform distribution over $\mZ_q^{\ell}\times \mZ_2$. Applying the Cauchy-Schwarz inequality,
\begin{align}
\frac{1}{2}\big\| f - U \big \|_1 &\leq \sqrt{\frac{q^\ell}{2}} \big\| {f} - {U} \big\|_2 \notag \\
&= \frac{1}{2} \big\| \hat{f} - \hat{U} \big\|_2 \notag \\
&= \frac{1}{2} \Big( \sum_{(\hat{\*v},\hat{z})\in \mZ_q^\ell \times \mZ_2\backslash\{(\*0,0)\}} \big| \hat{f}(\hat{\*v},\hat{z})\big|^2\Big)^{1/2}\;,\label{eq:fc-1}
\end{align}
where the second line follows from Parseval's identity, and for the third line we used $\hat{f}(\*0,0)=\hat{U}(0,0)=1$ and $\hat{U}(\hat{\*v},\hat{z})=0$ for all $(\hat{\*v},\hat{z})\neq(0^\ell,0)$. To bound~\eqref{eq:fc-1} we estimate the Fourier coefficients of $f$. Denoting $\omega_{2q} = e^{-\frac{2\pi i}{2q}}$, for any $(\hat{\*v},\hat{z}) \in \mZ_q^\ell \times \mZ_2$ we can write 
\begin{align}
\hat{f}({\hat{\*v}},{\hat{z}}) &= \Es{\*s}\Big[\omega_{2q}^{(2\cdot {\hat{\*v}}^TC + q\cdot {\hat{z}}\hat{\*d}^T)\*s}\Big]\notag\\
&= \Es{\*s}\big[\omega_{2q}^{\*w^T\*s}\big] \notag\\
&= \prod_i\Es{s_i}\big[\omega_{2q}^{w_is_i}\big]\;,\label{eq:fc-3}
\end{align}
where we wrote $\*w^T = 2\cdot {\hat{\*v}}^T\*C + q\cdot {\hat{z}}\hat{\*d}^T\in \mZ_{2q}^n$. It follows that $\hat{f}(0^\ell,1)=0$, since $(d\cdot s \mod 2)$ is uniform for $s$ uniform.  

We now observe that for all $i\in\{1,\ldots,n\}$ such that the representative of $({\hat{\*v}}^T\*C)_i$ in $(-q/2,q/2]$ has its absolute value in $(\frac{q}{8},\frac{3q}{8}]$ it holds that $\frac{w_i}{q}\in (\frac{1}{4},\frac{3}{4}]\bmod 1$, in which case
\begin{equation}
\big|\Es{s_i}[\omega_{2q}^{w_is_i}]\big| \,=\, \Big|\cos\Big(\frac{\pi}{2}\cdot \frac{w_i}{q}\Big)\Big| \,\leq\, \cos\Big(\frac{\pi}{8}\Big)\,\leq\, 2^{-\frac{1}{10}}\;.
\end{equation}
Since $\*C$ is moderate, there are at least $\frac{n}{4}$ such entries, so that from~\eqref{eq:fc-3} it follows that $|\hat{f}({\hat{\*v}},{\hat{z}})|\leq 2^{-\frac{n}{40}}$ for all $\hat{\*v} \neq \*0$. Recalling~\eqref{eq:fc-1}, the lemma is proved.  
\end{proof}

We now prove Lemma \ref{lem:hardcore-1} by generalizing Lemma \ref{lem:singled} to adaptive $d$ (i.e. $d$ can  depend on $\*C, \*C\*s$).

\begin{proof}[Proof of Lemma \ref{lem:hardcore-1}]
We assume $\*C$ is moderate; by Lemma \ref{lem:moderate}, $\*C$ is moderate with probability at least $1-q^\ell\cdot 2^{-\frac{n}{8}}$. Let $s$ be uniform over $\{0,1\}^n$, $D_1 = (\*C\*s,\hat{d}\cdot s \bmod 2)$, and $D_2$ uniformly distributed over $\mZ_q^{\ell}\times \{0,1\}$. Using that $\*C$ is moderate, it follows from Lemma \ref{lem:singled} that 
\begin{equation}\label{eq:fc-4}
\eps\,=\,\|D_1-D_2\|_{TV} \,\leq\, q^{\frac{\ell}{2}}\cdot 2^{\frac{-n}{40}}\;.
\end{equation}
 Fix $\*v_0 \in \mZ_q^\ell$ and let 
\begin{eqnarray}\label{eq:fixeddistance}
\Delta \,=\, \frac{1}{2}\sum_{b\in \{0,1\}}\Big|\Pr_{s \leftarrow_U \{0,1\}^n}\big(\hat{d}\cdot s \bmod 2 = b \,\big| \, \*C\*s = \*v_0\big) - \frac{1}{2}\Big|\;.
\end{eqnarray}
To prove the lemma it suffices to establish the appropriate upper bound on $\Delta$, for all $\*v_0$. By definition, 
\begin{eqnarray}
\eps\,=\,\|D_1-D_2\|_{TV} &=& \frac{1}{2}\sum_{b\in \{0,1\}, \*v\in \mZ_q^{\ell}}\Big|\Pr\big(\*C\*s = \*v\big) \Pr\big(\hat{d}\cdot s \bmod 2 = b \,\big|\, \*C\*s = \*v\big) - \frac{1}{2q^\ell} \Big|\notag\\
&\geq& \frac{1}{2}\sum_{b\in \{0,1\}}\Big|\Pr\big(\*C\*s = \*v_0\big) \Pr\big(\hat{d}\cdot s \bmod 2 = b \,\big|\, \*C\*s = \*v_0\big) - \frac{1}{2q^\ell} \Big|\notag\\
&=& \frac{1}{2}\sum_{b\in \{0,1\}}\Big|\Pr\big(\*C\*s = \*v_0\big) \Big(\frac{1}{2} + (-1)^b\Delta\Big) - \frac{1}{2q^\ell} \Big|\;,\label{eq:afterfixeddistanceplugin}
\end{eqnarray}
where all probabilities are under a uniform choice of $s\leftarrow_U \{0,1\}^n$, and the last line follows from the definition of $\Delta$ in~\eqref{eq:fixeddistance}. Applying the inequality $|a|+|b| \geq \max(|a-b|,|a+b|)$, valid for any real $a,b$, to~\eqref{eq:afterfixeddistanceplugin} it follows that 
\begin{equation}
\Pr\big(\*C\*s = \*v_0\big)\cdot \Delta \,\leq\, \eps\qquad \text{and} \qquad 
\Pr\big(\*C\*s = \*v_0\big) \,\geq \,\frac{1}{q^\ell} - 2\eps\;.\label{eq:afterfixeddistanceplugin-2}
\end{equation}
If $q^{3\ell/2} 2^{-\frac{n}{40}} > \frac{1}{3}$ the bound claimed in the lemma is trivial. If $q^{3\ell/2} 2^{-\frac{n}{40}} \leq \frac{1}{3}$, then $\eps q^\ell \leq \frac{1}{3}$, so it follows from~\eqref{eq:afterfixeddistanceplugin-2} that $\Delta \leq 3q^\ell\eps$, which together with~\eqref{eq:fc-4} proves the lemma. 
\end{proof}

\subsubsection{Proof of Lemma~\ref{lem:lweadaptiveleakage}}
\label{sec:lwehc}

We use Lemma~\ref{lem:hardcore-1} and prove computational indistinguishability by introducing a sequence of hybrids. For the first step we let 
\begin{eqnarray}\label{eq:initialtolossy}
{D}^{(1)} \,=\, \big((\tilde{\*A}, \tilde{\*A}\*s + \*e),\; (b,x,d,c)\leftarrow \mathcal{A}(\tilde{\*A},\tilde{\*A}\*s + \*e),\; I_{b,x}(d)\cdot s \mod 2\big)\;,
\end{eqnarray}
where $\tilde{\*A} = \*B\*C + \*F \leftarrow \lossy(1^n,1^m,1^\ell,q,D_{\mZ_q,B_L})$ is sampled from a lossy sampler (see Definition~\ref{def:lossy}). From the definition, $\*F\in \mZ_q^{m\times n}$ has entries i.i.d.\ from the distribution $D_{\mZ_q,B_L}$ over $\mZ_q$. To see that ${D}_{0}$ and ${D}^{(1)}$ are computationally indistinguishable, first note that the distribution of matrices $\*A$ generated by $\textrm{GEN}_{\mathcal{F}_{\lwe}}$ is negligibly far from the uniform distribution (see Theorem~\ref{thm:trapdoor}). Next, by Theorem~\ref{thm:lossy}, under the $\lwe_{\ell,q,D_{\mZ_q,B_L}}$ assumption a uniformly random matrix $\*A$ and a lossy matrix $\tilde{\*A}$ are computationally indistinguishable. Note that this step, as well as subsequent steps, uses that $\mathcal{A}$ and $I_{b,x}$ are efficiently computable. 

For the second step we remove the term $\*F\*s$ from the lossy LWE sample $\tilde{\*A}\*s + \*e$ to obtain the distribution
\begin{eqnarray}\label{eq:lossystatdistance}
{D}^{(2)} \,=\, \big( (\*B\*C+\*F, \*B\*C\*s + \*e),\; (b,x,d,c)\leftarrow \mathcal{A}(\*B\*C + \*F,\*B\*C\*s + \*e),\; I_{b,x}(d) \cdot s\mod 2\big) \;.
\end{eqnarray}
Using that $\*s$ is binary and the entries of $\*F$ are taken from a $B_L$-bounded distribution, it follows that $\|\*F\*s\|\leq n\sqrt{m}B_L$. Applying Lemma~\ref{lem:distributiondistance}, the statistical distance between ${D}^{(1)}$ and ${D}^{(2)}$ is at most 
\begin{equation}\label{eq:def-gamma}
\gamma \,=\, \sqrt{2}\Big(1 - e^{\frac{-2\pi mnB_L}{B_V}}\Big)^{1/2}\;,
\end{equation}
which is negligible, due to the requirement that $\frac{B_V}{B_L}$ is superpolynomial given in~\ref{a4}.

For the third step, observe that the distribution ${D}^{(2)}$ in~\eqref{eq:lossystatdistance} only depends on $s_b$ through $\*C\*s$ and $I_{b,x}(d)\cdot s$, where $\*C$ is uniformly random. It follows from Lemma~\ref{lem:hardcore-1} that provided $\frac{n}{2}=\Omega(\ell\log q + \lambda)$ as required by Assumption~\ref{a1}, with overwhelming probability (in the security parameter) over the choice of $\*C$, if we fix all variables except for $s_b$, the distribution of $(I_{b,x}(d) \cdot s \mod 2)$ is statistically indistinguishable (within statistical distance $2^{-\lambda}$) from $r\leftarrow_U \{0,1\}$ as long as the $\frac{n}{2}$ bits of  $I_{b,x}(d)$ associated with $s_b$ are not all $0$ (i.e. the first $\frac{n}{2}$ bits if $b = 0$ or the last $\frac{n}{2}$ bits if $b = 1$). Using that for $d\in \hat{\dset}_{s_{b\oplus 1},b,x}$ the $\frac{n}{2}$ bits of  $I_{b,x}(d)$ associated with $s_b$ are not all $0$,  the distribution 
 ${D}^{(2)}$ in~\eqref{eq:lossystatdistance} is statistically indistinguishable from 
\begin{eqnarray*}
{D}^{(3)} \,=\,\big( (\*B\*C+\*F, \*B\*C\*s + \*e),\; (b,x,d,c)\leftarrow \mathcal{A}(\*B\*C + \*F,\*B\*C\*s + \*e),\; (\delta_{d\in \hat{\dset}_{s_{b\oplus 1},b,x}} r )\oplus (I_{b,x}(d)\cdot s\mod 2))\big)\;, 
\end{eqnarray*}
where $r\leftarrow_U \{0,1\}$.

For the fourth step we reinsert the term $\*F\*s$ to obtain
\begin{eqnarray*}
{D}^{(4)} \,=\, \big(\tilde{\*A}, \tilde{\*A}\*s + \*e,\; (b,x,d,c)\leftarrow \mathcal{A}(\tilde{\*A},\tilde{\*A}\*s + \*e),\; (\delta_{d\in \hat{\dset}_{s_{b\oplus 1},b,x}} r) \oplus (I_{b,x}(d) \cdot s\mod 2) \big)\;. 
\end{eqnarray*}
Statistical indistinguishability between ${D}^{(3)}$ and ${D}^{(4)}$ follows similarly as between ${D}^{(1)}$ and ${D}^{(2)}$.
Finally, computational indistiguishability between ${D}^{(4)}$ and ${D}_1$ follows similarly to between ${D}^{(1)}$ and ${D}_0$.
\qed

\section{Protocol description}
\label{sec:protocol}

We introduce two protocols. The first we call the \emph{(general) randomness expansion protocol}, or Protocol~1. This is our main randomness expansion protocol. It is introduced in Section~\ref{sec:re-protocol}, and summarized in Figure~\ref{fig:protocol}. The protocol describes the interaction between a \emph{verifier} and \emph{prover}. Ultimately, we aim to obtain the guarantee that any computationally bounded prover that is accepted with non-negligible probability by the verifier in the protocol must generate transcripts that contain information-theoretic randomness. 

The second protocol is called the \emph{simplified protocol}, or Protocol~2. It is introduced in Section~\ref{sec:si-protocol}, and summarized in Figure~\ref{fig:protocol2}. This protocol abstracts some of the main features Protocol~1, and will be used as a tool in the analysis (it is not meant to be executed literally).

\subsection{The randomness expansion protocol}
\label{sec:re-protocol}

Our randomness expansion protocol, Protocol~1, is described in Figure~\ref{fig:protocol}. The protocol is parametrized by a security parameter $\lambda$ and a number of rounds $N$. The other parameters, the error tolerance parameter $\gamma \geq 0$ and the testing parameter $q\in (0,1]$, are assumed to be specified as a function of $\lambda$ and $N$. For intuition, $\gamma$ can be thought of as a small constant and $q$ as a parameter that scales as $\poly(\lambda)/N$. 

At the start of the protocol, the verifier executes $(k,t_k)\leftarrow \Gen_\mF(1^\lambda)$ to obtain the public key $k$ and trapdoor $t_k$ for a pair of functions $\{f_{k,b}:\mX\to \mD_\mY\}_{b\in\{0,1\}}$ from the NTCF family (see Definition~\ref{def:trapdoorclawfree}). The verifier sends the public key $k$ to the prover and keeps the associated trapdoor private. 

In each of the $N$ rounds of the protocol, the prover is first required to provide a value $y\in\mY$. For each $b\in\{0,1\}$, the verifier uses the trapdoor to compute $\hat{x}_b\leftarrow \Inv_\mF(t_k,b,y)$. (If the inversion procedure fails, the verifier requests another sample from the prover.) For convenience, introduce a set 
\begin{equation}\label{eq:def-gy}
 \hat{\dset}_y \,=\, \dset_{k,0,x_0} \cap \dset_{k,1,x_1}\;,
\end{equation}
where for $b\in\{0,1\}$ the set $\dset_{k,b,x_b}$ is defined in 4.(a) of Definition~\ref{def:trapdoorclawfree}.
The verifier then chooses a round type $G\in\{0,1\}$ according to a biased distribution: either a \emph{test round}, $G=0$, chosen with probability $\Pr(G=0)=q$, or a \emph{generation round}, $G=1$, chosen with the remaining probability $\Pr(G=1)=1-q$. The former type of round is less frequent, as the parameter $q$ will eventually be set to a very small value, that goes to $0$ as the number of rounds of the protocol goes to infinity. The prover is not told the round type. 

Depending on the round type, the verifier chooses a challenge $C\in\{0,1\}$ that she sends to the prover. In the case of a test round the challenge is chosen uniformly at random; in the case of a generation round the challenge is always $C=1$. In case $C=0$ the prover is asked to return a pair $(u,d)\in \{0,1\}\times\{0,1\}^w$. The pair is called valid if $u=d\cdot(\inj(\hat{x}_0)\oplus \inj(\hat{x}_1))$ and $d\in\hat{\dset}_{y}$, where the function $\inj$ is as in 4.(b) of Definition~\ref{def:trapdoorclawfree}. If $d\in\hat{\dset}_y$, the verifier sets a decision bit $W=1$ if the answer is valid, and $W=0$ if not. If $d\notin \hat{\dset}_y$, the verifier sets the decision bit $W\in\{0,1\}$ uniformly at random.\footnote{This choice is made for technical reasons that have to do with the definition of the adaptive hardcore bit property; see Section~\ref{sec:soundness} and the proof of Proposition~\ref{prop:change-d} for details.}
In case $C=1$, the prover should return  a pair $(b,x)\in \{0,1\}\times\mX$. The pair is called valid if CHK$_{\mathcal{F}}(k,b,x,y)=1$. The verifier sets a decision bit $W=1$ in case the pair is valid, and $W=0$ otherwise. The set of valid pairs on challenge $C=c\in\{0,1\}$ is denoted $V_{y,c}$. 

After each test round the verifier samples a fresh $(k,t_{k})\leftarrow \Gen_\mF(1^\lambda)$ and communicates the new public key $k$ to the prover. 

At the end of the protocol, the verifier computes the fraction of test rounds in which the decision bit has been set to $1$. If this fraction is smaller than $(1-\gamma)$, the verifier aborts. Otherwise, the verifier returns the concatenation of the bits $b$ obtained from the prover in generation rounds. (These bits are recorded in the verifier's output string $O_1\cdots O_N$, such that $O_i=0$ whenever the round is a test round.)

\begin{figure}[htbp]
\rule[1ex]{16.5cm}{0.5pt}\\
Let $\lambda$ be a security parameter, $N\geq 1$ a number of rounds, and $\gamma,q>0$ functions of $\lambda$ and $N$. Let $\mF$ be an NTCF family.\\
At the start of the protocol, the verifier communicates $N$ to the prover. In addition, the verifier samples an initial key $(k,t_k)\leftarrow \Gen_\mF(1^\lambda)$, sends $k$ to the prover and keeps the trapdoor information $t_k$ private. 
\begin{enumerate}
\item For $i=1,\ldots,N$:
\begin{enumerate}
\item The prover returns a  $y \in \mY$ to the verifier. For $b\in\{0,1\}$ the verifier uses the trapdoor to compute $\hat{x}_b\leftarrow \Inv_\mF(t_k,b,y)$. 
\item The verifier selects a round type $G_i \in \{0,1\}$ according to a Bernoulli  distribution with parameter $q$: $\Pr(G_i=0)=q$ and $\Pr(G_i=1)=1-q$. In case $G_i=0$ (\emph{test round}), she chooses a challenge $C_i\in \{0,1\}$ uniformly at random. In case $G_i=1$ (\emph{generation round}), she sets $C_i=1$. The verifier keeps $G_i$ private, and sends $C_i$ to the prover. 
\begin{enumerate}
\item In case $C_i=0$ the prover returns $(u,d)\in\{0,1\}\times \{0,1\}^w$. If $d\notin \hat{\dset}_y$, the set defined in~\eqref{eq:def-gy}, the verifier sets $W$ to a uniformly random bit. Otherwise, the verifier sets $W=1$ if $d\cdot (\inj(\hat{x}_0)\oplus \inj(\hat{x}_1)) = u$ and $W=0$ if not.
\item In case $C_i=1$ the prover returns $(b,x)\in\{0,1\}\times \mX$. The verifier sets $W$ as the value returned by CHK$_{\mathcal{F}}(k,b,x,y)$. 
\end{enumerate}
\item In case $G_i=1$, the verifier sets $O_i = b$. In case $G_i=0$, she sets $W_i = W$. 
\item In case $G_i=0$, the verifier samples a new key $(k,t_k)\leftarrow \Gen_\mF(1^\lambda)$. She sends $k$ to the prover and keeps the trapdoor information $t_k$ private. This key will be used until the next test round is completed.  
\end{enumerate}
\item If $\sum_{i: G_i=0} W_i < (1-\gamma)qN$, the verifier aborts. Otherwise, she returns the string $O$ obtained by concatenating the bits $O_i$ for all $i\in\{1,\ldots,N\}$ such that $G_i=1$. 
\end{enumerate}
\rule[1ex]{16.5cm}{0.5pt}
\caption{The randomness expansion protocol, Protocol~1. See Definition~\ref{def:trapdoorclawfree} for notation associated with the NTCF family $\mathcal{F}$.}
\label{fig:protocol}
\end{figure}

\subsection{The simplified protocol}
\label{sec:si-protocol}

\begin{figure}[htbp]
\rule[1ex]{16.5cm}{0.5pt}\\
Let $\lambda$ be a security parameter, $N\geq 1$ a number of rounds, and $\gamma,\eta,\kappa,q>0$ functions of $\lambda$ and $N$. 
\begin{enumerate}
\item For $i=1,\ldots,N$:
\begin{enumerate}
\item The verifier selects a round type $G_i \in \{0,1\}$ according to a Bernoulli  distribution with parameter $q$: $\Pr(G_i=0)=q$ and $\Pr(G_i=1)=1-q$. In case $G_i=0$ (\emph{test round}), she chooses $C_i\in \{0,1\}$ uniformly at random and $T_i\in\{0,1\}$ such that $\Pr(T_i=0)=1-\kappa$ and $\Pr(T_i=1)=\kappa$. In case $G_i=1$ (\emph{generation round}), she sets $C_i=1$ and $T_i=0$. The verifier keeps $G_i$ private, and sends $(C_i,T_i)$ to the prover. 
\begin{enumerate}
\item In case $C_i=0$ the prover returns $e\in\{0,1\}$. If $T_i=1$ the prover in addition reports $k\in\{0,1\}$.\footnotemark\ If $T_i=0$ the verifier sets $W_i = e$. If $T_i=1$ the verifier sets $W_i=e(1-k)$.  
\item In case $C_i=1$ the prover returns $v\in\{0,1,2\}$. The verifier sets $O_i=v$ and $W_i = 1_{v\in\{0,1\}}$.   
\end{enumerate}
\end{enumerate}
\item If $\sum_{i: G_i=0 \wedge T_i=1} W_i < (1-\frac{\gamma}{\kappa}-\eta)\kappa qN$, the verifier rejects the interaction. Otherwise, she returns the string $O$ obtained by concatenating the bits $O_i$ for all $i\in\{1,\ldots,N\}$ such that $G_i=1$. 
\end{enumerate}
\rule[1ex]{16.5cm}{0.5pt}
\caption{The simplified protocol, Protocol 2.}
\label{fig:protocol2}
\end{figure}

For purposes of analysis only we introduce a simplified variant of Protocol~1, which is specified in Figure~\ref{fig:protocol2}. We call it the \emph{simplified protocol}, or Protocol~2. The protocol is very similar to the randomness expansion protocol described in Figure~\ref{fig:protocol}, except that the prover's answers and the verifier's checks are simplified, and in test rounds there is an additional challenge bit $T\in\{0,1\}$. This new challenge asks the prover to perform a projective measurement on its private space that indicates whether the state lies in a ``good subspace'' (indicated by an outcome $K=0$) or in the complementary ``bad subspace'' (outcome $K=1$). The ``good'' and ``bad'' subspaces represent portions of space where the device's other two measurements, $M$ and $\Pi$ are anti-aligned and aligned respectively; see the definition of a simplified device in Section~\ref{sec:binary} for details.

 For the case of a challenge $C=0$, in Protocol~1 the prover returns an equation $(u,d)$. In the simplified protocol the prover returns a single bit $e\in\{0,1\}$ that is meant to directly indicate the verifier's decision (i.e. the bit $W$). If moreover $T=1$ the prover is required to reply with an additional bit $k\in\{0,1\}$. In this case, the verifier makes the decision to accept, i.e. sets $W=1$, if and only if $e=1$ and $k=0$. 
 For the case of a challenge $C=1$, in Protocol~1 the prover returns a pair $(b,x)$. In the simplified protocol the prover returns a value $v\in\{0,1,2\}$ that is such that $v=b$ in case $(b,x)$ is valid, i.e. $(b,x)\in V_{y,1}$, and $v=2$ otherwise. 

Note that this ``honest'' behavior for the prover is not necessarily efficient. Moreover, it is easy for a ``malicious'' prover to succeed in Protocol 2, e.g. by always returning $u=1$ (valid equation), $k=0$ (good subspace) and $v\in\{0,1\}$ (valid pre-image). Our analysis will not consider arbitrary provers in Protocol~2, but instead provers whose measurements satisfy certain  constraints that arise from the analysis of Protocol~1. For such provers, it will be impossible to succeed in the simplified protocol without generating randomness. Further details are given in Section~\ref{sec:soundness}.

\footnotetext{The bit $k$ should not be confused with the public key $k$ for the NTCF that is used in Protocol~1. In Protocol~2, there is no NTCF, and no key.}

\subsection{Completeness}
\label{sec:completeness}

We describe the intended behavior for the prover in Protocol 1. Fix an NTCF family $\mathcal{F}$ and a key $k\in\mathcal{K}_\mathcal{F}$. In each round, the ``honest'' prover performs the following actions.
\begin{enumerate}
\item The prover executes the efficient procedure  SAMP$_{\mathcal{F}}$ in superposition to obtain the state
\[ \ket{\psi^{(1)}}\,=\,    \frac{1}{\sqrt{|\sX|}}\sum_{x\in \sX,y\in \sY,b\in\{0,1\}}\sqrt{(f'_{k,b}(x))(y)}\ket{b,x}\ket{y}\;.\]
\item The prover measures the last register to obtain an $y\in\mY$. Using item 2. from the definition of an NTCF, the prover's re-normalized post-measurement state is
\[\ket{\psi^{(2)}} \,=\, \frac{1}{\sqrt{2}}\big(\ket{0,x_0}+\ket{1,x_1}\big)\ket{y}\;,\]
where for $b\in\{0,1\}$, $x_b =\, $INV$_{\mathcal{F}}(t_k,b,y)$. 
\begin{enumerate}
\item In case $C_i=0$, the prover evaluates the function $\inj$ on the second register, containing $x_b$, and then applies a Hadamard transform to all $w+1$ qubits in the first two registers. Tracing out the register that contains $y$, this yields the state 
\begin{align*}
\ket{\psi^{(3)}} &= 2^{-\frac{w+2}{2}}  \sum_{d,b,u} (-1)^{d\cdot \inj(x_b)\oplus ub} \ket{u}\ket{d}\\
&=  2^{-\frac{w}{2}}  \sum_{d\in\{0,1\}^w} (-1)^{d\cdot\inj(x_0)}\ket{d\cdot (\inj(x_0)\oplus \inj(x_1))}\ket{d}\;.
\end{align*}
The prover measures both registers to obtain an $(u,d)$ that it sends back to the verifier. 
\item In case $C_i=1$, the prover measures the first two registers of $\ket{\psi^{(2)}}$ in the computational basis, and returns the outcome $(b,x_b)$ to the verifier.
\end{enumerate}
\end{enumerate}

\begin{lemma}\label{lem:completeness}
For any $\lambda$ and $k\leftarrow\GEN_{\mathcal{F}}(1^\lambda)$, the strategy for the honest prover (on input $k$) in one round of the protocol can be implemented in time polynomial in $\lambda$ and is accepted with probability negligibly close to $1$.  
\end{lemma}

\begin{proof}
Both efficiency and correctness of the prover follow from the definition of an NTCF (Definition~\ref{def:trapdoorclawfree}). The prover fails only if it obtains an outcome $d\notin \hat{\dset}_y$, which by item 4(a) in the definition happens with negligible probability.
\end{proof}

\section{Devices}
\label{sec:device}

We model an arbitrary prover in the randomness expansion protocol (Protocol 1 in Figure~\ref{fig:protocol}) as a \emph{device} that implements the actions of the prover: the device first returns an $y\in\mY$; then, depending on the challenge $C\in\{0,1\}$, it either returns an equation $(u,d)$ (case $C=0$), or a candidate pre-image $(b,x)$ (case $C=1$). For simplicity we assume that the device makes the same set of measurements in each round of the protocol. This is without loss of generality, as we allow the state of the device to change from one round to the next; in particular the device is allowed to use a quantum memory  as a control register for the measurements.

In Section~\ref{sec:devices} we introduce our notation for modeling provers in Protocol~1 as  devices. In Section~\ref{sec:binary} we consider a simplified form of device, that is appropriate for modeling a prover in the simplified protocol, Protocol~2. In Section~\ref{sec:soundness} we give a reduction showing how to associate a specific simplified device to any computationally efficient general device, such that the randomness generation properties of the two devices can be related to each other (this is done in Section~\ref{sec:multi-round}).  

For the remainder of this section we fix an NTCF family $\mathcal{F}$ satisfying the conditions of Definition~\ref{def:trapdoorclawfree}, and use notation introduced in the definition. 

\subsection{General devices}
\label{sec:devices}

The following notion of device models the behavior of an arbitrary prover in the randomness expansion protocol, Protocol~1 (Figure~\ref{fig:protocol}). 

\begin{definition}\label{def:device}
Given $k\in \mK_\mF$, a device $D = (\phi,\Pi,M)$ (implicitly, compatible with $k$) is specified by the following:
\begin{enumerate}
\item A normalized density $\phi \in \Pos(\mH_\reg{D}\otimes \mH_\reg{Y})$. Here $\mH_\reg{D}$ is an arbitrary space private to the device, and $\mH_{\reg{Y}}$ is a space of the same dimension as the cardinality of the set $\mY$, also private to the device. 
 For every $y\in\mY$, define
$$\phi_y \,=\, (\Id_{\reg{D}} \otimes \bra{y}_\reg{Y})\,\phi\,(\Id_{\reg{D}} \otimes \ket{y}_\reg{Y})\,\in\,\Pos(\mH_\reg{D})\;.$$
Note that $\phi_y$ is sub-normalized, and $\sum_{y\in\mY} \Tr(\phi_y)=\Tr(\phi)=1$. 
\item For every $y\in\mY$, a projective measurement $\{M_y^{(u,d)}\}$ on $\mH_\reg{D}$, with outcomes $(u,d)\in \{0,1\}\times \{0,1\}^w$. 
\item For every $y\in\mY$, a projective measurement $\{\Pi_y^{(b,x)}\}$ on $\mH_\reg{D}$, with outcomes $(b,x)\in \{0,1\}\times\mX$. For each $y$, this measurement has two designated outcomes $(0,x_0)$ and $(1,x_1)$, which are the answers that are accepted on challenge $C=1$ in the protocol; recall that we use the notation $V_{y,1}$ for this set. For $b\in\{0,1\}$ we use the shorthand $\Pi_y^b = \Pi_y^{(b,x_b)}$, $\Pi_y= \Pi_y^0+\Pi_y^1$, and $\Pi_y^2 = \Id - \Pi_y^0-\Pi_y^1$.
\end{enumerate}
\end{definition}

By Naimark's theorem, up to increasing the dimension of $\mH_\reg{D}$ the assumption that $\{\Pi_y^{(b,x)}\}$ and $\{M_y^{(u,d)}\}$ are projective is without loss of generality. 

We explain the connection between the notion of device in Definition~\ref{def:device} and a prover in Protocol~1. Given a device $D = (\phi,\Pi,M)$, we can define actions for the prover in Protocol~1 as follows. The prover is initialized in state $\phi$. When a round of the protocol is initiated, the prover measures register $\reg{Y}$ in the computational basis and returns the outcome $y\in\mY$. We always assume that the prover directly measures the register, as any pre-processing unitary can be incorporated  in the definition of the state $\phi$. When sent challenge $C=0$ (resp. $C=1$), the prover measures register $\reg{D}$ using the device's projective measurement $\{M_y^{(u,d)}\}$ (resp. $\{\Pi_y^{(b,x)}\}$), and returns the outcome to the verifier. 

\begin{definition}
We say that a device $D = (\phi,\Pi,M)$ is \emph{efficient} if
\begin{enumerate}
\item There is a uniformly generated family of polynomial-size circuits that prepare the state $\phi$ given the NTCF key $k$ as input;
\item For every $y\in\mY$, the measurements $\{M_y^{(u,d)}\}$ and $\{\Pi_y^{(b,x)}\}$ can be implemented by polynomial-size circuits. 
\end{enumerate}
\end{definition}

Using the definition of an NTCF family (Definition~\ref{def:trapdoorclawfree}), it is straightforward to verify that the device associated with the ``honest'' prover described in Section~\ref{sec:completeness} is efficient. 

We introduce notation related to the post-measurement states generated by a device  in Protocol~1. An execution of Protocol~1 involves a choice of round types $g\in\{0,1\}^N$ and  challenges $c\in\{0,1\}^N$ by the verifier, and a sequence of outputs $o\in\{0,1,2\}^N$ computed by the verifier as a function of the answers provided by the device. Here, in case $g=0$ (test round) we use $o\in\{0,1\}$ to denote the outcome of the test (called $W$ in the protocol description), and in case $g=1$ (generation round) we use $o\in\{0,1,2\}$ such that  $o=2$ in case $W=0$, and $o=O$ as recorded by the verifier in case $W=1$. We call the tuple $(g,c,o)$ the transcript of the protocol; it contains all the information relevant to the verifier's final acceptance decision and to the extraction of randomness. Additional information such as the choice of NTCF key and the prover's complete answers (including the value $y$) is discarded for ease of presentation. We let $\Acc$ denote the set of transcripts $(g,c,o)$ that are accepted by the verifier in the last step of the protocol, i.e. such that $\sum_{i: g_i=0} o_i \geq (1-\gamma)qN$. 

\begin{definition}\label{def:post-meas}
Let $D=(\phi,\Pi,M)$ be a device. For any transcript $(g,c,o)$ for an execution of Protocol~1 with $D$, let $ \phi_\reg{D}^{co}$ be the post-measurement state of the device, conditioned on having received challenges $c$ and returned outcomes $o$. 
The joint state of the transcript and the device at the end of the $N$ rounds (but before the verifier's decision to abort) is
\begin{equation}\label{eq:pm-2}
 \phi^{(N)}_{\reg{COD}} \,=\, \sum_{g,c,o}\, q(g,c) \, \proj{c}_\reg{C} \otimes \proj{o}_\reg{O} \otimes \phi_\reg{D}^{co}\;,
\end{equation}
where $q(g,c)$ is the probability that the sequence of round types and challenges $(g,c)$ is chosen by the verifier in the protocol.

We write $\ket{\phi}_{\reg{DE}}$ for a purification of the initial state $\phi_{\reg{D}}$ of the device, with $\reg{E}$ the purifying register, and $\rho_{\reg{E}}^{co}$ for the post-measurement state on register $\reg{E}$ conditioned on the transcript being $(c,o)$. 
\end{definition}

\subsection{Simplified devices}
\label{sec:binary}

Next we introduce a simplified notion of device, that can be used to model the actions of a prover in the simplified protocol, Protocol~2 (Figure~\ref{fig:protocol2}).

\begin{definition}\label{def:binary-device}
A \emph{simplified device} is a tuple $(\phi,\Pi,M,K)$ such that:
\begin{enumerate}
\item $\phi=\{\phi_y\}_{y\in\mY} \subseteq \Pos(\mH_\reg{D})$ is a family of positive semidefinite operators on an arbitrary space $\mH_\reg{D}$ such that $\sum_y \Tr(\phi_y )\leq 1$; 
\item For each $y\in\mY$, $\{M_y^0,M_y^1=\Id-M_y^0\}$, $\{\Pi_y^0,\Pi_y^1,\Pi_y^2=\Id-\Pi_y^0-\Pi_y^1 \}$, and $\{K_y^0,K_y^1=\Id-K_y^0\}$  are projective measurements on $\mH_\reg{D}$;
\item For each $y\in\mY$, the measurement operators $K_y$ commute with the $M_y$ and with the $\Pi_y$. ($M_y$ and $\Pi_y$ do not necessarily commute with each other.)
\end{enumerate}
\end{definition}

We introduce a quantity called \emph{overlap} that measures how ``incompatible'' a simplified device's measurements are. This measure is analogous to the measure of overlap used to quantify incompatibility in the derivation of entropic uncertainty relations (see e.g.~\cite{maassen1988generalized}). 

\begin{definition}\label{def:overlap}
Given a simplified device $D=(\phi,\Pi,M,K)$, the \emph{overlap} of $D$ is 
\[\Delta(D)\,=\,\max_{y\in\mY}\,\big\|K_y^0\big(\Pi_y^0  M_y^1 \Pi_y^0 + \Pi_y^1 M_y^1 \Pi_y^1 \big)\big\|\;.\]
\end{definition}

Note that the overlap only quantifies the measurement incompatibility in the ``good subspace'' $K_y^0$. 

 To any simplified device $D=(\phi,\Pi,M,K)$ we associate the post-measurement states
\begin{align}
\forall e\in\{0,1\},\quad \phi_{00}^{e}&= \sum_{y\in\mY} \proj{y}\otimes  M_y^e \phi_y M_y^e\;,\notag\\
\forall e,k\in\{0,1\},\quad \phi_{01}^{ek}&= \sum_{y\in\mY} \proj{y}\otimes K_y^k M_y^e \phi_y M_y^e  K_y^k\;,\notag\\
\forall v\in\{0,1,2\},\quad\phi_1^v &= \sum_{y\in\mY} \proj{y}\otimes \Pi_y^v \phi_y \Pi_y^v\;.\label{eq:def-pm}
\end{align}

A simplified device can be used in the simplified protocol in a straightforward way: upon receipt of a challenge $C=0$ (resp. $C=1$), the device first samples an $y\in\mY$ according to the distribution with weights $\Tr(\phi_y)$. It then performs the projective measurement $\{M_y^0,M_y^1\}$ followed by, if $T=1$, $\{K_y^0,K_y^1\}$ (resp. $\{\Pi_y^0,\Pi_y^1,\Pi_y^2 \}$) on $\phi_y$, and returns the outcomes $e,k\in\{0,1\}$ (resp. $v\in\{0,1,2\}$) to the verifier. 

\begin{definition}\label{def:post-meas-simplified}
Let $D=(\phi,\Pi,M,K)$ be a simplified device. For any transcript $(g,c,t,o,k)$ for an execution of Protocol~2 with $D$, let $ \phi_\reg{D}^{ctok}$ be the post-measurement state of the device, conditioned on having received challenges $(c,t)$ and returned outcomes $(o,k)$. 
The joint state of the transcript and the device at the end of the $N$ rounds (but before the verifier's decision to abort) of the protocol is
\begin{equation}\label{eq:pm-2b}
 \phi^{(N)}_{\reg{CTOKD}} \,=\, \sum_{g,c,t,o,k}\, q(g,c,t)\, \proj{c,t}_\reg{CT} \otimes \proj{o,k}_\reg{OK} \otimes \phi_\reg{D}^{ctok}\;,
\end{equation}
where $q(g,c,t)=q(g,c)\kappa(t)$ with $\kappa(t) = \prod_i \kappa^{t_i}(1-\kappa)^{1-t_i}$ is the probability that the sequence of round types and challenges $(g,c,t)$ is chosen by the verifier in the protocol.
\end{definition}

\section{Single-round analysis}
\label{sec:soundness}

In this section we consider the behavior of an arbitrary device $D$ in a single round of the randomness expansion protocol, Protocol~1 in Figure~\ref{fig:protocol}. Our goal is to introduce a simplified device $D'$ such that analyzing the randomness generation properties of $D'$ is easier than it is for $D$, and such that bounds on the amount of randomness generated by $D'$ in the simplified protocol, Protocol~2 in Figure~\ref{fig:protocol2}, imply bounds on  the amount of randomness generated by $D$ in Protocol~1. 
Throughout the section we fix an NTCF family $\mathcal{F}$ (Definition~\ref{def:trapdoorclawfree}) and a key $k\in \mK_\mF$ sampled according to $\Gen(1^\lambda)$, for a parameter $\lambda$ that plays the role of security parameter. 

\subsection{A constraint on the measurements of any efficient device}

We start with a lemma showing that for any efficient device $D = (\phi,\Pi,M)$, the measurements $\Pi$ and $M$ must be strongly incompatible, in the sense that if the device first measures $\Pi$, and then measures $M$, it is unable to determine if the pair $(u,d)$ returned by $M$ corresponds to a valid pair, i.e. $(u,d)\in V_{y,0}$. Indeed, if this were the case the device could be used to violate the hardcore bit property~\eqref{eq:adaptive-hardcore}.
Recall the definition of the set $\hat{\dset}_y\subseteq\{0,1\}^w$ in~\eqref{eq:def-gy}.

\begin{lemma} \label{lem:break}
Let $D = (\phi,\Pi,M)$ be an efficient device. Define a sub-normalized density 
\begin{align}
\tilde{\phi}_{\reg{YBXD}} &= \sum_{y\in\mY} \proj{y}_\reg{Y}\otimes \sum_{b\in\{0,1\}} \proj{b,x_b}_\reg{BX} \otimes \Pi_y^{(b,x_b)} \,\phi_y \,\Pi_y^{(b,x_b)}\label{eq:def-sigmay}
\end{align}
Let
\begin{align}
\sigma_0 &=\sum_{b\in\{0,1\}} \proj{b,x_b}_{\reg{BX}} \otimes  \sum_{(u,d) \in {V}_{y,0}} \proj{u,d}_{\reg{U}}\otimes (\Id_\reg{Y}\otimes M_y^{(u,d)}) \tilde{\phi}^{(b)}_{\reg{YD}} (\Id_\reg{Y}\otimes M_y^{(u,d)})\;,\notag\\
 \sigma_1 &=\sum_{b\in\{0,1\}} \proj{b,x_b}_{\reg{BX}} \otimes  \sum_{(u,d)\notin V_{y,0}} 1_{d\in \hat{\dset}_{y} } \proj{u,d}_{\reg{U}}\otimes (\Id_\reg{Y}\otimes M_y^{(u,d)}) \tilde{\phi}^{(b)}_{\reg{YD}} (\Id_\reg{Y}\otimes M_y^{(u,d)})\;,\label{eq:def-rho}
\end{align}
where $1_{d\in \hat{\dset}_{y} }$ denotes the indicator function for the event that $d\in \hat{\dset}_y$. 
Then $\sigma_0$ and $\sigma_1$ are computationally indistinguishable. 
\end{lemma}

Informally, $\sigma_0$ and $\sigma_1$ in~\eqref{eq:def-rho} are the result of performing the pre-image measurement $\{\Pi^{(b,x_b)}\}$ on $\phi_y$, directly followed by an equation measurement: $\sigma_0$ is the post-measurement state associated with correct equations, and $\sigma_1$ with wrong equations. Indistinguishability of the two states follows from the hardcore bit property~\eqref{eq:adaptive-hardcore}, which specifies that it is computationally infeasible to obtain a valid pre-image together with a correct equation. 

\begin{proof}
Suppose for contradiction that there exists an efficient observable $O$ such that 
\begin{equation}\label{eq:bias-o}
\Tr(O(\sigma_0-\sigma_1)) \,\geq\, \mu\;,
\end{equation}
 for some non-negligible function $\mu(\lambda)$. Consider the following efficient procedure. The procedure first prepares the state $\tilde{\phi}_{\reg{YBXD}}$ in~\eqref{eq:def-sigmay}. This can be done efficiently by first preparing $\phi_{\reg{YD}}$, then measuring a $y\in \mY$, then applying the measurement $\{\Pi_y^{(b,x)}\}$ to $\phi_y$, and returning a special abort symbol if the outcome is invalid, i.e. CHK$_{\mathcal{F}}(k,b,x,y)=0$. 

The procedure then applies the measurement $\{M_y^{(u,d)}\}$ to $\tilde{\phi}_{\reg{YBXD}}$, obtaining an outcome $(u,d)$. At this point, conditioned on the event that $d\in \hat{\dset}_{y}$, depending on whether $(u,d)\in V_{y,0}$ or $(u,d)\notin V_{y,0}$ the procedure has either prepared $\sigma_0$ or $\sigma_1$. Finally, the procedure  measures $O$ to obtain a bit $v$, and returns $(b,x,d,v\oplus u)$. This defines an efficient procedure. Moreover, using~\eqref{eq:bias-o} it follows  that the procedure violates the hardcore bit property~\eqref{eq:adaptive-hardcore}. (The cases where $d\notin \hat{\dset}_{y}$ are not taken into account by the hardcore bit property, so it is sufficient to have a good distinguishing ability conditioned on $d\in \hat{\dset}_{y}$.) 
\end{proof}

\subsection{Angles between incompatible measurements}

We show a general lemma that argues about the principal angles between two binary-outcome measurements that have a certain form of incompatibility. 

\begin{lemma}\label{lem:angles}
Let $\Pi,M$ be two orthogonal projections on $\mH$ and $\phi$ a state on $\mH$. Let $\gamma = 1-\Tr(M\phi)$ and 
\[\mu \,=\, \Big|\frac{1}{2} - \Tr\big(M\Pi\phi \Pi\big)-\Tr\big(M(\Id-\Pi)\phi(\Id-\Pi)\big)\Big|\;.\]
Let $\frac{1}{2} <  \omega \leq 1$. Let $K$ be the orthogonal projection on the direct sum of eigenspaces of $\Pi M \Pi + (\Id-\Pi)M(\Id-\Pi)$ with associated eigenvalue in $[1-\omega,\omega]$. Then 
\[ \Tr\big((\Id-K)\phi\big)\,\leq\, \frac{ 2 \mu + 10 \sqrt{\gamma} }{1-4\omega(1-\omega)}\;.\]
\end{lemma}

\begin{proof}
Using Jordan's lemma we find a basis of $\mH$ in which  
\begin{equation}\label{eq:n-form}
M = \oplus_j \begin{pmatrix} c_j^2 & c_js_j \\ c_js_j & s_j^2 \end{pmatrix}\quad \text{and}\quad \Pi = \oplus_j \begin{pmatrix} 1 & 0 \\ 0 & 0 \end{pmatrix}\;,
\end{equation}
where $c_j= \cos \theta_j$, $s_j=\sin \theta_j$, for some angles $\theta_j$. There may be $1$-dimensional blocks in the Jordan decomposition, but up to adding a few dimensions these can be identified with two-dimensional blocks such that $c_j^2 \in \{0,1\}$. Let $K$ be the orthogonal projection on those $2$-dimensional blocks such that $\min(c_j^2,s_j^2) \geq 1-\omega$. Note that $K$ commutes with both $M$ and $\Pi$, but not necessarily with $\phi$. It is easy to verify that this definition of $K$ coincides with the definition given in the lemma.  

Suppose first that $\gamma=0$.  Then  $\phi$ is supported on the range of $M$. For any block $j$, let $P_j$ be the projection on the block and $\alpha_j = \Tr(P_j \phi)$. It follows from the decomposition in~\eqref{eq:n-form} and the definition of $\mu$ that 
\begin{equation}\label{eq:gm-bound-0}
 \Big|\frac{1}{2}-\sum_j \,\alpha_j \,\big(c_j^4+ s_j^4\big) \Big|\,\leq\, \mu\;.
\end{equation}
Using that for $j$ such that $\min(c_j^2,s_j^2) \leq 1-\omega$ we have 
\[c_j^4 + s_j^4  \,=\, 1-2\max(c_j^2,s_j^2)\big(1-\max(c_j^2,s_j^2)\big)\,\geq\, \frac{1}{2} + \Big(\frac{1}{2} - 2\omega(1-\omega)\Big)\;,\]
and $c_j^4 + s_j^4 \geq \frac{1}{2}$ always, it follows from~\eqref{eq:gm-bound-0} that for any $\omega > \frac{1}{2}$, 
\begin{equation}\label{eq:gmbound-1}
\Tr\big((\Id-K) \phi\big)\,\leq\, \frac{2\mu}{1-4\omega(1-\omega)}\;.
\end{equation}
Next consider the case where $\gamma>0$. Assume $\Tr(M\phi)>0$, as otherwise the lemma is trivial. Let $\phi'=M\phi M/\Tr(M\phi)$. By the gentle measurement lemma (see e.g.~\cite[Lemma 9.4.1]{wilde2011classical}), 
\begin{equation}\label{eq:gm-bound}
\big\|\phi'-\phi \big\|_1 \,\leq \, 2 \sqrt{\gamma}\;.
\end{equation}
Using the definition of $\mu$, it follows that
\[\Big|\frac{1}{2} - \Tr\big(M\Pi\phi' \Pi\big)-\Tr\big(M(\Id-\Pi)\phi'(\Id-\Pi)\big)\Big|\,\leq\, \mu + 4\sqrt{\gamma}\;.\] 
Applying the same reasoning as for the case $\gamma=0$ yields an analogue of~\eqref{eq:gmbound-1}, with $\phi'$ instead of $\phi$ on the left-hand side and $\mu + 4\sqrt{\gamma}$ instead of $\mu$ on the right-hand side. 
Finally, using again~\eqref{eq:gm-bound} the same bound transfers to $\phi$ up to an additional loss of $2\sqrt{\gamma}$. 
\end{proof}

\subsection{Simulating an efficient device using a simplified device}

Recall the definitions of a simplified device (Definition~\ref{def:binary-device}) and of the overlap of a simplified device (Definition~\ref{def:overlap}). 
Recall also the definition of 
post-measurement states $\{\phi^{co}\}$ associated with a device $D=(\phi,\Pi,M)$ given in Definition~\ref{def:post-meas}, and of post-measurement states $\{(\phi')^{ctok}\}$ associated with a simplified device $D'=(\phi',\Pi',M',K)$ given in Definition~\ref{def:post-meas-simplified}. These ensembles of states provide a means to meaningfully compare a device $D$ and a simplified device $D'$. We record this in the following definition. 

\begin{definition}\label{def:simulate}
Let $D=(\phi,\Pi,M)$ be  a device and $D'=(\phi',\Pi',M',K)$ a simplified device. We say that \emph{$D'$ simulates $D$} if for every $(c,o)\in\{0,1\}^N\times \{0,1,2\}^{N}$ and $t=0^N$ the states $\phi^{co}$ and $(\phi')^{cto}$ are identical. 
\end{definition}

The following proposition shows that any efficient device can be simulated by a simplified device whose measurements generally make an angle that is bounded away from $1$. As in Lemma~\ref{lem:break}, the only assumption required on the efficient device is that it does not break  the hardcore bit property~\eqref{eq:adaptive-hardcore}.

\begin{proposition}\label{prop:change-d}
Let $D=(\phi,\Pi,M)$ be an efficient device and $\frac{1}{2} <  \omega \leq 1$. Then there is a (not necessarily efficient) simplified device $\tilde{D}=(\phi,\tilde{\Pi},\tilde{M},K)$ such that the following hold:
\begin{enumerate}
\item $\tilde{D}$ has overlap $\Delta(\tilde{D}) \leq \omega$;
\item The simplified device $\tilde{D}$ simulates the device $D$;
\item For any advice states $\phi'=\{\phi'_y\}$ that are independent from the key $k\in \mK_\mF$ (see Definition~\ref{def:lwe-ass}) it holds that 
\begin{equation}\label{eq:good-space}
\sum_{y}\Tr(  K_y^1 \phi'_y) \,\leq\, C\, \sqrt{\sum_y \Tr( \tilde{M}_y^1 \phi'_y)} + \negl(\lambda)\;,
\end{equation}
where $C>0$ is a constant depending only on $\omega$.
\end{enumerate}
\end{proposition}

\begin{proof}
For each $y\in\mY$ let  
$$\hat{M}_y = \sum_{(u,d):\, d\notin \hat{\dset}_y} M_{y}^{(u,d)}\;,\qquad M_y = \sum_{(u,d) \in V_{y,0}} M_y^{(u,d)} + \frac{1}{2} \hat{M}_y\;,$$
and for $b \in\{0,1\}$, $\Pi_y^b = \Pi_y^{(b,x_b)}$. By introducing an isometry $U_y: \mH_\reg{D} \to \mH_{\reg{D}'}$ into a larger space, we can embed $M_y$ into a projection $\ol{M}_y$ such that $M_y = U_y^\dagger \ol{M}_y U_y$. For $b\in\{0,1\}$ let $\ol{\Pi}_y^b$ be such that $U_y^\dagger \ol{\Pi}_y^b U_y = \Pi_y^b$. 
	
The device $\tilde{D}$ is defined as follows. The device first measures an $y\in\mY$ exactly as $D$ would. It then applies the isometry $U_y$. This defines the $\{\phi'_y\}$. Let $\phi' = \sum_y \phi'_y$, and note that $\Tr(\phi')=\Tr(\phi) \leq 1$. 
\begin{itemize}
\item The measurement $\{\tilde{\Pi}_y^0,\tilde{\Pi}_y^1,\tilde{\Pi}_y^2\}$ is defined as follows. The device first coherently performs the measurement $\{\ol{\Pi}_y^{(b,x)}\}$. If an outcome $(b,x) \in V_{y,1}$ is obtained the device returns $v=b$. Otherwise the device returns $v=2$. 
\item The measurement $\{\tilde{M}_y^0,\tilde{M}_y^1\}$ is defined as follows. The device first performs the measurement $\{\tilde{\Pi}_y^2,\Id-\tilde{\Pi}_y^2\}$. If the first outcome is obtained, it returns a random outcome. Otherwise, it coherently performs the measurement
 $\{ \ol{M}_y^{(u,d)}\}$. If $d\notin \hat{\dset}_y$ the device returns a random outcome. Otherwise, if $(u,d)\in V_{y,0}$ it returns a $0$, and $1$ if not. 
\item Let $K_y$ be the projection obtained by applying Lemma~\ref{lem:angles} to the projections $\Pi=\ol{\Pi}_y^0$ and $M=\ol{M}_y$ and the state 
\[\phi \,=\, \frac{(\ol{\Pi}_y^0+\ol{\Pi}_y^1) \phi'_y(\ol{\Pi}_y^0+\ol{\Pi}_y^1) }{\Tr\big((\ol{\Pi}_y^0+\ol{\Pi}_y^1) \phi'_y\big)}\;.\]
The measurement $\{K_y^0,K_y^1\}$ is defined by setting 
\[K_y^0=(\ol{\Pi}_y^0+\ol{\Pi}_y^1)K_y + (\Id - \ol{\Pi}_y^0-\ol{\Pi}_y^1)\quad\text{and}\quad K_y^1 = (\ol{\Pi}_y^0+\ol{\Pi}_y^1)(\Id-K_y)\;.\]
\end{itemize}

The first two conditions on $D'$ claimed in the lemma follow by definition. The overlap property holds by definition of $K_y^0$. For the simulation property, note that it is possible for $D'$ to further measure the post-measurement states to locally obtain an equation, or a pre-image, as $D$ would have; this guarantees that the post-measurement states of the two devices are identical in each round.

It remains to show the third item. 
It follows from computational indistinguishability of $\sigma_0$ and $\sigma_1$ shown in Lemma~\ref{lem:break} that both operators have a trace that is within negligible of each other. Using the notation introduced here, and in particular the definition of $\ol{M}_y$, this implies that the difference
\[\Big| \sum_{b\in\{0,1\}} \Tr\big(\ol{M}_y\ol{\Pi}_y^b\phi'_y\ol{\Pi}_y^b\big) -  \sum_{b\in\{0,1\}} \Tr\big((\Id-\ol{M}_y)\ol{\Pi}_y^b\phi'_y\ol{\Pi}_y^b\big)\Big|\;\]
is negligible. Since the two expressions sum to $\Tr((\Id-\ol{\Pi}_y^2)\phi'_y)$, it follows that, letting 
\begin{equation}\label{eq:def-tilde-aa}
\tilde{\phi}_y \,=\,\frac{(\Id-\ol{\Pi}_y^2)\phi'_y(\Id-\ol{\Pi}_y^2)}{\Tr((\Id-\ol{\Pi}_y^2)\phi'_y)}\;,
\end{equation}
we get that
\[\Tr(\ol{M}_y\ol{\Pi}_y^0\tilde{\phi}_y\ol{\Pi}_y^0) + \Tr(\ol{M}_y\ol{\Pi}_y^1\tilde{\phi}_y\ol{\Pi}_y^1)\]
 is within negligible of $\frac{1}{2}$. To conclude we apply  Lemma~\ref{lem:angles} to the operators $\Pi = \ol{\Pi}^0$ and $M = \ol{M}_y$. The conclusion of the lemma gives that 
\begin{equation}\label{eq:good-space-a}
\Tr\big(  (\Id-K_y) \tilde{\phi}_y\big) \,\leq\, C\, \sqrt{\Tr( (\Id-\tilde{M}_y^0) \tilde{\phi}_y)} + \negl(\lambda)\;,
\end{equation}
for some universal constant $C$ (depending on $\omega$). Multiplying both sides of~\eqref{eq:good-space-a} by  
\[ p_y \,=\, \frac{ \Tr((\Id-\ol{\Pi}_y^2)\phi'_y)}{\sum_y \Tr((\Id-\ol{\Pi}_y^2)\phi'_y)}\;,\]
summing over $y$ and applying Jensen's inequality gives 
\begin{equation}\label{eq:good-space-aa}
\sum_y \Tr\big(  (\Id-\ol{\Pi}_y^2)(\Id-K_y)(\Id-\ol{\Pi}_y^2) {\phi'}_y\big) \,\leq\, C\, \sqrt{\sum_y \Tr\big( (\Id-\ol{\Pi}_y^2)(\Id-\tilde{M}_y^0)(\Id-\ol{\Pi}_y^2) {\phi'}_y\big)} + \negl(\lambda)\;,
\end{equation}
where we also used $\sum_y \Tr((\Id-\ol{\Pi}_y^2)\phi'_y) \leq 1$ since $\Tr(\sum_y \phi'_y )\leq 1$. Using that by definition $(\Id-\ol{\Pi}_y^2)(\Id-K_y) = K_y^1$ and $(\Id-\tilde{M}_y^0)$ commutes with $(\Id-\ol{\Pi}_y^2)$ gives~\eqref{eq:good-space}.
\end{proof}

\section{Accumulating randomness across multiple rounds}
\label{sec:multi-round}

To analyze the randomness generated by a device in the randomness expansion protocol we proceed in two steps. First, we show that the randomness generated by the device can be related to the randomness generated by the simplified device $\tilde{D}$ that is associated to it by Proposition~\ref{prop:change-d}, when it is used as a device in the simplified protocol, Protocol~2. This is done in Section~\ref{sec:reduction}. Then, in Section~\ref{sec:simplified} we analyze the randomness generated in a single round of the simplified protocol, and  in Section~\ref{sec:randomness} we analyze multiple rounds of the protocol. 

\subsection{Reduction to the simplified protocol}
\label{sec:reduction}

Let $D = (\phi,\Pi,M,K)$ be a simplified device. 
The main difference between the behavior of the simplified device and the original device it is derived from is that the simplified device (sometimes) performs an additional projective measurement $\{K^0,K^1\}$, in addition to the ``equation'' measurement $\{M^0,M^1\}$. (Recall that in protocol $2$, the device performs the measurement whenever the verifier sends a challenge bit $T=1$, which happens with probability $\Pr(T=1)=\kappa$ in the test rounds.) Informally, the measurement $\{K^0,K^1\}$ is used to detect if the state of the device is in the ``good subspace'' in which the device's measurements generate randomness, as measured by the overlap $\Delta(D)$. This measurement is a conceptual tool that is not performed as part of the real protocol, but is included in the simplified protocol to facilitate the randomness generation analysis. 

In order to lift the analysis of the randomness generated in Protocol 2 to Protocol~1 we will show that, in most test rounds of Protocol~1, the state of the device lies largely within the ``good subspace'' $K=K^0$. 
Recall the definition of the states $\{\phi^{ctok}\}$ associated with the simplified device in Definition~\ref{def:post-meas-simplified}. Let 
\begin{equation}\label{eq:k-expansion}
\ket{\phi^{co}} \,=\, \sum_{k} \,\ket{\phi^{ctok}}\, =\, \sum_k \, P_{ct}^{ok} \,\ket{\phi}\;,
\end{equation}
 where $P_{ct}^{ok}$ is notation for the operator that corresponds to applying the device's (projective) measurement operators $M$, $\Pi$ and $K$ indicated by $c$ and $t$ respectively, and obtaining the sequence of outcomes $o$ and $k$ respectively. The fact that $\ket{\phi^{co}}$ does not depend on $t$ is justified by the fact that $\{K,\Id-K\}$ is a projective measurement.  

Our goal is to bound the contribution to~\eqref{eq:k-expansion} of terms $P_{ct}^{ok} \ket{\phi}$ that correspond to a large fraction of $(\Id-K)$ (``bad subspace'') outcomes, i.e. such that the Hamming weight $|k|$ of the string $k$ is large. Establishing the right bound is made delicate by the possibility of interference between the branches. We first state and prove a general lemma, and then show how the lemma can be applied in our context. 

\begin{lemma}\label{lem:branch-bound}
Let $n$ be an integer, $0<\kappa<1$, and $T=(T_1,\ldots,T_n)$ a sequence of independent Bernoulli random variables such that for any $t\in\{0,1\}^n$, $\Pr(T=t)= \kappa(t) = \prod_i \kappa^{t_i} (1-\kappa)^{1-t_i}$. 
Let $M=(M_1,\ldots,M_n)$ and $K=(K_1,\ldots,K_{|T|})$ be sequences of random variables over $\{0,1\}$ that may be correlated between themselves and with $T$ but satisfy that for each $i\geq 1$, $M_i$ and $T_i$ are independent conditioned on $(T,M,K)_{<i}$. (For an integer $i\in\{1,\ldots,n\}$ we write $(T,M,K)_{<i}$ for the triple formed by the length-$(i-1)$ prefixes of $T$ and $M$, and the length-$|T_{<i}|$ prefix of $K$.\footnote{Recall that we write $|T|$ for the Hamming weight of the string $T$. Here, we think of each $K_j$ as a random variable that is correlated with the random variable $M_i$, where $i$ is the index of the $j$-th non-zero entry of $T$.})

 Assume that there is a monotone concave function $g:[0,1]\to[0,1]$ such that $g(0)=0$, $g(x)\geq x$ for all $x\in[0,1]$, and for any $i\in\{1,\ldots,n\}$ and any sequences $t,m\in\{0,1\}^{i-1}$ and $k\in\{0,1\}^{|t|}$ it holds that 
\begin{align}
\Pr\big(K_{|t|+1}=1 \,\big|\,(T,M,K)_{<i}&=(t,m,k),\,T_{i}=1\big)\notag\\
&\leq g\Big( \Pr\big( M_{i}=0 \,\big|\,(T,M,K)_{<i}=(t,m,k),\,T_{i}=1\big)\Big)\;.\label{eq:assumption-k}
\end{align}
Then for any $0<\eta<1$ there are $\kappa_0,C_0>0$ such that for all $0\leq \kappa \leq \kappa_0$ and $\gamma = \kappa^{3/2}$, for all integer $n\geq 1$, 
\begin{equation}\label{eq:branch-concl}
\sum_{t\in\{0,1\}^n}\, \kappa(t) \sum_{m:\, |m| \geq (1-\gamma)n} \,\Big( \sum_{k:\,|k|> \eta\kappa n} \sqrt{\Pr\big( (T,M,K)=(t,m,k)\big)} \Big)^2\,\leq\, C_0 \,2^{-\kappa n}   \;.
\end{equation}
\end{lemma}

Intuitively, the lemma holds because the condition $|m|\geq (1-\gamma) n$ ensures that the outcome $M_i=0$ is fairly unlikely, in which case~\eqref{eq:assumption-k} implies that whenever $T_{i}=1$ the outcome $K_{j}=1$, where $j$ is the number of nonzero entries of $T$ in indices less or equal to $i$, should also be unlikely. The proof is made a little difficult by the square roots, whose presence is motivated by the application to norms of quantum states detailed later. Nevertheless, to understand the statement of the lemma it may be useful to consider the case when all $M_i$ (resp. $K_j$) are independent and identically distributed, and the square root and the square are not present. In this case, the lemma reduces to showing that if 
\begin{equation}\label{eq:def-gb}
\mathcal{G}\, =\, \Big\{m\in\{0,1\}^n:\,|m|\geq (1-\gamma)n\Big\}\;,\quad \mathcal{B}_t = \Big\{(m,k):\, m\in\{0,1\}^{n},\,k\in\{0,1\}^{|t|},\,|k| \geq \eta\kappa n\Big\}\;,
\end{equation}
for $t\in \{0,1\}^{n}$,
then 
\[ \sum_{t\in \{0,1\}^n} \,\kappa(t) \, \Pr\big(\mathcal{G}\wedge\mathcal{B}_t\big) \,\leq\, C_0 2^{-\kappa n}\;.\] 
Note that here as in the remainder of the section, for a set $\mathcal{X}$ we often slightly abuse notation and write $\Pr(\mathcal{X})$ for $\Pr(X\in\mX)$ whenever it is clear from context which random variable is referred to.

 As a first step, note that we may safely assume that $\Pr(M_i=0) \leq \gamma  + C \sqrt{\kappa}$ for some constant $C$, as otherwise by a Chernoff bound $\Pr(\mathcal{G}) \leq e^{-\Omega(C\kappa n)} \leq C_0 2^{-\kappa n}$ provided $C$ is large enough. 
Then, using~\eqref{eq:assumption-k} it follows that $\Pr(K_j=1)\leq g(\gamma+C\sqrt{\kappa})$, so that applying the Prohorov bound~\cite{prokhorov1959extremal} (see Theorem~\ref{theorem:fan} below for an extension to marginales) with $\xi$ there equal to $K_i$ here, $n$ there equal to $|t|$ here, $t=\eta$ and $v^2 = g(\gamma+C\sqrt{\kappa})$, we get that for any $t$,
\[\Pr(\mathcal{B}_t) \,\leq\, e^{-\frac{\eta}{2}\arcsinh\big(\frac{\eta}{2g(\gamma+C\sqrt{\kappa})}\big)|t|}\;.\]
Using the assumption that $g$ is monotone non-decreasing such that $g(0)=0$, by choosing $\kappa_0$ small enough with respect to $\eta$ one can ensure that the exponent 
\[ \frac{\eta}{2}\,\arcsinh\Big(\frac{\eta}{2g(\gamma+C\sqrt{\kappa})}\Big)\]
is an arbitrarily large constant $C_1$. Then 
\begin{align*}
 \sum_{t\in \{0,1\}^n} \,\kappa(t) \,\Pr(\mathcal{B}_t)
&\leq 
\sum_{t\in \{0,1\}^n} \,\kappa(t) \,2^{-C_1|t|}\\
&= \big(1+ (2^{-C_1}-1)\kappa\big)^n \\
&\leq C_0\, 2^{-\kappa n}\;,
\end{align*}
where the second line uses the expression for the moment generating function for the binomial distribution and the last inequality uses $2<e$ and holds for $C_1$ large enough. This completes the argument. To extend it to the general case, we use two tail bounds for martingales that replace the use of the Chernoff bound and the Prohorov bound respectively. The first is Azuma's inequality. 

\begin{theorem}[Azuma's inequality]\label{theorem:azuma}
Let $(\xi_i,\mathcal{F}_i)_{0\leq i \leq n}$ be a martingale difference sequence such that $\xi_0=0$ and $|\xi_i| \leq 1$ for each $i\in\{1,\ldots,n\}$. 
Then for any $t\geq 0$,
\[ \Pr\Big( \Big|\sum_{i=1}^n\,\xi_i \Big|\,\geq\, t n \Big) \,\leq\, 2e^{-\frac{t^2}{2}n }\;.\]
\end{theorem}

The second is 
a version of the Prohorov bound for martingales.

\begin{theorem}[Corollary 2.2 in \cite{fan2012hoeffding}]\label{theorem:fan}
Let $(\xi_i,\mathcal{F}_i)_{0\leq i \leq n}$ be a martingale difference sequence such that $\xi_0=0$ and $|\xi_i| \leq 1$ for each $i\in\{1,\ldots,n\}$. Let 
\[ X_n \,=\, \sum_{i=1}^n \,\xi_i\quad\text{and}\quad\langle X \rangle_n \,=\, \sum_{i=1}^n \,\Es{}\big[ \xi_i^2 \big|\, \mathcal{F}_{i-1}\big]\;.\]
Then for any $t\geq 0$ and $v>0$,
\[ \Pr\Big( \big|X_n\big| \geq t n \,\text{ and }\, \langle X\rangle_n \leq v^2 n \Big) \,\leq\, e^{-\frac{t}{2}\arcsinh\big(\frac{t}{2v^2}\big)n }\;.\]
\end{theorem}

We give the proof of Lemma~\ref{lem:branch-bound}.

\begin{proof}[Proof of Lemma~\ref{lem:branch-bound}]
We reduce the proof of~\eqref{eq:branch-concl} to a sequence of martingale tail bounds. Define a filtration $(\mathcal{F}_1,\ldots,\mathcal{F}_i,\ldots,\mathcal{F}_n)$ where $\mathcal{F}_i$ is the $\sigma$-algebra generated by $(M,T,K)_i$. Let $\mathcal{F}_{< i}= \cap_{j< i} \mathcal{F}_j$. Recall the definition of the event $\mathcal{G}$ in~\eqref{eq:def-gb}.
The proof proceeds in 3 steps, that we each formulate as a separate claim. 

\begin{claim}[First step: conditional expectations of $M$]\label{claim:branch-1}
Let $C$ be a sufficiently large universal constant ($C$ is specified in~\eqref{eq:z-bound-1b} in the proof). Let $\delta'_1 =  \gamma + 2\sqrt{(C+1)\kappa}$. Let 
\begin{equation}\label{eq:mart-bp}
\mathcal{B}' \,=\, \Big\{(t,m,k):\, \sum_{i=1}^n \,\Es{}\big[M_i\big|\,(T,M,K)_{<i}=(t,m,k)_{<i}\big] \,\leq\, (1-\delta'_1) n\Big\}\;.
\end{equation}
Then it holds that 
\begin{equation}
\sum_{m \in \mathcal{G}}\sum_t \,\kappa(t) \Big(\sum_{k:(t,m,k)\in\mathcal{B}'}\sqrt{ \Pr\big((M,K)=(m,k)|T=t\big)}\Big)^2 \,\leq\,2\cdot 2^{-\kappa n}\;. \label{eq:mart-1}
\end{equation}
\end{claim}

\begin{proof}
For $i\in\{1,\ldots,n\}$ let $Z_i = M_i - \Es{}[M_i|\mathcal{F}_{<i}]$ and $W_i = Z_1+\cdots+Z_i$. Then by definition the sequence $(W_1,\ldots,W_n)$ is a martingale. Moreover, for any $i\geq 2$, it holds that $|W_i-W_{i-1}|=|Z_i|\leq 1$ since $M_i\in\{0,1\}$. Applying Azuma's inequality (Theorem~\ref{theorem:azuma}) it follows that for any $\delta_1 >0$,
\begin{equation}\label{eq:z-bound-1}
 \Pr\Big( \Big|\sum_{i=1}^n Z_i \Big|\geq \delta_1 n  \Big)\,\leq\, 2\,e^{-\frac{\delta_1^2}{2} n}\;.
\end{equation}
Let $\delta_1$ be large enough such that the right-hand side of~\eqref{eq:z-bound-1} is less than $2^{-(C+1)\kappa n}$, for some constant $C$ to be determined below. Let $\delta'_1 = \delta_1+\gamma$ and $\mathcal{B}'$ as in~\eqref{eq:mart-bp}.
Then by the Cauchy-Schwarz inequality 
\begin{align}
 \sum_{m \in \mathcal{G}}\sum_t \,\kappa(t) &\Big(\sum_{k:(t,m,k)\in\mathcal{B}'}\sqrt{ \Pr\big((M,K)=(m,k)|T=t\big)}\Big)^2 \notag\\
&\leq\sum_{m \in \mathcal{G}}\sum_t \,\kappa(t) \Big(\sum_{k:(t,m,k)\in\mathcal{B}'}\Pr\big((M,K)=(m,k)|T=t\big)\Big) \Big(
\sum_{k:(t,m,k)\in\mathcal{B}'} 1\Big) \notag\\
&\leq \sum_{m \in \mathcal{G}}\sum_t \,\kappa(t) 2^{|t|} \Big(\sum_{k:(t,m,k)\in\mathcal{B}'}\Pr\big((M,K)=(m,k)|T=t\big)\Big) \;,\label{eq:z-bound-1aa}
\end{align}
where the last inequality follows since by definition the string $k$ ranges over $\{0,1\}^{|t|}$. 
Let $\mathcal{T}$ be the event
\[\mathcal{T} = \big\{(t,m,k):\, m\in\mathcal{G}\,\wedge\, (t,m,k)\in\mathcal{B}'\big\}\;.\]
Then for $(t,m,k)\in \mathcal{T}$ it holds that $|m|\geq(1-\gamma)n$ and
\begin{align*}
\sum_i z_i &= \sum_i m_i - \Es{}[M_i|(T,M,K)_{<i} = (t,m,k)_{<i}] \\
&\geq (1-\gamma) n -(1-\delta'_1)n\\
 &= (\delta'_1-\gamma)n\;.
\end{align*}
Thus it follows from~\eqref{eq:z-bound-1} and our choice of $\delta_1$ that
\begin{equation}\label{eq:z-bound-1c}
\Pr\big(\mathcal{T}\big)\,=\,\sum_{m \in \mathcal{G}}\sum_{t,k:\,(t,m,k)\in\mathcal{B}'} \Pr\big((T,M,K)=(t,m,k)\big) \leq 2^{-(C+1)\kappa n}\;.
\end{equation}
Finally, note that by the Chernoff bound, for $C$ large enough, 
\begin{equation}\label{eq:z-bound-1b}
\sum_{|t|\geq C\kappa n} \kappa(t)2^{|t|} \leq 2^{-\kappa n}\;.
\end{equation}
Fix $C$ so that~\eqref{eq:z-bound-1b} holds. 
Then starting from~\eqref{eq:z-bound-1aa} we get
\begin{align}
 \sum_{m \in \mathcal{G}}\sum_t \,\kappa(t) &\Big(\sum_{k:(t,m,k)\in\mathcal{B}'}\sqrt{ \Pr\big((M,K)=(m,k)|T=t\big)}\Big)^2 \notag\\
&\leq  2^{C\kappa n}\Big(\sum_{m \in \mathcal{G}}\sum_{t,k:\,(t,m,k)\in\mathcal{B}'} \Pr\big((T,M,K)=(t,m,k)\big)\Big) + 2^{-\kappa n}\notag\\
&\leq 2\cdot 2^{-\kappa n}\;, \notag
\end{align}
where the first inequality uses~\eqref{eq:z-bound-1b} and the second uses~\eqref{eq:z-bound-1c}.
\end{proof}

\begin{claim}[Second step: conditional expectations of $T(1-M)$]\label{claim:branch-2}
For any $\delta'_2=\gamma/\kappa + \delta_2$, where $\delta_2$ is sufficiently large compared to $\delta'_1$, letting
\begin{equation}\label{eq:mart-bpp}
\mathcal{B}'' = \Big\{(t,m,k)\notin \mathcal{B}':\, \sum_i \,\Es{}\big[T_i(1-M_i)\big|\,(T,M,K)_{<i}=(t,m,k)_{<i}\big] \,\geq\, \delta'_2 \kappa n\Big\}
\end{equation}
we have that 
\begin{align}
 \sum_{m \in \mathcal{G}}\sum_{t}\,\kappa(t)\Big(\sum_{k:(t,m,k)\in\mathcal{B}''}\sqrt{ \Pr\big((M,K)=(m,k)|T=t\big)}\Big)^2 
&\leq 2^{-\kappa n}\;. \label{eq:mart-2}
\end{align}
\end{claim}

\begin{proof}
For $i\in\{1,\ldots,n\}$ let $Z'_i = T_i(1-M_i) - \Es{}[T_i(1-M_i)|\mathcal{F}_{<i},\overline{\mB'}]$ and $W'_i = Z'_1+\cdots+Z'_i$. Then the sequence $(W'_1,\ldots,W'_n)$ is a martingale such that $|W'_i-W'_{i-1}|\leq 1$. Let 
\[v_{Z'}^2 = \sum_i \Es{}[|Z'_{i}|^2| \mathcal{F}_{<i},\overline{\mB'}]\;.\]
For $(t,m,k)\notin \mathcal{B}'$, using that by assumption $T_i$ is independent from $M_i$ conditioned on $\mathcal{F}_{<i}$ and $\Es{}[T_i|\mathcal{F}_{<i}]=\kappa$ it holds that 
\begin{align*}
v_{Z'}^2 &\leq \sum_i \Es{}\big[ T_i(1-M_i)|\mathcal{F}_{<i},\overline{\mB'}\big]\\
&= \kappa \sum_i  \Es{}\big[(1-M_i)|\mathcal{F}_{<i},\overline{\mB'}\big]
\leq \delta'_1 \kappa n\;.
\end{align*}
 Let $v^2=\delta'_1\kappa n$. Applying Theorem~\ref{theorem:fan}, for any $\delta_2>0$, 
\begin{equation}\label{eq:z-bound-1a}
\Pr\Big( \Big|\sum Z'_i \Big|\geq \delta_2 \kappa n \;\wedge\; v_{Z'}^2 \leq v^2 n  \Big)\,\leq\, e^{-\frac{1}{2}\delta_2\kappa  \arcsinh\big(\frac{\delta_2}{2\delta'_1}\big)n }\;.
\end{equation}
Assume $\delta_2$ chosen sufficiently large compared to ${\delta'_1}$ so that the right-hand side in~\eqref{eq:z-bound-1a} is less than $2^{-(C+1)\kappa n}$. Let $\delta'_2 = \delta_2 + \gamma/\kappa$ and $\mathcal{B}''$ as in~\eqref{eq:mart-bpp}.
Then proceeding similarly to the end of the proof of Claim~\ref{claim:branch-1} we get~\eqref{eq:mart-2}.
\end{proof}

\begin{claim}[Third step: conditional expectations of $T(1-M)K$]\label{claim:mart-3}
Let  
\[\mathcal{B} = \Big\{(t,m,k):\, (t,m)\in\{0,1\}^{2n},\,k\in\{0,1\}^{|t|},\,|k| \geq \eta\kappa n\Big\}\;,\]
and 
\begin{equation}\label{eq:mart-bppp}
\mathcal{B}'''= \ol{\mathcal{B}''\cup\mathcal{B}'}\cap \mathcal{B}\;.
\end{equation}
Assume that $\eta = \gamma/\kappa + \delta_3$ where $\delta_3$ is sufficiently large compared to ${g(\delta'_2)}$. 
Then
\begin{align}
 \sum_{m \in \mathcal{G}}\sum_{t}\,\kappa(t)\Big(\sum_{k:(t,m,k)\in\mathcal{B}'''}\sqrt{ \Pr\big((M,K)=(m,k)|T=t\big)}\Big)^2 
&\leq 2^{-\kappa n}\;. \label{eq:mart-3}
\end{align}
\end{claim}

\begin{proof}
We have 
\begin{align}
\sum_i  \Es{}\big[&T_i(1-M_i) K_{|T_{<i}|+1} \big|\,(T,M,K)_{<i}=(t,m,k)_{<i}\big]\notag\\
&\leq \sum_i  \Es{}\big[T_i K_{|T_{<i}|+1} \big|\,(T,M,K)_{<i}=(t,m,k)_{<i}\big]\notag\\
&= \sum_i  \Pr\big( K_{|T_{<i}|+1} = 1 \big|\,(T,M,K)_{<i}=(t,m,k)_{<i},\, T_i=1\big) \Pr\big(T_i=1 \big|\,(T,M,K)_{<i}=(t,m,k)_{<i}\big)\notag\\
&\leq \kappa \sum_i g\Big(  \Pr\big( M_{i} = 0 \big|\,(T,M,K)_{<i}=(t,m,k)_{<i},\, T_i=1\big)\Big) \notag\\
&\leq \kappa n g\Big( \frac{1}{n}\sum_i  \Pr\big( M_{i} = 0 \big|\,(T,M,K)_{<i}=(t,m,k)_{<i},\, T_i=1\big)\Big) \;.\label{eq:z-bound-2b}
\end{align}
Here for the second line we used $0\leq (1-M_i)\leq 1$, the third follows by an application of Bayes' rule, for the fourth line we used assumption~\eqref{eq:assumption-k} and the fact that for all $i$, 
\begin{equation}\label{eq:z-bound-2a}
\Pr\big(T_i=1 \big|\,(T,M,K)_{<i}=(t,m,k)_{<i}\big)\,=\,\kappa\;,
\end{equation}
 and for the last line we used concavity of $g$. For any $(t,m,k) \notin (\mathcal{B}''\cup\mathcal{B}')$ it holds that
\begin{align*}
  \sum_i \,\Es{}\big[(1-M_i)\big|\,(T,M,K)_{<i}=(t,m,k)_{<i}\big] 
  &= \frac{1}{\kappa} \sum_i \,\Es{}\big[T_i(1-M_i)\big|\,(T,M,K)_{<i}=(t,m,k)_{<i}\big]\\
 &\leq \delta'_2 n\;,
  \end{align*}
  where the equality uses~\eqref{eq:z-bound-2a} and the fact that $T_i$ and $M_i$ are independent conditioned on $(T,M,K)_{<i}$, and the inequality uses the definition of $\mathcal{B}''$. Combined with~\eqref{eq:z-bound-2b} we get that for any $(t,m,k) \notin (\mathcal{B}''\cup\mathcal{B}')$,
\begin{equation}\label{eq:z-bound-2}
 \sum_i  \Es{}\big[T_i(1-M_i) K_{|T_{<i}|+1} \big|\,(T,M,K)_{<i}=(t,m,k)_{<i}\big] \,\leq\, g(\delta'_2) \kappa n \;.
\end{equation}
For $i\in\{1,\ldots,n\}$ let 
\[Z''_i = T_i(1-M_i)K_i - \Es{}\big[T_i(1-M_i)K_i|\mathcal{F}_{<i},\,\overline{\mathcal{B}''\cup\mathcal{B}'}\big]\]
 and $W''_i = Z''_1+\cdots+Z''_i$.
 Then the sequence $(W''_1,\ldots,W''_n)$ is a martingale such that $|W''_i-W''_{i-1}|\leq 1$ and by~\eqref{eq:z-bound-2}, 
\[ v_{Z''}^2 = \sum_i \Es{}\big[|Z''_{i}|^2| \mathcal{F}_{<i},\,\overline{\mathcal{B}''\cup\mathcal{B}'}\big]\leq g(\delta'_2) \kappa n\;.\]
 Applying  Theorem~\ref{theorem:fan}, for any $\delta_3>0$ it holds that
$$ \Pr\Big( \Big|\sum Z''_i \Big|\geq \delta_3 \kappa n \;\wedge\; \ol{\mathcal{B''}\cup\mathcal{B}'}  \Big)\,\leq\, e^{-\frac{1}{2} \delta_3\kappa  \arcsinh\big(\frac{\delta_3}{2g(\delta'_2)}\big)n}\;.$$
By choosing $\delta_3$ sufficiently large  compared to ${g(\delta'_2)}$ the right-hand side can be made less than $2^{-(C+1)\kappa n}$. Assume further that $\delta_3 + \gamma/\kappa\leq\eta$. Let $\mathcal{B}'''$ be as in~\eqref{eq:mart-bppp}. Then~\eqref{eq:mart-3} follows similarly to the proof of~\eqref{eq:mart-1} and~\eqref{eq:mart-2} in Claim~\ref{claim:branch-1} and Claim~\ref{claim:branch-2} respectively. 
\end{proof}

The lemma follows by combining Claim~\ref{claim:branch-1}, Claim~\ref{claim:branch-2} and Claim~\ref{claim:mart-3} with the triangle inequality.
\end{proof}

Recall the definition of the states $\ket{\phi^{ctok}}$ in Definition~\ref{def:post-meas-simplified}. 
For a parameter $\eta>0$ and any $t\in\{0,1\}^N$ let 
\begin{equation}\label{eq:def-phicto}
\ket{\ol{\phi}^{cto}} \,=\, \sum_{k: |k |\leq \eta \kappa q N} \ket{\phi^{ctok}}\;,
\end{equation}
 and $\ol{\phi}^{cto}$ the sub-normalized density obtained by taking the partial trace of $\ket{\ol{\phi}^{cto}}$ over register $\reg{D}$. 

\begin{corollary}\label{cor:mart}
Let $D = (\phi,\Pi,M,K)$ be a simplified device such that condition~\eqref{eq:good-space} from Proposition~\ref{prop:change-d} holds. Then for any $0<\eta<1$ there is a $\kappa_0>0$ such that for all $0\leq \kappa \leq \kappa_0$ and $\gamma = \kappa^{3/2}$, 
\begin{equation}\label{eq:mart-concl}
\sum_{g,c\in\{0,1\}^N} q(g,c) \sum_{t\in\{0,1\}^{N-|g|}} \kappa(t) \sum_{o: \,(g,c,o)\in\Acc}\, \big\| \phi^{co} - \ol{\phi}^{cto} \big\|_1 \,=\, O\big(2^{-\kappa q N}\big)\;.
\end{equation}
\end{corollary}

\begin{proof}
We apply Lemma~\ref{lem:branch-bound}. Fix $g,c\in\{0,1\}^N$ and let $n=|\{i:c_i=0\}|$. Let $T_1,\ldots,T_n$ be independent Bernoulli random variables distributed as in the statement of Lemma~\ref{lem:branch-bound}. Let $M$ and $K$ be distributed as the measurement outcomes associated with the measurements $\{\Id-M^0,\Id-M^1\}$ and $\{K^0,K^1\}$ made by the device in those rounds $i\in\{0,\ldots,N\}$ such that $c_i=0$. Note that this is well-defined since the two measurements are required to commute by Definition~\ref{def:binary-device}. Moreover, with this choice the assumption that $M_i$ and $T_i$ are independent conditioned on the past holds (in contrast $K_{|T|_{<i}+1}$ is correlated with $T_i$ and with $M_i$). 

 Using~\eqref{eq:good-space} from Proposition~\ref{prop:change-d} it follows that these random variables satisfy the assumptions of Lemma~\ref{lem:branch-bound} for a choice of the function $g(x)=C\sqrt{x}$, for a large enough constant $C$. The conclusion~\eqref{eq:branch-concl} of the lemma gives~\eqref{eq:mart-concl}.
\end{proof}

We conclude with a lemma that relates the randomness in the states $\ol{\phi}^{cto}$ to randomness in the states ${\phi}^{ctok}$, for $k$ such that $|k| \leq \eta \kappa q N$, as these are the post-measurement states associated with the simplified device in Protocol~2. The lemma relies on the following variant of the Cauchy-Schwarz inequality. 

\begin{lemma}\label{lem:matrix-cs}
Let $\ell\geq 1$ be an integer and $\ket{v_1},\ldots,\ket{v_\ell}$ arbitrary vectors in $\C^{d}$. Then
\[ \Big(\sum_{i=1}^\ell \ket{v_i}\Big) \Big(\sum_{i=1}^\ell \ket{v_i}\Big)^\dagger\,\leq\,\ell \,\sum_{i=1}^\ell \ket{v_i}\bra{v_i}\;.\]
\end{lemma}

\begin{proof}
Taking the overlap with an arbitrary unit vector $\ket{x}$, 
the claimed inequality is equivalent to showing 
\[ \Big|\sum_{i=1}^\ell \bra{x} v_i \rangle \Big|^2 \,\leq\, \ell \sum_{i=1}^\ell \big|\bra{x} v_i\rangle\big|^2\;.\]
This follows from the Cauchy-Schwarz inequality appled to the sequences $(1,\ldots,1)$ and $(\bra{x} v_1 \rangle ,\ldots,\bra{x} v_\ell \rangle)$.  
\end{proof}

Using the lemma, we show the following. 

\begin{lemma}\label{lem:mart-ub}
Let $D= (\phi,\Pi,M,K)$ be a simplified device, and $\ol{\phi}^{cto}$ the ensemble of states associated with $D$ as described in~\eqref{eq:def-phicto}. Then 
\begin{align*}
 \sum_{g,c\in\{0,1\}^N} q(g,c)\sum_{\substack{t\in\{0,1\}^{N-|g|}\\|t|\leq 2\kappa q N}} &\kappa(t) \sum_{o:\,(g,c,o)\in\Acc} \,\big\langle \ol{\phi}^{cto} \big\rangle_{1+\eps} \\
&\leq 2^{O(H(\eta))\kappa qN}  \sum_{g,c\in\{0,1\}^N} q(g,c)\sum_{t\in\{0,1\}^{N-|g|}}\kappa(t) \sum_{o,k:\,(g,c,t,o,k)\in\Acc_2} \, \big\langle {\phi}^{ctok}\big\rangle_{1+\eps}\;,
\end{align*}
where $\Acc_2$ denotes the set of transcripts that are accepted by the verifier in Protocol~2. 
\end{lemma}

\begin{proof}
From the definition of $\ol{\phi}^{cto}$ in~\eqref{eq:def-phicto}, applying Lemma~\ref{lem:matrix-cs} to the vectors $\ket{\phi^{ctok}}$ and taking the partial trace over $E$ we deduce that
\begin{equation}\label{eq:mart-ub-1}
 \ol{\phi}^{cto}  \,\leq\, {\kappa q N\choose\leq \eta \kappa q N} \sum_{k:\,|k|\leq \eta\kappa q N} \, \ol{\phi}^{ctok} \;,
\end{equation}
where ${\kappa q N\choose\leq \eta \kappa q N}$ denotes the number of sequences $k\in\{0,1\}^{|t|}$ such that $|k| \leq \eta \kappa q N$. Using standard tail bounds for the binomial distribution, this is at most $2^{O(H(\eta))\kappa q N}$. Applying the operator monotone function $\langle \cdot \rangle_{1+\eps}$ on both sides of~\eqref{eq:mart-ub-1} and using the approximate linearity~\eqref{eq:approx-lin} we obtain
\begin{equation*}
 \big\langle\tilde{\phi}^{cto}\big\rangle_{1+\eps}   \,\leq\, 2^{O(H(\eta))\kappa q N} \sum_{k:\,|k|\leq \eta\kappa q N} \, \big\langle\tilde{\phi}^{ctok}\big\rangle_{1+\eps}  \;,
\end{equation*}
where the factors $(1+O(\eps))$ from the approximate linearity got absorbed in the prefactor. 
To conclude the bound claimed in the lemma, note that the conditions that $(g,c,o)\in\Acc$ and $|k| \leq \eta\kappa qN$ imply $(g,c,t,o,k)\in\Acc_2$.
\end{proof}

\subsection{Randomness accumulation in the simplified protocol}
\label{sec:simplified}

In this section we consider the behavior of a simplified device $D=(\phi,\Pi,M,K)$ in a single round of Protocol~2. The following lemma shows that, provided the device has overlap $\Delta(D)$ bounded away from $1$, then if the state $\phi$ of the device has high overlap with the projection operator $M^1$,  performing a measurement of $\{\Pi^0,\Pi^1,\Pi^2\}$ on $\phi$ necessarily perturbs the state (hence generates randomness). The proof is based on a ``measurement-disturbance trade-off'' from~\cite{miller2017universal}, itself a consequence of uniform convexity for certain matrix $p$-norms. 

\begin{lemma}\label{lem:ms-uncertainty}
Let $D = (\phi,\Pi,M,K)$ be a simplified device with overlap $\Delta(D)\leq \omega$, for some $\omega<1$. Let  $0\leq \eps \leq \frac{1}{2}$ and 
\begin{equation}\label{eq:game-operator}
t = \frac{\langle \phi_{G} \rangle_{1+\eps} }{\langle \phi \rangle_{1+\eps}}\;,\qquad\text{where}\quad G \,=\, \frac{1}{2} \big(\Pi^0 + \Pi^1\big) + \frac{1}{2} M^1K^0\quad\text{and}\quad \phi_{G} =  \sqrt{G}\phi\sqrt{G}\;.
\end{equation}
Then 
$$\frac{ \langle \phi_{1}^0 \rangle_{1+\eps} + \langle \phi_{1}^1 \rangle_{1+\eps} + \langle \phi_{1}^2 \rangle_{1+\eps}}{\langle \phi \rangle_{1+\eps}} \,\leq\, 2^{-\eps \lambda_\omega(t)} +O(\eps^2)\;,$$
where the post-measurement states $\phi_{1}^v$, $v\in\{0,1,2\}$, are introduced in~\eqref{eq:def-pm}, and
\begin{equation}\label{eq:def-lambda}
\lambda_\omega(t) = \log(e)\Big(t-\frac{1}{2}-\frac{\omega}{2}\Big)^2\;
\end{equation} 
if $t\geq \frac{1}{2}+\frac{\omega}{2}$, and $0$ otherwise. 
\end{lemma}

\begin{proof}
The proof uses ideas from~\cite{miller2017universal}. Let $\phi$ be as in the lemma and $\phi' = \sum_v \Pi^v\phi\Pi^v$.
 Then
\begin{align*}
\big\langle \sum_v  \sqrt{G}\Pi^v \phi \Pi^v \sqrt{G} \big\rangle_{1+\eps} 
&\leq  \sum_v \langle \phi^{1/2} \Pi^v G\Pi^v \phi^{1/2} \rangle_{1+\eps}+ O(\eps) \\
&\leq \Big(\frac{1}{2}+\frac{\omega}{2}\Big)\,\langle  \phi^{1/2}  \big(\Pi^0+\Pi^1\big)  \phi^{1/2}\rangle_{1+\eps} + \frac{1}{2} \langle \phi^{1/2}  \Pi^2 \phi^{1/2}\rangle_{1+\eps} + O(\eps)\\
&\leq \Big(\frac{1}{2}+\frac{\omega}{2}\Big)\,
 \langle \phi' \rangle_{1+\eps} + O(\eps)\;,
\end{align*}
where the first and last lines use the approximate linearity relations~\eqref{eq:approx-lin}, and the second line uses the definition of $K$ and $G\leq\Id$. This allows us to proceed as in the proof of~\cite[Theorem 5.8]{miller2017universal} to obtain 
$$ \langle \phi-\phi'\rangle_{1+\eps} \,\geq\, 2\Big( t- \frac{1}{2}-\frac{\omega}{2}\Big)\langle \phi \rangle_{1+\eps} - O(\eps)\;,$$
and conclude by applying~\cite[Proposition 4.4]{miller2017universal}.
\end{proof}

Using Lemma~\ref{lem:ms-uncertainty} we proceed to quantify the accumulation of randomness across multiple rounds of the simplified protocol, when it is executed with a simplified device that has overlap bounded away from $1$. 
The following proposition provides a measure of the randomness present in the transcript, conditioned on the verifier not aborting the protocol at the end, i.e. on $(g,c,t,o,k)\in\Acc_2$. (To see the connection with entropy, recall the definition of the $(1+\eps)$ conditional R\'enyi entropy in Definition~\ref{def:renyi}. The connection will be made precise in Section~\ref{sec:randomness}.)

\begin{proposition}\label{prop:d-rand}
Let $D = (\phi, \Pi, M,K)$ be a simplified device such that $\Delta(D)\leq \omega$ for some $\omega < 1$. Let $0<\eps \leq \frac{1}{2}$.  Let $\gamma,\eta,\kappa,q>0$ and $N$ an integer be parameters for an execution of Protocol 2 (Figure~\ref{fig:protocol2}) with $D$. Then
\begin{equation}\label{eq:d-rand-0}
- \frac{1}{\eps N} \log \Big( \frac{ \sum_{(g,c,t,o,k) \in \Acc_2} \,q(g,c)\kappa(t)\, \langle \phi^{ctok} \rangle_{1+\eps}}{\langle \phi \rangle_{1+\eps} }\Big) \,\geq\,  \lambda_\omega\Big(1-\frac{\gamma}{\kappa}-\eta\Big) - O\Big( q+\frac{\eps}{\kappa q}\Big)\;,
\end{equation}
where the states $\phi^{ctok}$ are introduced in Definition~\ref{def:post-meas-simplified}, $\lambda_\omega$ is the function defined in~\eqref{eq:def-lambda}, and $q(g,c)$ and $\kappa(t)$ are the distributions on $N$-bit strings $(g,c)$ and $t$ as selected by the verifier in Protocol 2.
\end{proposition}

\begin{proof}
The proof follows a similar argument as used in~\cite[Section 7]{miller2017universal}, and we outline the main steps. 

Let $t = \frac{\langle \phi_G \rangle_{1+\eps} }{\langle \phi \rangle_{1+\eps}}$ be as defined in Lemma~\ref{lem:ms-uncertainty} (this $t$ should not be confused with the string $t$ involved in the protocol description). Recall the notation for the post-measurement states introduced in~\eqref{eq:def-pm}. After one round of Protocol 2 is executed, the post-measurement state of the device can be decomposed into three components. First, in case $G_i=1$, which happens with probability $(1-q)$, the round is a generation round. The randomness generated in such a round is captured by the bound from Lemma~\ref{lem:ms-uncertainty},
\begin{equation}\label{eq:d-rand-1a}
(1-q)\big( \langle \phi_{1}^0 \rangle_{1+\eps} + \langle \phi_{1}^1 \rangle_{1+\eps} + \langle \phi_{1}^2 \rangle_{1+\eps}\big) \,\leq\, (1-q)\big(1-\ln(2)\eps \lambda_\omega(t) +O(\eps^2)\big)\langle \phi \rangle_{1+\eps}\;.
\end{equation}
The second case corresponds to $G_i=0$, which happens with probability $q$. In this case, for reasons that will become clear later in this proof we weigh the ``success'' and ``failure'' components of the post-measurement state differently. For the ``failure'' part we simply write
\begin{equation}\label{eq:d-rand-2}
\frac{q}{2}\big( (1-\kappa)\langle \phi_{00}^0 \rangle_{1+\eps} + \kappa \langle \phi_{01}^{00} \rangle_{1+\eps} + \kappa \langle \phi_{01}^{01} \rangle_{1+\eps}+\kappa \langle \phi_{01}^{11} \rangle_{1+\eps} +  \langle \phi_{1}^2 \rangle_{1+\eps}\big)\;.
\end{equation}
For the ``success'' part we add a weight of $2^{\frac{\eps s}{\kappa q}}$, where $s=O(1)$ is a real parameter to be determined later, to the cases where $T_i=1$:
\begin{align}
\frac{(1-\kappa)q}{2}&\big( \langle \phi_{00}^1 \rangle_{1+\eps} +  \langle \phi_{1}^0 \rangle_{1+\eps}+ \langle \phi_{1}^1 \rangle_{1+\eps}\big) + \frac{\kappa q}{2}2^{\frac{\eps s}{\kappa q}}\big( \langle \phi_{01}^{10} \rangle_{1+\eps} + \langle \phi_{1}^0 \rangle_{1+\eps}+ \langle \phi_{1}^1 \rangle_{1+\eps}\big)\notag\\
&\leq \frac{(1-\kappa)q}{2}\big( \langle \phi_{00}^1 \rangle_{1+\eps} +  \langle \phi_{1}^0 \rangle_{1+\eps}+ \langle \phi_{1}^1 \rangle_{1+\eps}\big) + \kappa q\Big(1+\ln(2)\frac{\eps s}{\kappa q} + O\Big(\frac{\eps^2}{\kappa^2 q^2}\Big)\Big) \,t \,\langle \phi\rangle_{1+\eps}\;,\label{eq:d-rand-3}
 \end{align}
where the inequality follows from the definition of $t$. Using the first inequality in~\eqref{eq:approx-lin} and regrouping terms, the sum of the left-hand sides of~\eqref{eq:d-rand-1a},~\eqref{eq:d-rand-2} and~\eqref{eq:d-rand-3} is at most
\begin{equation}\label{eq:d-rand-4}
\eqref{eq:d-rand-1a}~+~\eqref{eq:d-rand-2}~+~\eqref{eq:d-rand-3} \,\leq\,\Big( 1 - \eps\ln(2)\Big( \lambda_\omega(t)-st+O\Big( q+\frac{\eps}{\kappa q}\Big)\Big)\Big)\,\langle \phi\rangle_{1+\eps}\;.
\end{equation}
A convenient choice of $s$ is to take the derivative $s=\lambda_\omega'(r)$ for some $r\in[0,1]$ to be determined. With this choice, using that $\lambda_\omega$ is convex it follows that $\min_{t\in[0,1]} \lambda_\omega(t)- st = \lambda_\omega(r)-\lambda_\omega'(r)r$.
By chaining the inequality~\eqref{eq:d-rand-4} $N$ times, where at each step the density $\phi$ is updated with the one obtained from the previous round, and using that $\Acc_2$ contains those sequences $(g,c,t,o,k)$ such that the number of occurrences of $(c,t,o,k)\in\{(0,1,1,0),(1,*,0,*),(1,*,1,*)\}$ is at least $(1-\gamma/\kappa-\eta)\kappa qN$ we obtain
\begin{align*}
 - \frac{1}{\eps N} \log \Big( \frac{ \sum_{(g,c,t,o,k) \in \Acc_2} \,q(g,c)\kappa(t) \,\langle \phi^{ctok} \rangle_{1+\eps}}{\langle \phi \rangle_{1+\eps} }\Big) \,\geq\, (\lambda_\omega(r)-\lambda_\omega'(r)r)& + \big(1-\frac{\gamma}{\kappa}-\eta\big)\lambda_\omega'(r)\\
& -  O\Big( q+\frac{\eps}{\kappa q}\Big)\;,
\end{align*}
with  the term $(1-\frac{\gamma}{\kappa}-\eta)\lambda_\omega'(r)$ on the right-hand side  correcting for the weights $2^{\frac{\eps s}{\kappa q}}$ that would appear on the left-hand side with an exponent derived from the  acceptance criterion. Choosing $r=(1-\frac{\gamma}{\kappa}-\eta)$ completes the proof. 
\end{proof}

\subsection{Randomness accumulation in the general protocol}
\label{sec:randomness}

In this section we combine the results obtained in the previous two sections to analyze the randomness generated in Protocol~1.  
The main step is given in the following proposition. 

\begin{proposition}\label{prop:randomness}
Let $D=(\phi,\Pi,M)$ be an efficient device. Then for any $\eta>0$ and $q>0$ (that may be a function of $N$)  there is a choice of parameters $0<\kappa,\gamma<1$ for protocol~1 such that the following hold. 
Let $\ket{\phi}_{\reg{DE}}$ denote an arbitrary purification of $\phi_\reg{D}$, and $\ol{\rho}_{\reg{COE}}$ the joint state of the verifier's choice of challenges, the outputs computed by the verifier, and the adversary's system $\reg{E}$, restricted to transcripts that are accepted by the verifier in the protocol.\footnote{The state $\ol{\rho}$ is sub-normalized.} 
 Then there is a $\delta' = 2^{-\Omega(\kappa qN)}$ and a constant $C>0$ such that for any $N$ and $\delta$ (that may depend on $N$),
\begin{equation}\label{eq:ent-bound-1}
\frac{1}{N}\Hmin^{\delta + \delta'}(O|CE)_{\ol{\rho}} \,\geq\, \lambda_\omega\big(1-\kappa^{1/2}-\eta\big) - O\Big(q+H(\eta)^{1/2} + \frac{1+\log(2/\delta)}{H(\eta)^{1/2}\kappa qN}\Big)\;.
\end{equation}
\end{proposition}

\begin{proof}
Let $\tilde{D} = ({\phi},\tilde{\Pi},\tilde{M},K)$ be the elementary device obtained by applying Proposition~\ref{prop:change-d} to the device $D$, for a choice of $\omega=\frac{3}{4}$. Let $\tilde{\phi}= \phi^{\frac{1}{1+\eps}}$, where $\eps>0$ is a small parameter to be specified later. We apply Proposition~\ref{prop:d-rand} to $\tilde{D}$, with $\phi$ replaced by $\tilde{\phi}$. Then~\eqref{eq:d-rand-0} gives
\begin{equation}\label{eq:d-rand-1}
- \frac{1}{\eps N} \log \Big( \frac{ \sum_{(g,c,t,o,k) \in \Acc_2} \,q(g,c)\kappa(t)\, \langle \tilde{\phi}^{ctok} \rangle_{1+\eps}}{\langle \tilde{\phi} \rangle_{1+\eps} }\Big) \,\geq\,  \lambda_\omega\Big(1-\frac{\gamma}{\kappa}-\eta\Big) - O\Big( q+\frac{\eps}{\kappa q}\Big)\;.
\end{equation}
Next we apply Lemma~\ref{lem:mart-ub} to obtain
\begin{equation}\label{eq:d-rand-2b}
 - \frac{1}{\eps N} \log \Big(  \frac{ \sum_{(g,c,t,o):(g,c,o) \in \Acc} \,q(g,c)\kappa(t)\, \langle \tilde{\phi}^{cto} \rangle_{1+\eps}}{\langle \tilde{\phi} \rangle_{1+\eps} }\Big)\,\geq\,  \lambda_\omega\Big(1-\frac{\gamma}{\kappa}-\eta\Big) - O\Big(H(\eta)\kappa \frac{q}{\eps} + q+\frac{\eps}{\kappa q}\Big)\;,
\end{equation}
where the correction $H(\eta)\kappa \frac{q}{\eps} $ comes from the exponential prefactor in the bound from Lemma~\ref{lem:mart-ub}. The left-hand side of the bound in Lemma~\ref{lem:mart-ub} only considers those sequences such that $|t|\leq 2\kappa q N$, but adding those sequences back only incurs a negligible error $2^{-\Omega(\kappa q N)}$ (inside the logarithm), due to the Chernoff bound. 

We make one ultimate re-writing step. For any fixed $t$, the post-measurement state $\tilde{\phi}^{cto}$ can be expressed as 
$$P_N\cdots P_1\tilde{\phi} P_1 \cdots P_N\;,$$
 where $P_i$ is the measurement operator associated with challenge $c_i$ and outcome $o_i$. Using $\langle XX^*\rangle_{1+\eps} = \langle X^* X \rangle_{1+\eps}$ for any $X$, and recalling the definition of $\tilde{\phi} = \phi^{\frac{1}{1+\eps}}$, 
$$\langle P_N\cdots P_1\tilde{\phi} P_1 \cdots P_N \rangle_{1+\eps} \,=\, \langle \phi^{\frac{-\eps}{2(1+\eps)}} \phi^{\frac{1}{2}}P_1\cdots P_N^2 \cdots P_1\phi^{\frac{1}{2}}\phi^{\frac{-\eps}{2(1+\eps)}} \rangle_{1+\eps}\;.$$
Introduce a sub-normalized density 
$$\rho_\reg{E}^{cto}\,=\, \phi^{\frac{1}{2}}P_1\cdots P_N^2 \cdots P_1\phi^{\frac{1}{2}}\;,$$
that corresponds to the post-measurement state of register $\reg{E}$ (recall we assumed a purification $\ket{\phi}_{\reg{DE}}$ of $\phi$) at the end of Protocol $1$, for a given transcript $(c,o)$ for the interaction. 

We are in a position to apply Theorem~\ref{thm:ms}, with 
$$\rho_{\reg{CTOE}}^{o} = \sum_{(g,c,t):\,(g,c,o)\in \Acc} \,q(g,c)\kappa(t)\, \proj{c,t}_\reg{CT} \otimes \proj{o}_\reg{O}\otimes \rho_\reg{E}^{cto}\;,$$
and $\sigma_\reg{CTE} = \sum_{(g,c,t)} q(g,c)\kappa(t)\proj{c,t} \otimes \phi$. Applying the theorem and using~\eqref{eq:d-rand-2b} and $\langle\tilde{\phi}\rangle_{1+\eps} = 1$ by definition, we get that for any $\delta >0$,
\begin{equation}\label{eq:d-rand-3b}
\frac{1}{N} \Hmin^\delta(O|CTE)_{{\rho}} \,\geq \,  \lambda_\omega\big(1-\frac{\gamma}{\kappa}-\eta\big) - O\Big(H(\eta)\kappa \frac{q}{\eps} +  q+\frac{\eps}{\kappa q}\Big) - \frac{1+2\log(1/\delta)}{\eps N}\;.
\end{equation}
Using that the bound in~\eqref{eq:ent-bound-1} only considers registers $\reg{C}$ and $\reg{O}$ (the transcript) and $\reg{E}$, by Corollary~\ref{cor:mart} for any choice of $0<\eta<1$ there is a $\kappa_0>0$ such that for all $0\leq \kappa \leq \kappa_0$ and $\gamma = \kappa^{3/2}$, the bound~\eqref{eq:d-rand-3b} extends to a lower bound on the entropy $\Hmin^{\delta + \delta'}(O|CE)_{\ol{\rho}}$ at the cost of an additional $\delta' = O(2^{-\kappa q N})$ in the smoothing parameter.

Choose $\eta$ to be an arbitrarily small constant, set $\kappa = \gamma^{2/3}$ and $\gamma$ small enough so that $\kappa \leq \kappa_0$. Let $\eps$ be chosen as $ H(\eta)^{1/2} \kappa q$. With this choice of parameters, the term in the $O(\cdot)$ on the right-hand side of~\eqref{eq:d-rand-3b} is $O(q+H(\eta)^{1/2})$. 
\end{proof}

Making an appropriate choice of parameters for an execution of Protocol 1, Proposition~\ref{prop:randomness} gives our main result.

\begin{theorem}\label{thm:expansion}
Let $\mathcal{F}$ be an NTCF family and $\lambda$ a security parameter. Let $N$ be a polynomially bounded function of $\lambda$ such that $N = \Omega(\lambda^2)$. Set $q = \lambda/N$. Then there is a setting of $\eta,\gamma,\kappa$ and a $\delta = 2^{-\Omega( q N)}$ such that for any efficient prover, and side information $E$ correlated with the prover's initial state,
$$\Hmin^{N\delta}(O|CE)_{\ol{\rho}} \geq  \xi N\;,$$
where $\ol{\rho}$ is the final state of the output, challenge, and adversary registers, restricted to transcripts that are accepted by the verifier in the protocol and $\xi$ is a positive constant.\footnote{The constant $\xi$ is at least some positive universal constant of order $1/10$.}
\end{theorem}

Assume that an execution of $\Gen(1^\lambda)$ requires $O(\lambda^r)$ bits of randomness, for some constant $r$. (For example, for the case of our construction of a NTCF family based on LWE, we have $r=2$.) Then an execution of the protocol using the parameters in Theorem~\ref{thm:expansion} requires only $\poly(\lambda,\log N)$ bits of randomness for the verifier to generate the key $k$ and select the challenges. Taking $N$ to be slightly sub-exponential in $\lambda$, e.g. $N=2^{\sqrt{\lambda}}$, yields sub-exponential randomness expansion. 

\begin{proof}[Proof of Theorem~\ref{thm:expansion}]
Let $D$ be a device that is accepted with non-negligible probability in Protocol~1, where the parameters are a stated in the theorem. Applying Proposition~\ref{prop:randomness} to $D$ with a small enough choice of $\eta$ gives the result, with e.g.\ $\xi = \lambda_\omega/2$. 
\end{proof}

\bibliography{randomness,qpip}

\notesendofpaper

\end{document}